\documentclass{ecai} 
\usepackage{latexsym}
\usepackage{amssymb}
\usepackage{amsmath}
\usepackage{amsthm}
\usepackage{booktabs}
\usepackage{enumitem}
\usepackage{graphicx}
\usepackage{color}

\usepackage{mathtools}
\usepackage{natbib}
\usepackage{proof}
\usepackage{enumitem}
\usepackage{thm-restate}

\usepackage{tikz}
\usetikzlibrary {decorations.pathmorphing, decorations.pathreplacing, decorations.shapes, shapes.gates.logic.US}
\usetikzlibrary{shapes.geometric}
\usetikzlibrary{shapes.callouts}
\usetikzlibrary{calc}
\usetikzlibrary{trees}
\usetikzlibrary{backgrounds}
\usetikzlibrary{shapes.geometric}
\usetikzlibrary{positioning}
\usetikzlibrary{backgrounds}
\usetikzlibrary{arrows.meta}

\tikzset{
	small node/.style={circle, fill=black, minimum size=4pt, inner sep=0pt},
	small node transparent/.style={circle, fill=black!40, minimum size=4pt, inner sep=0pt},
	directed edge/.style={->, thick}, x=0.75cm, y=1.75cm,
	reversed edge/.style={<-, thick},
	next edge/.style={->, teal, dashed, thick},
	next edge transparent/.style={->, teal!40, dashed, thick},
	diamond edge/.style={->, purple, decorate,decoration={coil, segment length=4pt}},
	brace/.style={decorate,decoration={brace, segment length=4pt}}
}

\newtheorem{theorem}{Theorem}

\newtheorem{proposition}[theorem]{Proposition}

\newtheorem{definition}[theorem]{Definition}
\newtheorem{remark}[theorem]{Remark}
\newtheorem{example}[theorem]{Example}

\newtheorem{apptheorem}{Theorem}[section]
\newtheorem{applemma}[apptheorem]{Lemma}

\newtheorem{appcorollary}[apptheorem]{Corollary}
\newtheorem{appproposition}[apptheorem]{Proposition}

\newcommand{\BibTeX}{B\kern-.05em{\sc i\kern-.025em b}\kern-.08em\TeX}

\usepackage{xspace}
\usepackage{xcolor}

\newcommand{\mn}[1]{\ensuremath{\mathsf{#1}}}

\newcommand{\e}{\emph}

\newcommand{\A}{\ensuremath{\mathcal{A}}\xspace}

\newcommand{\T}{\ensuremath{\mathcal{T}}\xspace}

\newcommand{\p}{\varphi}

\newcommand{\mdl}{\models}

\newcommand{\sbs}{\subseteq}
\newcommand{\sqs}{\sqsubseteq}

\newcommand{\xto}[1]{\ensuremath{\xrightarrow{#1}}\xspace}

\newcommand{\ovl}[1]{\overline{#1}}

\newcommand{\F}{\ensuremath{\mathcal{F}}\xspace}
\newcommand{\G}{\ensuremath{\mathcal{G}}\xspace}

\newcommand{\Trig}{\ensuremath{\mathcal{T}_\textit{rig}}\xspace}
\newcommand{\Triglin}{\ensuremath{\mathcal{T}_\textit{rig}^\textit{\,lin}}\xspace}

\newcommand{\B}{\ensuremath{\mathcal{B}}\xspace}

\newcommand{\ld}{\ensuremath{\ .\ }}

\newcommand{\grto}{ \quad \to \quad }
\newcommand{\grand}{ \ \& \ }
\newcommand{\infers}{\ \ \ \vdash\ \ \ }

\newcommand{\der}{\textbf{der}\xspace}

\newcommand{\dom}{\mn{prem}}
\DeclareMathOperator{\rn}{\mn{conc}}

\newcommand{\role}[1]{\ensuremath{\rho^r_{#1}}\xspace}

\newcommand{\dlos}{\ensuremath{\mathit{Datalog}_{\textup{1S}}}\xspace}

\newcommand{\impd}{\leftarrow}

\newcommand{\whalecalf}{\e{Whale Calf}\xspace}

\newcommand{\proj}{\ensuremath{\mathrm{proj}}\xspace}

\newcommand{\bigO}[1]{{\mathcal{O}\left(#1\right)}}
\newcommand{\poly}[1]{{\mathrm{poly}\left(#1\right)}}

\newcommand{\ind}{\IN}

\newcommand{\IN}{\ensuremath{\mathsf{N_I}}\xspace}
\newcommand{\CN}{\ensuremath{\mathsf{N_C}}\xspace}

\newcommand{\RN}{\ensuremath{\mathsf{N_R}}\xspace}
\newcommand{\RNrig}{\ensuremath{\mathsf{N_R^{rig}}}\xspace}
\newcommand{\RNloc}{\ensuremath{\mathsf{N_R^{loc}}}\xspace}
\newcommand{\NN}{\ensuremath{\mathsf{N_N}}\xspace}

\newcommand\mydots{\makebox[1em][c]{.\hfil.\hfil.}}

\newcommand{\Circ}{\raisebox{0.25ex}{\text{\scriptsize{$\bigcirc$}}}}
\newcommand{\Next}{{\ensuremath{\Circ}}\xspace}
\newcommand{\Prev}{{\ensuremath{\Circ^{-}}}\xspace}

\newcommand{\Diam}{\raisebox{0.1ex}{\text{\small{$\Diamond$}}}}
\newcommand{\D}{\ensuremath{\Diam}\xspace}
\newcommand{\Df}{\ensuremath{\D}\xspace}
\newcommand{\Dp}{\ensuremath{\D\!^{-}}\xspace}

\newcommand{\monodic}{\ensuremath{\mathop{\ooalign{$\Box$ \cr \kern0.57ex \raisebox{0.2ex}{\scalebox{0.55}{$1$}}}\rule{0pt}{1.5ex} \kern-0.7ex}}\xspace}

\newcommand{\Nextone}{\ensuremath{\mathop{\ooalign{$\Next$ \cr \kern0.57ex \raisebox{0.3ex}{\scalebox{0.55}{$1$}}}\rule{0pt}{1.5ex} \kern-0.7ex}}\xspace}

\newcommand{\Wnextone}{\ensuremath{\mathop{\ooalign{$\Wnext$ \cr \kern0.57ex \raisebox{0.2ex}{\scalebox{0.55}{\textcolor{white}{$1$}}}}\rule{0pt}{1ex} \kern-0.7ex}}\xspace}

\newcommand{\Xnext}{^{\,\scalebox{0.5}{\smash{\raisebox{-3pt}{$\bigcirc$}}}}}

\newcommand{\Xall}{^{\,{\scalebox{0.5}{\smash{\raisebox{-3pt}{$\bigcirc$}}}} {\smash{\raisebox{-2.5pt}{$\diamond$}}}}}

\newcommand{\dlnd}{\ensuremath{\textsl{Datalog}\Xall}\xspace}

\newcommand{\NC}{\ensuremath{{\sf N_C}}\xspace}
\newcommand{\NI}{\ensuremath{{\sf N_I}}\xspace}
\newcommand{\NR}{\ensuremath{{\sf N_R}}\xspace}

\newcommand{\LTL}{\ensuremath{\textsl{L\!TL}}\xspace}

\newcommand{\EL}{\ensuremath{\Emc\Lmc}\xspace}
\newcommand{\TEL}{\ensuremath{\Tmc\Emc\Lmc}\xspace}
\newcommand{\TELn}{\ensuremath{\TEL^{\Xnext}}\xspace}

\newcommand{\TELnl}{\ensuremath{\TEL^{\Xnext}_{\textit{lin}}}\xspace}
\newcommand{\TELnc}{\ensuremath{\TELnl}\xspace}

\newcommand{\TELnf}{\ensuremath{\TEL^{\Xnext}_{\textit{future}}}\xspace}

\newcommand{\NLogSpace}{\textsc{NL}\xspace}
\newcommand{\PTime}{\textsc{P}\xspace}

\newcommand{\PSpace}{\textsc{PSpace}\xspace}

\newcommand{\ExpSpace}{\textsc{ExpSpace}\xspace}

\newcommand{\Jmf}{\ensuremath{\mathfrak{J}}\xspace}

\newcommand{\Qmf}{\ensuremath{\mathfrak{Q}}\xspace}

\newcommand{\Amc}{\ensuremath{\mathcal{A}}\xspace}
\newcommand{\Bmc}{\ensuremath{\mathcal{B}}\xspace}
\newcommand{\Cmc}{\ensuremath{\mathcal{C}}\xspace}
\newcommand{\Dmc}{\ensuremath{\mathcal{D}}\xspace}
\newcommand{\Emc}{\ensuremath{\mathcal{E}}\xspace}
\newcommand{\Fmc}{\ensuremath{\mathcal{F}}\xspace}
\newcommand{\Gmc}{\ensuremath{\mathcal{G}}\xspace}

\newcommand{\Imc}{\ensuremath{\mathcal{I}}\xspace}
\newcommand{\Jmc}{\ensuremath{\mathcal{J}}\xspace}

\newcommand{\Lmc}{\ensuremath{\mathcal{L}}\xspace}
\newcommand{\Mmc}{\ensuremath{\mathcal{M}}\xspace}
\newcommand{\Nmc}{\ensuremath{\mathcal{N}}\xspace}

\newcommand{\Smc}{\ensuremath{\mathcal{S}}\xspace}
\newcommand{\Tmc}{\ensuremath{\mathcal{T}}\xspace}

\newcommand{\Nbb}{\ensuremath{\mathbb{N}}\xspace}
\newcommand{\Zbb}{\ensuremath{\mathbb{Z}}\xspace}
\newcommand{\Z}{\Zbb}
\newcommand{\N}{\Nbb} 

\begin{document}

%%%%%%%%%%%%%%%%%%%%%%%%%%%%%%%%%%%%%%%%%%%%%%%%%%%%%%%%%%%%%%%%%%%%%%%%

\begin{frontmatter}

%%% Use this command to specify your submission number.
%%% In doubleblind mode, it will be printed on the first page.

\paperid{8583} 

%%% Use this command to specify the title of your paper.

\title{Analysing Temporal Reasoning in Description Logics\\ 
Using Formal Grammars}

%%% Use this combinations of commands to specify all authors of your 
%%% paper. Use \fnms{} and \snm{} to indicate everyone's first names 
%%% and surname. This will help the publisher with indexing the 
%%% proceedings. Please use a reasonable approximation in case your 
%%% name does not neatly split into "first names" and "surname".
%%% Specifying your ORCID digital identifier is optional. 
%%% Use the \thanks{} command to indicate one or more corresponding 
%%% authors and their email address(es). If so desired, you can specify
%%% author contributions using the \footnote{} command.

\author[A]{\fnms{Camille}~\snm{Bourgaux}%\thanks{Email: camille.bourgaux@ens.fr}
}
\author[B]{\fnms{Anton}~\snm{Gnatenko}\thanks{Email: anton.gnatenko@student.unibz.it}}
\author[A]{\fnms{Micha\"el}~\snm{Thomazo}%\thanks{Email: michael.thomazo@inria.fr}
} 

\address[A]{DI ENS, ENS, CNRS, PSL University \& Inria, Paris, France}
\address[B]{Free University of Bozen-Bolzano, Italy}

%%% Use this environment to include an abstract of your paper.

\begin{abstract}
	We establish a correspondence between (fragments of) \TELn, a temporal extension of the \EL description logic with the LTL operator $\Next^k$, and some specific kinds of formal grammars, in particular, conjunctive grammars (context-free grammars equipped with the operation of intersection). 
	This connection implies that \TELn does not possess the property of ultimate periodicity of models, and further leads to undecidability of query answering in \TELn, closing a question left open since the introduction of \TELn. Moreover, it also allows to establish decidability of query answering for some new interesting fragments of \TELn, and to reuse for this purpose existing tools and algorithms for conjunctive grammars.
\end{abstract}

\end{frontmatter}

%%%%%%%%%%%%%%%%%%%%%%%%%%%%%%%%%%%%%%%%%%%%%%%%%%%%%%%%%%%%%%%%%%%%%%%%

% !TEX root =  ../ms.tex

\section{Introduction}\label{sec:introduction}

\emph{Ontology-mediated query answering} (OMQA) aims at improving data access by enriching data with an ontology that expresses domain knowledge \cite{DBLP:journals/jods/PoggiLCGLR08,DBLP:conf/rweb/KontchakovZ14,DBLP:conf/rweb/BienvenuO15}. In this setting, an ontology is a set of logical formulas, typically expressed in a given description logic (DL) \cite{DBLP:conf/dlog/2003handbook} or via extensions of Datalog \cite{DBLP:books/aw/AbiteboulHV95,DBLP:journals/ai/BagetLMS11,DBLP:journals/ws/CaliGL12}. It provides a formalized vocabulary that allows a user to formulate queries in familiar terms, and to obtain more complete answers to queries, as answers are based not only on facts explicitly stored in the data (or ABox, in DL parlance) but also on facts that can be deduced through logical reasoning using the ontology (TBox).  
In the large literature on OMQA, special attention has been devoted to the so-called lightweight description logics, such as the DL-Lite family \cite{DBLP:journals/jair/ArtaleCKZ09} or the \EL family \cite{DBLP:conf/ijcai/BaaderBL05,DBLP:conf/owled/BaaderLB08}, which allow for tractable reasoning and underpin the OWL 2 QL and OWL 2 EL profiles of the Semantic Web standard ontology language \cite{profiles}. In particular, many large real-world ontologies, including the bio-medical ontology SNOMED CT% \cite{} %\footnote{A comprehensive terminology used in healthcare, www.snomed.org}
, use languages from the \EL family.

As many real-world applications require to query temporal data, various extensions of the OMQA framework have been proposed to integrate temporal modelling \cite{DBLP:conf/time/ArtaleKKRWZ17}. In this paper, we consider (fragments of) \TEL, a temporal extension of the DL language \EL introduced by \citet{Basulto-et-al:TEL}. In \TEL, the ABox facts are associated with \emph{timestamps}, and the TBox concept inclusions may feature some operators from \emph{linear temporal logic} (\LTL): $\Next$ (next), $\Next^-$ (previous), $\Diam$ (eventually) and $\Diam^-$ (eventually in the past). For instance, the concept inclusion $\mn{Prof} \sqs \Next \mn{Prof}$ intuitively means that at any moment, someone that is a professor is also a professor at the next instant. Moreover, \TEL allows the user to specify that some roles (binary predicates) are \emph{rigid}, i.e., that the relations they model do not change over time. 

\begin{example}\label{ex:running-example:init}
	Imagine that Alice is a professor in 2025, denoted $\mn{Prof}(\textup{Alice}, 2025)$. Professorship is permanent and requires advising students, who in three years become doctors. Being an advisor of a doctor makes one proud, and proud professors are happy. This knowledge is formalized as follows (using a rigid role \mn{advisorOf}):
	\begin{equation*}
		\begin{gathered}
			\ensuremath{
			\mn{Prof} \sqs \Next \mn{Prof} \quad
			\mn{Prof}\sqcap\mn{Proud} \sqs \mn{Happy} \quad 
			\mn{Student} \sqs \Next^3\mn{Dr}\\
			\mn{Prof} \sqs \exists\mn{advisorOf}.\mn{Student} \quad
			{\exists\mn{advisorOf}.\mn{Dr} \sqs \mn{Proud}}
			}
		\end{gathered}
	\end{equation*}
	Figure~\ref{fig:introduction:alice-example} provides a graphical representation of some information about Alice that can be inferred from $\mn{Prof}(\textup{Alice}, 2025)$ and the above \TEL TBox. In particular, Alice is happy at year $2028$. 
\end{example}

\begin{figure}[t]
	\centering
	
	% !TEX root =  ../ms.tex

\begin{tikzpicture}[node distance=1.8cm]
	
	\node (t0) {2025};
	\node[right of=t0] (t1) {2026};
	\node[right of=t1] (t2) {2027};
	\node[right of=t2] (t3) {2028};
	
	\node[small node, below of=t0, node distance=1cm, label=above:$\mn{Prof}$] (a0) {};

	\node[small node, below of=a0, node distance=1.4cm, label=below:{$\mn{Student}$}] (b1) {};

	\node[small node, right of=a0, label=above:{$\mn{Prof}$}] (a1) {};
	\node[small node, right of=a1, label=above:{$\mn{Prof}$}] (a2) {};
	\node[small node, right of=a2, label=above:{$\mn{Prof}, \mn{Proud}, \mn{Happy}$}] (a3) {};

	\draw[next edge] (a0) to (a1);
	\draw[next edge] (a1) to (a2);
	\draw[next edge] (a2) to (a3);
	
	\node[left of=a0, node distance=0.8cm] (a) {$\textup{Alice}$};
	
	\node[small node,right of=b1] (b2) {};
	\node[small node,right of=b2] (b3) {};
	\node[small node, right of=b3, label=below:{$\mn{Dr}$}] (b4) {};
	
	\draw[next edge] (b1) to (b4);
		
	\draw[directed edge] (a0) to node[pos=0.5, left] {$\mn{advisorOf}$} (b1);
	\draw[directed edge,dotted] (a3) to node[pos=0.5, right] {$\mn{advisorOf}$} (b4);
	\draw[directed edge,dotted] (a2) to (b3) {};
	\draw[directed edge,dotted] (a1) to (b2) {};
	
\end{tikzpicture}
		
	\caption{Representation of some inferences for Example~\ref{ex:running-example:init}. Dashed lines represent the temporal evolution of a given element, while dotted lines represent a relation  whose existence is known due to role rigidity.}
	\label{fig:introduction:alice-example}
\end{figure}
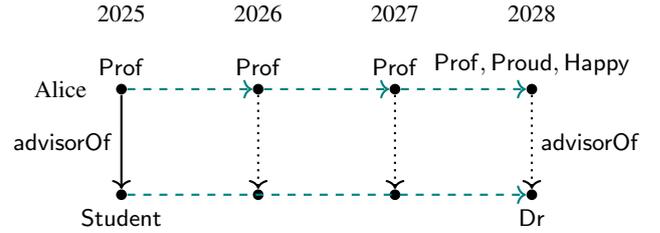

Integrating temporal reasoning in OMQA gave rise to a flourishing literature, with a large body of work on the theoretical side as well as some implementations  \cite{DBLP:journals/amcs/Kalayci0CRXZ19,DBLP:conf/aaai/WangHWG22,DBLP:conf/dlog/ThostHO15}, in a wide variety of settings. For example, an alternative way to model fact temporal validity is to use time intervals instead of timestamps \cite{DBLP:conf/rr/Gutierrez-BasultoK12,DBLP:journals/jair/BrandtKRXZ18,DBLP:journals/tocl/BaaderBKOT20}. 
On the ontology-mediated query side, it is also possible to use a standard, atemporal TBox and temporal queries, built from conjunctive queries and \LTL operators \cite{% DBLP:conf/rr/Gutierrez-BasultoK12 CB: moved this because about ABox with time intervals 
DBLP:conf/ijcai/BorgwardtT15,DBLP:journals/ws/BorgwardtLT15}. 
Finally, \TEL stems from a line of research which studies combinations of various DL languages and \LTL operators \cite{DBLP:conf/time/ArtaleKLWZ07,DBLP:conf/ijcai/ArtaleKWZ13,DBLP:journals/jair/ArtaleKKRWZ22}. 
For a more in-depth discussion of temporal reasoning within OMQA, we refer the reader to the survey by \citet{DBLP:conf/time/ArtaleKKRWZ17}.

Despite both \LTL and \EL being decidable, \citet{Basulto-et-al:TEL} showed that their combination in \TEL quickly leads to undecidability. However, they exhibited fragments of \TEL for which atomic query answering is decidable. Most of such fragments restrict the expressiveness of the temporal modelling, by only allowing 
operators \Next and \Prev, giving rise to the \TELn language. \citet{Basulto-et-al:TEL} left open the question of whether that restriction by itself is enough to regain decidability and proposed additional syntactic constraints, based on some form of acyclicity (either on the description logics side, or on the temporal side). All these constraints enforce a crucial property: the existence of models that are \emph{ultimately periodic}. In a nutshell, a model is ultimately periodic if the evolution over time of any given element is, after some initial segment, periodic. 

Our main contribution is to link temporal reasoning with \TELn-TBoxes to the study of associated formal languages, which allows us to close the open question of \citet{Basulto-et-al:TEL} and obtain additional results. We first consider the case  of \TELnf, where concept inclusions only allow to derive novel information about the future (or present) and not the past. For this fragment, we show in Section~\ref{sec:conjunctive-grammars} that the task of deciding whether a concept inclusion of the form $A\sqsubseteq \Next^n B$ is entailed by a TBox can be reduced to deciding whether $n$ belongs to the Parikh image \cite{Parikh:Theorem} of a unary conjunctive language, and vice-versa. Such languages are defined by \emph{conjunctive grammars}, introduced by \citet{Okhotin:Conjunctive-Grammars} as a generalization of context-free grammars which allows for a conjunction operation in rules. We then turn our attention to \TELnl, a temporal extension of linear \EL \cite{DBLP:conf/rr/DimartinoCPW16}, and obtain in Section~\ref{sec:linear} a similar reduction from reasoning in \TELnl  to deciding whether the Parikh image of a context-free grammar over a binary alphabet fulfils some property. 
In Section~\ref{sec:consequences}, we exploit the previous correspondences to obtain results for \TELn and fragments thereof. First, we close negatively the question of whether all \TELn-TBoxes enjoy ultimate periodicity, 
and complete the complexity picture of \TEL given by \citet{Basulto-et-al:TEL}: atomic query answering in \TELn is undecidable. Second, we provide results for the new fragments \TELnf and \TELnl we introduce. As for \TELnf, we prove that atomic query answering is solvable in polynomial time (in combined and data complexity), but becomes undecidable if \TELnf is extended with rigid concept names or the universal concept $\top$. Regarding \TELnl, we prove that \TELnl-TBoxes enjoy ultimate periodicity, and derive complexity results for atomic query answering.  Detailed proofs are provided in the appendix.

% !TEX root =  ../ms.tex

\section{Preliminaries}\label{sec:preliminaries}

In this section, we introduce \TELn and conjunctive grammars.

\subsection{The \TELn temporal description logic}

	We recall the syntax and semantics of \TELn from \citet{Basulto-et-al:TEL}.
	Let \IN, \CN, \RN be disjoint countably infinite sets of \emph{individual names}, \emph{concept names} and \emph{role names}, respectively, with \RN 
	partitioned into \emph{rigid role names} \RNrig and \emph{local role names} \RNloc. 
	
	\paragraph{Syntax}
	A \emph{fact} is an expression of the form  $A(a, n)$ or $r(a, b, n)$, where $a, b \in \IN$, $A \in \CN$, $r \in \RN$, and $n \in \Z$. A \emph{(temporal) ABox} (data instance) \A is a finite set of facts. A \TELn-\emph{TBox} (ontology) \T is a finite set of \emph{concept inclusions} of the form
	\begin{align}\label{math:ci:preliminaries:normal-form}
		&A \sqs \Next^n B
		&&A \sqcap A' \sqs B
		&&\exists r.A \sqs B
		&&A \sqs \exists r.B
	\end{align}
	where $A, A', B \in \CN$, $r\in\RN$, and $n \in \Z$. 
 When $n = 0$, $n = 1$, or $n=-1$, we simply write $A \sqs B$, $A \sqs \Next B$, and $A\sqs \Next^- B$, respectively.
	We will further consider two fragments of \TELn: the \emph{future} fragment, \TELnf, is obtained by setting $n \geqslant 0$, and the \emph{linear} fragment, \TELnl, disallows concept inclusions of the form $A \sqcap A' \sqs B$. 
	A \emph{(temporal) knowledge base (KB)} is a pair $(\T,\A)$. 
	Note that we consider \TELn-TBoxes \emph{in normal form}  \cite{Basulto-et-al:TEL}. We denote by $\CN(\T)$, $\RN(\T)$, $\RNrig(\T)$, and $\RNloc(\T)$, respectively, the sets of concept names, role names, and rigid and local role names appearing in \T.	 The \emph{size} $|\T|$ of \T (resp.~$|\A|$ of \A) is the number of symbols needed to write it down, with integers encoded in \emph{unary}.

	\paragraph{Semantics}
	An \emph{interpretation} $\Jmf$ is a structure $(\Delta^\Jmf, (\Imc_i)_{i\in\Z})$ where each $\Imc_i=(\Delta^\Jmf,\cdot^{\Imc_i})$ is a classical DL interpretation with domain $\Delta^\Jmf$:  
	for every $a\in\IN$, $a^{\Imc_i}=a$ (\emph{standard name assumption}, assuming that $\IN\subseteq\Delta^\Jmf$), 
	for every $A\in\CN$, $A^{\Imc_i}\subseteq \Delta^\Jmf$, and for every $r\in\RN$, $r^{\Imc_i}\subseteq \Delta^\Jmf\times\Delta^\Jmf$. 
		Moreover, for every $r \in \RNrig$, $r^{\Imc_i}=r^{\Imc_0}$ for every $i\in\Z$. 
	The interpretation function $\cdot^{\Imc_i}$ is often written as $\cdot^{\Jmf,i}$ and is extended to interpret complex concepts as expected: 
	\begin{align*}
	 (\Next^n A)^{\Jmf,i}=&A^{\Jmf,i+n} \quad\quad\quad
	 (A\sqcap B)^{\Jmf,i}=A^{\Jmf,i}\cap B^{\Jmf,i}\\
	 (\exists r.A)^{\Jmf,i}=&\{d\mid \exists e\in A^{\Jmf,i}, (d,e)\in r^{\Jmf,i}\}
	\end{align*}
The interpretation $\Jmf$ is a \emph{model} of a fact $A(a,n)$ (resp.~$r(a,b,n)$) if $a\in A^{\Jmf,n}$ (resp.~$(a,b)\in r^{\Jmf,n}$), and of a concept inclusion $C\sqsubseteq D$ if $C^{\Jmf,i}\subseteq D^{\Jmf,i}$ for every $i\in\Z$. It is a model of an ABox \A (resp.~a TBox \T) if it is a model of all facts in $\Amc$ (resp.~all concept inclusions in \T), and of a KB $(\T,\A)$ if it is a model of \T and \A. We write $\Jmf\models \alpha$ to denote that $\Jmf$ is a model of $\alpha$. A TBox $\Tmc$ \emph{entails} a concept inclusion $\alpha$, written $\Tmc\models\alpha$, if $\Jmf\models\Tmc$ implies $\Jmf\models\alpha$, and a KB $(\T,\A)$ entails a fact $\alpha$, $(\T,\A)\models\alpha$, if $\Jmf\models(\T,\A)$ implies $\Jmf\models\alpha$. 
Note that \emph{rigid concept names} can be  defined in a similar fashion as rigid role names, and can be simulated with $A\sqsubseteq\Next A$ and $A\sqsubseteq\Next^{-1} A$ in \TELn and \TELnl (but not in \TELnf).

\paragraph{Derivations}
	We will use the following derivation system for \TELn. Let $\NN$ be an infinite countable set of \emph{named nulls} (constants) disjoint from $\IN$. Given a \TELn-TBox \T and ABox \A, we write $(\T,\A)\vdash A(a,n)$ if there exists a derivation of $A(a,n)$ from $\A\cup\T$ using rules of the form \eqref{math:inference-rule:rigid}--\eqref{math:inference-rule:existential}. 
	Formally, such a derivation is a sequence $(\Fmc_0,\dots,\Fmc_m)$ 
	with $\Fmc_0=\A\cup\T$, $A(a,n)\in\Fmc_m$, and for $1\leqslant i\leqslant m$, $\Fmc_i$ is obtained from $\Fmc_{i-1}$ by choosing a rule such that all formulas in the left-hand side are in $\Fmc_{i-1}$ and adding %to $\Fmc_{i-1}$ 
	the formulas on the right-hand side.
	\begin{align}
		{r(a, b, n),\ r \in \RNrig,\ k \in \Z} &&\infers & r(a, b, k)\label{math:inference-rule:rigid} \\		
		{A(a, n),\ A \sqs \Next^k B} &&\infers &{B(a, n + k)} 
			\label{math:inference-rule:shift}\\
		{A(a, n),\ A'(a, n),\ A \sqcap A' \sqs B} &&\infers &{B(a, n)}
			\label{math:inference-rule:conjunction}\\
		{r(a, b, n),\ A(b, n),\ \exists r.A \sqs B} &&\infers &{B(a, n)}  
			\label{math:inference-rule:return}
			\\
		{A(a, n),\ A \sqs \exists r. B} &&\infers & r(a, b, n),\ B(b, n)\label{math:inference-rule:existential}
	\end{align}	
	where in \eqref{math:inference-rule:existential}, $b$ is a fresh element from $\NN$. The next proposition is easily shown using the canonical model of a \TELn KB, defined in a similar way as in the temporal DL-Lite case \cite{DBLP:conf/ijcai/ArtaleKWZ13,DBLP:conf/ijcai/ArtaleKKRWZ15}.
	
	\begin{proposition}\label{prop:derivationTEL}
	For every \TELn-TBox \T, ABox \A, $A,B\in\NC$, $a\in\IN$, and $n,k\in\Z$:
	\begin{itemize} 
		\item $(\T,\A)\models A(a,n)$ iff $(\T,\A)\vdash A(a,n)$, and
	\item $\T\models A\sqsubseteq\Next^n B$ iff $(\T,\{A(a,k)\})\vdash B(a,k+n)$.
	\end{itemize}
	\end{proposition}

	\begin{example}[Example~\ref{ex:running-example:init} cont'd]\label{ex:running-example:inference}
		 Below is a derivation of $\mn{Happy}(\textup{Alice}, 2028)$, from the TBox \T that contains the concept inclusions from Example~\ref{ex:running-example:init} and $\A=\{\mn{Prof}(\textup{Alice}, 2025)\}$, with $\mn{advisorOf} \in \RNrig$. The form of the rule applied is given as a subscript to $\vdash$ and the left-hand side is left implicit. 
		 \begin{align*}
		 		&\vdash_{\eqref{math:inference-rule:shift}} \mn{Prof}(\textup{Alice}, 2026)\\
				&\vdash_{\eqref{math:inference-rule:shift}} \mn{Prof}(\textup{Alice}, 2027)\\
				&\vdash_{\eqref{math:inference-rule:shift}} \mn{Prof}(\textup{Alice}, 2028)\\
		 		&\vdash_{\eqref{math:inference-rule:existential}} \mn{advisorOf}(\textup{Alice}, b, 2025), \mn{Student}(b, 2025)\\
		 		&\vdash_{\eqref{math:inference-rule:shift}} \mn{Dr}(b, 2028)\\
		 		&\vdash_{\eqref{math:inference-rule:rigid}} \mn{advisorOf}(\textup{Alice}, b, 2028)\\
		 		&\vdash_{\eqref{math:inference-rule:return}} \mn{Proud}(\textup{Alice}, 2028)\\
		 		&\vdash_{\eqref{math:inference-rule:conjunction}} \mn{Happy}(\textup{Alice}, 2028)
		 \end{align*}
		 By Proposition~\ref{prop:derivationTEL}, it holds that $(\T, \A) \mdl \mn{Happy}(\textup{Alice}, 2028)$ and $\T\models\mn{Prof}\sqsubseteq \Next^3 \mn{Happy}$.
	\end{example}
	
	\paragraph{Query answering and ultimately periodic TBoxes}
	The \emph{temporal atomic query answering} (TAQA) problem is that of deciding, given a temporal KB $(\T, \A)$ and a fact $A(a, n)$, whether $(\T, \A) \mdl A(a, n)$. 
	We consider \emph{combined complexity}, where the size of the input is $|\T|+|\A|+|A(a, n)|$, and \emph{data complexity}, where \T is fixed. %and the size of the input is $|\A|+|A(a, n)|$. 
	\citet{Basulto-et-al:TEL} show that decidability of TAQA is ensured by a property of the TBox, namely, \emph{ultimate periodicity}, which they define using the \emph{canonical quasimodel} of the TBox. 
	Since we do not use the notion of quasimodel in this work, we rephrase this property using the sets of numbers $\{n \in \Z \mid \T \mdl A \sqs \Next^n B\}$. 
	A set $\Lmc \sbs \Z^n$ is \emph{linear} if $\Lmc = \{\vec b + k_1\vec p_1 + \dots + k_l \vec p_l \mid k_1, \dots, k_l \in \N\}$ for some $\vec b \in \Z^n$, {called \e{offset}}, and $\vec p_1, \dots, \vec p_l \in \Z^n$, {called \e{periods}}. 
	A \emph{semilinear} set is a union of finitely many linear sets. A TBox \T is \emph{ultimately periodic} if for every $A, B \in \CN(\T)$, the set $\{n \in \Z \mid \T \mdl A \sqs \Next^n B\}$ is \emph{semilinear}. %The following theorem holds. 
	
	\begin{theorem}[\citet{Basulto-et-al:TEL}]\label{thm:basulto:ultimate-periodicity-means-pspace} 
		TAQA with ultimately periodic \TELn-TBoxes is in \PSpace for data complexity.
	\end{theorem}

	\subsection{Conjunctive grammars}	
	To analyse ultimate periodicity of general \TELn-TBoxes, we employ 	conjunctive grammars over a unary alphabet. 
	
	\paragraph{Syntax} A \emph{conjunctive grammar}, as introduced by \citet{Okhotin:Conjunctive-Grammars, Okhotin:Conjunctive-and-Boolean-Grammars-Survey}, is a quadruple $G = (N, \Sigma, \Smc, R)$, where $N$ and $\Sigma$ are disjoint alphabets of \emph{nonterminals} and \emph{terminals}, respectively, $\Smc \in N$ is a distinguished \emph{start symbol}, and $R$ is a finite set of \emph{rules} of the form:
	\begin{align}\label{math:grammar:preliminaries:definition}
		\Nmc \grto \alpha_1 \grand\dots \grand\alpha_n
	\end{align}
	with $\Nmc \in N$, $n\geqslant 1$, and $\alpha_i \in (N \cup \Sigma)^*$. Each $\alpha_i$ is called a \emph{conjunct}. If a grammar has a unique conjunct in every rule, then it is a \emph{context-free grammar}, and if further this conjunct has form either $\varepsilon$ or $c\,\Nmc'$, for $c \in \Sigma$, $\Nmc' \in N$, then it is a \emph{regular grammar} \cite{Chomsky:Grammars, Hopcroft-Ullman:Automata-Theory}. When the start symbol is not specified, we write $G = (N, \Sigma, R)$.
	Several rules with the same left-hand side \Nmc can also be written as a single rule (with $\mid$ used to separate the right-hand sides):
	\begin{align*}
		\Nmc \quad \to \quad \alpha^1_{1} \grand\dots \grand\alpha^1_{n_1} \mid \dots \mid \alpha^m_{1} \grand\dots \grand\alpha^m_{n_m}
	\end{align*}
	The \emph{size} $|G|$ of $G$ is the number of symbols needed to write it down. 
	
	\paragraph{Semantics}
	Intuitively, the semantics of conjunctive grammars extends that of context-free grammars with \emph{intersection}: given rule \eqref{math:grammar:preliminaries:definition}, apply ``in parallel'' context-free rules $\Nmc \to \alpha_i$ and take the intersection of the generated languages. 
	Formally, \emph{derivations} for grammars are defined in a similar way as derivations for knowledge bases. 
	Given $G$, let $\{X(w) \mid X \in N \cup \Sigma, w \in \Sigma^*\}$ be a set of \emph{propositions}, each meaning “a word $w$ has a property $X$”. The \emph{axioms} are
		\begin{align}\label{math:grammar:semantics:axioms}
			c(c) \quad (\text{for every } c \in \Sigma)
		\end{align}
		and the derivation rules are obtained as follows. 
		For every rule of form \eqref{math:grammar:preliminaries:definition} in $R$, each $\alpha_i$ is of the form $X^i_{1} \dots X^i_{k_i}$ with $k_i\geqslant 0$ and $X^i_j\in N \cup \Sigma$. For all words $u^i_{j}\in\Sigma^*$ with $1\leqslant i\leqslant n$ and $1\leqslant j\leqslant k_i$  such that $u^1_{1}\dots u^1_{k_1}=\cdots=u^n_{1}\dots u^n_{k_n}=w$, we have the rule:
		\begin{align}\label{math:grammar:semantics:deduction-rules}
				X^1_{1}(u^1_{1}), \mydots, X^1_{k_1}(u^1_{k_1}),\mydots, X^n_{1}(u^n_{1}),\mydots, X^n_{k_n}(u^n_{k_n})\ \vdash\  \Nmc(w)
		\end{align}
	Then we write $G \vdash X(w)$ whenever $X(w)$ can be derived from the axioms using the rules. 
	The \emph{language} of $X\in N\cup\Sigma$ is $L_G(X) = \{w \in \Sigma^* \mid G \vdash X(w)\}$, and the language of the grammar $G$ is $L(G) = L_G(\Smc)$.
	We refer to the survey by \citet{Okhotin:Conjunctive-and-Boolean-Grammars-Survey} for discussion of alternative equivalent definitions of the semantics. 
	
	\begin{example}[\citet{Okhotin:Conjunctive-and-Boolean-Grammars-Survey}]\label{ex:preliminaries:conjuntive grammar}
		The language $\{a^nb^nc^n \mid n \in \N\}$ is generated by $G=(\{\Smc,\Amc,\Bmc,\Cmc,\Dmc\},\{a,b,c\},\Smc,R)$ where $R$ contains:
		\begin{align*}
			&\Smc \grto \Amc\Bmc\, \&\, \Dmc\Cmc\\
			&\Amc \grto a\Amc\, \mid \varepsilon
			&\Bmc \grto b\Bmc c \mid \varepsilon\\
			&\Cmc \grto c\,\Cmc \mid \varepsilon
			&\Dmc \grto a\Dmc b \mid \varepsilon
		\end{align*}
	\end{example}
		
	We call a language $L \sbs \Sigma^*$ \emph{conjunctive}, (\emph{context-free}, \emph{regular}) if $L = L(G)$ for a conjunctive (respectively, context-free or regular) grammar $G$. A language (or a grammar) is called \emph{unary} when the underlying alphabet contains just one symbol, i.e. $\Sigma = \{c\}$. The membership problem for conjunctive grammars is \PTime-complete.
	\begin{theorem}[Okhotin \citep{Okhotin:Conjunctive-Grammar-Parsing:Uniform, Okhotin:Conjunctive-Grammar-Parsing:Fixed}]\label{thm:membership-testing}
		Checking whether $w \in L(G)$, for a given conjunctive grammar $G$ and $w \in \Sigma^*$, is \PTime-complete.
	\end{theorem}

	\paragraph{Parikh images and expressiveness} 
	Given an alphabet $\Sigma = \{c_1, \dots, c_n\}$ enumerated in a fixed order, let $\#c_i(w)$ denote the number of occurrences of $c_i$ in $w \in \Sigma^*$. The \emph{Parikh image} $p(w)$ of a word $w$ is a vector $\vec u \in \Z^n$ such that $u_i = \#c_i(w)$, for all $1 \leqslant i \leqslant n$. The Parikh image $p(L)$ of a language $L$ is the set $\{p(w) \mid w \in L\}$. It is easy to see that when $L$ is regular, $p(L)$ is semilinear. A deeper result is the following. 
	\begin{theorem}[\citet{Parikh:Theorem}]\label{thm:parikh}
		If $L$ is context-free, then $p(L)$ is semilinear. Moreover,\ for any semilinear $S\sbs \N^n$ there exists a regular language $L'$ such that $S = p(L')$.
	\end{theorem}
	
	A unary language $L$ can be seen as a set of numbers given in unary that coincides with its Parikh image: $c^n \in L$ if and only if $n \in p(L)$.
	Theorem~\ref{thm:parikh} implies that a unary language is regular iff its Parikh image is semilinear, and thus unary context-free languages are regular. However, unary conjunctive languages may not be regular. 
	
	\begin{example}[\citet{Jez:Nonregular-Unary-Conjunctive-Language}]\label{ex:preliminaries:jez} 
For the following grammar $G$, the language $L_G(\Nmc_1) = \{c^{4^n} \mid n \in \N\}$ is not regular:
		\begin{align*}
			&\Nmc_1 \grto \Nmc_1\Nmc_3 \grand \Nmc_2\Nmc_2 \mid c 
			\\
			&\Nmc_2 \grto \Nmc_1\Nmc_1 \grand \Nmc_2\Nmc_{6} \mid cc\\
			&\Nmc_3 \grto \Nmc_1\Nmc_2 \grand \Nmc_{6}\Nmc_{6} \mid ccc\\
			&\Nmc_{6} \grto \Nmc_1\Nmc_2 \grand \Nmc_3\Nmc_3
		\end{align*}
		To understand the rules above, associate every word $c^n$ with the number $n$, and each $\Nmc_i$ with the set $\{i \cdot 4^n \mid n \in \N\}$. Since $c^nc^k = c^{n + k}$, concatenation of words corresponds to the summation of the respective numbers. The expression $\Nmc_1\Nmc_3 \,\&\, \Nmc_2\Nmc_2$ encodes the equation $4^m + 3\cdot 4^k = 2\cdot 4^l + 2\cdot 4^s$, which holds if and only if $m = k = l = s$, when both sides become equal to $4^{k + 1}$. This newly obtained number is assigned, by the first rule, to the set of $\Nmc_1$.
	\end{example}

	Building on the idea behind Example~\ref{ex:preliminaries:jez}, \citet{Jez-Okhotin:Unary-Conjunctive-Grammars-Undecidability} devised grammars encoding Turing machine computations, leading to:
	% the following result.
	
	\begin{theorem}[\citet{Jez-Okhotin:Unary-Conjunctive-Grammars-Undecidability}]\label{thm:jez-okhotin:undecidability}
		Given a unary conjunctive grammar $G$, it is undecidable (and co-r.e.-complete) whether $L(G)$ is (i)~empty, (ii)~finite, or (iii)~regular.
	\end{theorem}
	
% !TEX root =  ../ms.tex

\section{Future~\TELn~and~unary~conjunctive~grammars}\label{sec:conjunctive-grammars} 

In this section, we prove two theorems that establish our key result: 
$\{\{n \in \Nbb \mid \T \mdl A \sqs \Next^n B\}\mid \Tmc\ \TELnf\text{-TBox }, A,B\in\NC\}$ 
is the set of Parikh images of unary conjunctive languages. This will allow us to apply Theorems~\ref{thm:parikh} and \ref{thm:jez-okhotin:undecidability} to analyse ultimate periodicity of \TELn-TBoxes.

\begin{restatable}[TBoxes to Grammars]{theorem}{TBoxesToGrammars}\label{thm:conjunctive-grammars:tbox-to-grammar}
	For every \TELnf-TBox \T, one can construct in polynomial time a unary conjunctive grammar $G_\T = (N, \{c\}, R)$ such that for any $A, B \in \CN(\T)$, there is $\Nmc_{AB} \in N$ such that $c^n \in L_{G_\T}(\Nmc_{AB})$ iff $\T \mdl A \sqs \Next^n B$.
\end{restatable}

Given a \TELnf-TBox \T, we sketch the construction of $G_\T$. The first step is to ensure that every role name in \T can be treated as rigid. For each $C \in \CN(\T),\ r \in \RNloc(\T)$, introduce a pair of fresh concept names $C_r, C'_r$, and let \Trig be obtained from \T as follows:
	\begin{enumerate}
		\item for each $r \in \RNloc(\T)$, substitute every $A \sqs \exists r . B \in \T$ with 
		\begin{align}\label{math:ci:rigidisation:existential}
			&A \sqs \exists r . B_r
			&&B_r \sqs B
		\end{align}
		
		\item for each $C_r$, substitute every $\exists r . A \sqs B \in \T$ with 
		\begin{align}\label{math:ci:rigidisation:return}
			&A \sqcap C_r \sqs A'_r
			&&\exists r . A'_r \sqs B
		\end{align}
		
		\item substitute each $r \in \RNloc(\T)$ with a fresh $r' \in \RNrig$.
	\end{enumerate}
Intuitively, in a derivation using \Trig, a fact $B_r(b, n)$ created from some $A(a,n)$ and $A \sqs \exists r . B_r$ guards the ``locality'' of $r(a, b, n)$ at time~$n$. Any application of a derivation rule of form \eqref{math:inference-rule:return} using a fact of the form $r(a, b, k)$, and hence any effect of $b$ on $a$, is only possible using some $\exists r.A'_r\sqsubseteq C$ and $A'_r(b,k)$, and thus is limited to $k=n$, since $A'_r(b,k)$ can only be derived using $A\sqcap B_r\sqsubseteq A'_r$ and $B_r(b, n)$. 
The translation from \T to \Trig is polynomial, and it is not hard to prove the following lemma.
\begin{restatable}{lemma}{conjunctivegrammarslocalroleremoval}\label{lm:conjunctive-grammars:local-role-removal}
	Let $\T$ be a \TELn-TBox. For any $A, B \in \CN(\T)$ and $n\in\Z$, $\T \mdl A \sqs \Next^n B$ if and only if $\Trig \mdl A \sqs \Next^n B$.
\end{restatable}

\begin{definition}\label{def:conjuntive-grammars:GT}
	Given a \TELnf-TBox \T, $G_\T = (N_\T, \{c\}, R_\T)$, where $N_\T = \{\Nmc_{AB} \mid A, B \in \CN(\Trig)\}$ and $R_\T$ contains exactly:% the following rules.
	\begin{align}
		&\Nmc_{AB} \grto \varepsilon, &&\text{for } A \sqs B \in \Trig \text{ or } A = B
		\label{math:grammar:from-tbox:epsilon}\\
		&\Nmc_{AB} \grto c^n, &&\text{for } A \sqs \Next^n B \in \Trig,\ n > 0
		\label{math:grammar:from-tbox:shift}\\
		&\Nmc_{AB} \grto \Nmc_{AC} \grand \Nmc_{AD}, &&\text{for } \begin{array}{l}
			A \in \NC(\Trig),\\
			C \sqcap D \sqs B \in \Trig
		\end{array}
		\label{math:grammar:from-tbox:conjunction}\\
		&\Nmc_{AB} \grto \Nmc_{CD}, &&\text{for } \left\{\begin{array}{l}
			A \sqs \exists r . C\\
			\exists r . D \sqs B 
		\end{array}\right\} \sbs \Trig
		\label{math:grammar:from-tbox:down-and-up}\\
		&\Nmc_{AB} \grto \Nmc_{AC}\,\Nmc_{CB}, &&\text{for } A,B,C \in \CN(\Trig)
		\label{math:grammar:from-tbox:middle}
	\end{align}
\end{definition}

	Intuitively, for every pair of concept names $A,B\in\NC(\Trig)$, $G_\T$ encodes every possible way of deriving $B(a, n)$ from $\{A(a, 0)\}\cup\Trig$: either directly~(\ref{math:grammar:from-tbox:epsilon}, \ref{math:grammar:from-tbox:shift}), or by obtaining $C(a, n)$ and $D(a, n)$ that together give $B(a, n)$~\eqref{math:grammar:from-tbox:conjunction}, or by going through a null~\eqref{math:grammar:from-tbox:down-and-up}, or through an intermediate point $C(a, k)$, $0 \leqslant k \leqslant n$ \eqref{math:grammar:from-tbox:middle}.
One can show that a derivation witnessing $(\T, \{A(a, 0)\}) \vdash B(a, n)$ using rules of the form \eqref{math:inference-rule:rigid}--\eqref{math:inference-rule:existential} corresponds to a derivation for $G_\T \vdash \Nmc_{AB}(c^n)$ from axiom \eqref{math:grammar:semantics:axioms} using rules \eqref{math:grammar:semantics:deduction-rules}, and vice versa. 
Theorem~\ref{thm:conjunctive-grammars:tbox-to-grammar} follows by 
Lemma~\ref{lm:conjunctive-grammars:local-role-removal}, 
Proposition~\ref{prop:derivationTEL} and definition of $L_{G_\T}(\Nmc_{AB})$.

\begin{example}[Ex.~\ref{ex:running-example:inference} cont'd]\label{ex:conjunctive-grammars:tbox-to-grammar}
	Recall that $\T \mdl \mn{Prof} \sqs \Next^3 \mn{Happy}$. Figure \ref{fig:conjunctive-grammars:tbox-to-grammar} illustrates the 
	derivations witnessing $G_\T \vdash \Nmc_{\mn{Prof} \mn{Prof}}(c^3)$ and 				
	$G_\T \vdash \Nmc_{\mn{Prof} \mn{Proud}}(c^3)$. 
	We obtain $G_\T \vdash\Nmc_{\mn{Prof} \mn{Happy}}(c^3)$ using the rule $\Nmc_{\mn{Prof} \mn{Happy}}\rightarrow \Nmc_{\mn{Prof} \mn{Prof}}\, \&\, \Nmc_{\mn{Prof} \mn{Proud}}$. 
	One can further check that $ L_{G_\T}(\Nmc_{\mn{Prof} \mn{Happy}})=\{c^{3+n}\mid n\in\N\}$, since $\Nmc_{\mn{Prof} \mn{Happy}}\rightarrow \Nmc_{\mn{Prof} \mn{Prof}}\Nmc_{\mn{Prof} \mn{Happy}}$ is in $R_\T$ and $G_\T \vdash \Nmc_{\mn{Prof} \mn{Prof}}(c^n)$ for every $n$.
\end{example}

\begin{figure}[t]
	\centering
	
	% !TEX root =  ../ms.tex

\begin{tikzpicture}[node distance=1.35cm]
		
	\node[small node transparent] (t0) {};
	
	\node[right of=t0] (c01) {{\Large $c$}};
	
	\node[small node transparent, right of=c01] (t1) {};
	
	\node[right of=t1] (c12) {{\Large $c$}};
	
	\node[small node transparent, right of=c12] (t2) {};
	
	\node[right of=t2] (c23) {{\Large $c$}};
	
	\node[small node transparent, right of=c23] (t3) {};
	
	\draw[next edge transparent] (t0) to (c01) to (t1);
	\draw[next edge transparent] (t1) to (c12) to (t2);
	\draw[next edge transparent] (t2) to (c23) to (t3);
	
	\draw ($(t0) + (0, -0.1)$) to ($(t0) + (0, -0.2)$) to ($(t3) + (0, -0.2)$) to  ($(t3) + (0, -0.1)$);
	
	\node (N-Student-Dr) at ($(t1)!0.5!(t2) + (0, -0.35)$) {$R_3 \colon \Nmc_{\mn{Student} \mn{Dr}}(ccc)$};
	
	\draw ($(t0) + (0, -0.3)$) to ($(t0) + (0, -0.6)$) to ($(t3) + (0, -0.6)$) to  ($(t3) + (0, -0.3)$);
	
	\node (N-Prof-Proud) at ($(t1)!0.5!(t2) + (0, -0.75)$) {$R_4 \colon \Nmc_{\mn{Prof} \mn{Proud}}(ccc)$};
	
	\draw ($(t0) + (0, 0.1)$) to ($(t0) + (0, 0.2)$) to ($(t1) + (-0.05, 0.2)$) to  ($(t1) + (-0.05, 0.1)$);
	\draw ($(t1) + (0.05, 0.1)$) to ($(t1) + (0.05, 0.2)$) to ($(t2) + (-0.05, 0.2)$) to  ($(t2) + (-0.05, 0.1)$);
	
	\node (N-Prof-Prof-01) at ($(t0)!0.5!(t1) + (0, +0.35)$) {$R_1 \colon \Nmc_{\mn{Prof} \mn{Prof}}(c)$};
	\node (N-Prof-Prof-12) at ($(t1)!0.5!(t2) + (0, +0.35)$) {$R_1 \colon \Nmc_{\mn{Prof} \mn{Prof}}(c)$};
	
	\draw ($(t0) + (0, 0.3)$) to ($(t0) + (0, 0.6)$) to ($(t2) + (-0.05, 0.6)$) to  ($(t2) + (-0.05, 0.3)$);
	\draw ($(t2) + (0.05, 0.1)$) to ($(t2) + (0.05, 0.6)$) to ($(t3) + (0, 0.6)$) to  ($(t3) + (0, 0.1)$);
	
	\node (N-Prof-Prof-02) at ($(t0)!0.5!(t2) + (0, +0.75)$) {$R_2 \colon \Nmc_{\mn{Prof} \mn{Prof}}(cc)$};
	\node (N-Prof-Prof-23) at ($(t2)!0.5!(t3) + (0, +0.75)$) {$R_1 \colon \Nmc_{\mn{Prof} \mn{Prof}}(c)$};
	
	\draw ($(t0) + (0, 0.7)$) to ($(t0) + (0, 1)$) to ($(t3) + (0, 1)$) to  ($(t3) + (0, 0.7)$);
	
	\node (N-Prof-Prof03) at ($(t0)!0.5!(t3) + (0, 1.15)$) {$R_2 \colon \Nmc_{\mn{Prof} \mn{Prof}}(ccc)$};
\end{tikzpicture}
	
	\caption{
		Derivations witnessing $G_\T \vdash \Nmc_{\mn{Prof} \mn{Prof}}(c^3)$, in the upper part, and 				
		$G_\T \vdash \Nmc_{\mn{Prof} \mn{Proud}}(c^3)$, in the lower part. 
		Intuitively, each symbol $c$ stands for a step forward in time (cf. Figure \ref{fig:introduction:alice-example}).
		The grammar rule of $G_\T$ used to obtain each proposition through \eqref{math:grammar:semantics:deduction-rules} is given in the following denotation: 	
		$R_1 \colon \Nmc_{\mn{Prof} \mn{Prof}}\rightarrow c$, 
		$\quad R_2 \colon \Nmc_{\mn{Prof} \mn{Prof}}\rightarrow \Nmc_{\mn{Prof} \mn{Prof}}\Nmc_{\mn{Prof} \mn{Prof}}$,  
		$\quad R_3 \colon \Nmc_{\mn{Student} \mn{Dr}}\rightarrow c^3$, 
		$\quad R_4 \colon \Nmc_{\mn{Prof} \mn{Proud}}\rightarrow \Nmc_{\mn{Student} \mn{Dr}}$.
	}
	
	\label{fig:conjunctive-grammars:tbox-to-grammar}
\end{figure}

We now turn our attention to the converse translation.

\begin{restatable}[Grammars to TBoxes]{theorem}{GrammarsToTBoxes}\label{thm:conjunctive-grammars:grammar-to-tbox}
	For every unary conjunctive grammar $G = (N, \{c\}, R)$, one can construct in polynomial time a \TELnf-TBox $\T_G$ and $A \in \CN(\T_G)$, such that for every $\B \in N$ there is $B \in \CN(\T_G)$ such that $\T_G \mdl A \sqs \Next^n B$ iff $c^n \in L_G(\Bmc)$.
\end{restatable}

Let $G$ be a unary conjunctive grammar with nonterminals $N = \{\B_1, \dots, \B_m\}$. 
W.l.o.g., we assume that its rules are of the forms 
\begin{align}
	&\B_i \grto \varepsilon 
		\label{math:grammar:to-tbox:epsilon}\\
	&\B_i \grto c^n, \quad n > 0
		\label{math:grammar:to-tbox:terminal}\\
	&\B_i \grto \alpha_1
		\label{math:grammar:to-tbox:one}\\
	&\B_i \grto \alpha_1 \grand \alpha_2
	\label{math:grammar:to-tbox:two}
\end{align}
where $\alpha_1, \alpha_2$ are nonempty strings of nonterminals. Indeed, every unary grammar can be converted to this form in polynomial time. 

Fix concept names $A, B_1, \dots, B_m$, and, for each $\alpha_l = \B_{i_1}\dots \B_{i_{k}}$ that occurs in the rules of $G$, introduce concept names $C_{i_j \dots i_{k}}$ and rigid role names $r_{i_j \dots i_{k}}$, for $1 \leqslant j < {k}$. Let~\Jmc denote the set of number sequences $i_j \dots i_{k}$ appearing in the subscripts of these symbols. We use the symbol $\iota$ to denote the elements of~\Jmc, and write $i\iota$ to mean the sequence obtained from $\iota$ by appending $i$ in the beginning. Moreover, let $\iota(\alpha_l)$ denote exactly the sequence $i_1 \dots i_{k}$. 

\begin{definition}\label{def:grammar-to-tbox}
	Given $G = (N, \{c\}, R)$, with $N = \{\B_1, \dots, \B_m\}$ and rules of the forms \eqref{math:grammar:to-tbox:epsilon}--\eqref{math:grammar:to-tbox:two},	$\T_G$ contains exactly the following concept inclusions.
	\begin{align}
		&A \sqs B_i, 
		&&\text{for each rule of the form \eqref{math:grammar:to-tbox:epsilon}} 
			\label{math:ci:from-grammar:now} 
			\tag{\ref{math:grammar:to-tbox:epsilon}$^*$}\\
		&A \sqs \Next^n B_i,  
		&&\text{for each rule of the form \eqref{math:grammar:to-tbox:terminal}} 
			\label{math:ci:from-grammar:shift} 
			\tag{\ref{math:grammar:to-tbox:terminal}$^*$}\\
		&C_{\iota(\alpha_1)} \sqs B_i, 
		&&\text{for each rule of the form \eqref{math:grammar:to-tbox:one}} 
			\label{math:ci:from-grammar:imply} 
			\tag{\ref{math:grammar:to-tbox:one}$^*$}\\
		&C_{\iota(\alpha_1)} \sqcap C_{\iota(\alpha_2)} \sqs B_i, 
		&&\text{for each rule of the form \eqref{math:grammar:to-tbox:two}} 
			\label{math:ci:from-grammar:meet} 
			\tag{\ref{math:grammar:to-tbox:two}$^*$}\\
		&B_{i} \sqs \exists r_{i\iota} \ld A,
		&&\text{for } i\iota \in \Jmc 
			\label{math:ci:from-grammar:down} 
			\tag{\number\numexpr\value{equation}+1\relax$^*$}\\
		&\exists r_{i \iota} \ld C_{\iota} \sqs C_{i \iota}
		&&\text{for } \iota,\ i \iota \in \Jmc 
			\label{math:ci:from-grammar:up} 
			\tag{\number\numexpr\value{equation}+1\relax$^*$}
			\refstepcounter{equation}\\
		&B_{i} \sqs C_i
%		&&\text{for } C_i \text{ with a subscript}
%			\label{math:ci:from-grammar:conclude} 
%			\tag{\number\numexpr\value{equation}+1\relax$^*$}
%			\refstepcounter{equation}\\ 
%		& && \quad\ \,\text{consisting of } i \text{ only} \notag
		&&\text{for } i\in\{1,\dots,m\} \text{ s.~t. }i\in\Jmc
			\label{math:ci:from-grammar:conclude} 
			\tag{\number\numexpr\value{equation}+1\relax$^*$}
			\refstepcounter{equation}
\end{align}
\refstepcounter{equation}
\end{definition}\label{def:conjunctive-grammars:TG}

We show that $c^n \in L_G(\B_i)$ iff $\T \mdl A \sqs \Next^n B_i$. We first illustrate this on an example.

\begin{figure}[t]
		
		\centering

		% !TEX root =  ../ms.tex

\scalebox{0.95}{
\begin{tikzpicture}[node distance=1.5cm]
	
	\node (t0) {0};
	\node[right of=t0] (t1) {1};
	\node[right of=t1] (t2) {2};
	\node[right of=t2] (t3) {3};
	\node[right of=t3] (t4) {4};
	
	\node[small node, below of=t0, node distance=1cm, label=below:$A$] (a0) {};
	\node[small node, right of=a0, label=right:{$B_1^{\,\eqref{math:ex:grammar-to-tbox:one}}$}] (a1) {};
	\node[small node, right of=a1, label=right:{$B_2^{\,\eqref{math:ex:grammar-to-tbox:two}}$}] (a2) {};
	\node[right of=a2] (a3) {};
	\node[small node, right of=a3, label=above:{$C_{13}^{\,(\ref{math:ex:grammar-to-tbox:down-and-up:first}, iii)}, C_{22}^{\,(\ref{math:ex:grammar-to-tbox:down-and-up:second}, iii)}, B_1^{\,\eqref{math:ex:grammar:-to-tbox:conjunction}}$}] (a4) {};
	
	\draw[next edge] (a0) to (a1);
	
	\node[left of=a0, node distance=0.5cm] (a) {$a$};
	\node[below of=a, node distance=1cm] (b) {};
	\node[below of=b, node distance=1cm] (c) {};
	
	\node[small node, below of=a1, node distance=1cm, label=below:{$A^{(\ref{math:ex:grammar-to-tbox:down-and-up:first}, i)}$}] (b1) {};
	\node[right of=b1] (b2) {};
	\node[right of=b2] (b3) {};
	\node[small node, right of=b3, label=below:{\hspace{-5pt}$B_{3}^{\,\eqref{math:ex:grammar-to-tbox:three}}, C_3^{\,(\ref{math:ex:grammar-to-tbox:down-and-up:first}, ii)}$}] (b4) {};
	
	\draw[next edge] (b1) to (b4);
	
	\draw[directed edge] (a1) to node[pos=0.5, left] {$r_{13}^{(\ref{math:ex:grammar-to-tbox:down-and-up:first}, i)}$} (b1);
	
	\draw[directed edge, dotted] (a4) to node[pos=0.4, left]  {$r_{13}$} (b4);
	
	\draw[next edge] (a0) to[out=30, in=150] (a2);
	
	\node[small node, below of=b2, node distance=1cm, label=below:$A^{(\ref{math:ex:grammar-to-tbox:down-and-up:second}, i)}$] (c2) {};
	\node[right of=c2] (c3) {};
	\node[small node, right of=c3, label=below:{$B_2^{\,\eqref{math:ex:grammar-to-tbox:two}}, C_2^{\,(\ref{math:ex:grammar-to-tbox:down-and-up:second}, ii)}$}] (c4) {};
	
	\draw[next edge] (c2) to (c4);
	
	\draw[directed edge] (a2) to node[pos=0.73, right] {$r_{22}^{(\ref{math:ex:grammar-to-tbox:down-and-up:second}, i)}$} (c2);
	
	\draw[directed edge, dotted] (a4) to[out=0, in=90] ($(a4)!0.5!(c4) + (1.5, 0)$) to[out=-90, in=0] (c4);
	
	\node[right of=b4] (r22) {$r_{22}$};
	
\end{tikzpicture}
}
		\vspace{-0.2cm}
		\caption{Derivation witnessing $(\T_G,\{A(a, 0)\})\vdash B_1(a, 4)$. 		The concept inclusions used to obtain facts are given in superscripts.}
		\label{fig:conjunctive-grammars:grammar-to-tbox}
	\end{figure}
	
\begin{example}
\label{ex:conjunctive-grammars:grammar-to-tbox}
	Consider the grammar $G$ defined by the following rules (subset of those presented in Example~\ref{ex:preliminaries:jez} such that $L_G(\Bmc_1)=\{c,c^4\}$) and the corresponding concept inclusions of $\T_G$: 
	\begin{align}
		&\B_1 \grto \B_1\B_3 \grand \B_2\B_2 
		&&C_{13} \sqcap C_{22} \sqs B_1
		\tag{\ensuremath{I_1}}\label{math:ex:grammar:-to-tbox:conjunction}\\
		&\B_1 \grto c 
		&&A \sqs \Next B_1
		\tag{\ensuremath{I_2}}\label{math:ex:grammar-to-tbox:one}\\
		&\B_2 \grto cc
		&&A \sqs \Next^2 B_2
		\tag{\ensuremath{I_3}}\label{math:ex:grammar-to-tbox:two}\\
		&\B_3 \grto ccc
		&&A \sqs \Next^3 B_3
		\tag{\ensuremath{I_4}}\label{math:ex:grammar-to-tbox:three}
	\end{align}
	Additionally, $\T_G$ contains:
	\begin{align}
			&B_1 \sqs \exists r_{13} . A  \ \ (i)
			&&B_3 \sqs C_3  \ \ (ii)
			&&\exists r_{13} . C_3 \sqs C_{13}  \ \ (iii)
			\tag{\ensuremath{I_5}}\label{math:ex:grammar-to-tbox:down-and-up:first}\\
			&B_2 \sqs \exists r_{22} . A  \ \ (i)
			&&B_2 \sqs C_2  \ \ (ii)
			&&\exists r_{22} . C_2 \sqs C_{22}  \ \ (iii)
			\tag{\ensuremath{I_6}}\label{math:ex:grammar-to-tbox:down-and-up:second}
	\end{align}
	Then, we derive $B_1(a, 4)$ from $\T_G\cup\{A(a, 0)\}$ as shown in Figure~\ref{fig:conjunctive-grammars:grammar-to-tbox}.
\end{example}

More generally, consider a rule of the form \eqref{math:grammar:to-tbox:one}: $\B_i \to \B_{i_1} \dots \B_{i_k}$. It states that $c^n \in L(\B_i)$ if $n = n_1 + \dots + n_k$ and $c^{n_j} \in L(\B_{i_j}),\ 1 \leqslant j \leqslant k$. The right-hand side of the rule is represented in $\T_G$ by $C_{i_1\dots i_k}$, and whenever $C_{i_1\dots i_k}(e,\ell)$ is derived, $B_i(e,\ell)$ follows by the corresponding concept inclusion \eqref{math:ci:from-grammar:imply}. The derivation is performed in steps. Starting from $A(a, 0)$, we derive $B_{i_1}(a, n_1)$, and then go, by $B_{i_1}\sqsubseteq\exists r_{i_1\dots i_k}. A$ \eqref{math:ci:from-grammar:down}, to a new null, $b$, where the process restarts, recursively, from $A(b, n_1)$ targeting $C_{i_2\dots i_k}$. This fact is stored in the index of the role name $r_{i_1\dots i_k}$, so when $C_{i_2\dots i_k}(b, n_1 + n')$ is inferred, and only at that point, it is lifted up to $C_{i_1\dots i_k}(a, n_1 + n')$ by $\exists r_{i_1\dots i_k}.C_{i_2\dots i_k}\sqsubseteq C_{i_1\dots i_k}$ \eqref{math:ci:from-grammar:up}. The inclusion  $B_i\sqsubseteq C_i$ \eqref{math:ci:from-grammar:conclude} ends the recursion. 
Lemma~\ref{lm:conjunctive-grammars:meet} formalizes this intuition and is proved by induction on $k$.

\begin{restatable}{lemma}{LmCgMeet}\label{lm:conjunctive-grammars:meet}
	For $\iota = i_1 \dots i_k \in \Jmc$, $\T_G \mdl A \sqs \Next^n C_\iota$ if and only if $n = n_1 + \dots + n_k$, such that $\T_G \mdl A \sqs \Next^{n_j} B_{i_j}$ for $1 \leqslant j \leqslant k$.
\end{restatable}
Concept inclusions of form \eqref{math:ci:from-grammar:meet} just generalize this to the case of conjunction. 
We prove the next lemma by showing that the existence of a derivation witnessing $\T_G \vdash A \sqs \Next^n B_i$ implies the existence of a derivation witnessing $G \vdash \B_i(c^n)$, and vice versa, by induction on the derivation length. This finalizes the proof of Theorem~\ref{thm:conjunctive-grammars:grammar-to-tbox}.
\begin{restatable}{lemma}{LmCgMain}\label{lm:conjunctive-grammars:main}
	$\T_G \mdl A \sqs \Next^n B_i$ if and only if $G \vdash \B_i(c^n)$, for $i \in \{1, \dots, m\}$.
\end{restatable}

% !TEX root =  ../ms.tex

\section{Linear \TELn and context-free grammars}\label{sec:linear}

	If a TBox \T belongs to both \TELnf and \TELnc, the grammar $G_\T$ provided by Definition~\ref{def:conjuntive-grammars:GT} in the proof of Theorem \ref{thm:conjunctive-grammars:tbox-to-grammar} 
	does not contain rules of type \eqref{math:grammar:from-tbox:conjunction}, and is thus context-free. In this section, we prove a similar result about the linear fragment in general, which is used later in Section~\ref{sec:consequences} for TAQA with \TELnc-TBoxes.
	
	\begin{restatable}{theorem}{ThmLinearTBoxesToCFGrammars}\label{thm:linear:tbox-to-grammar}
		For every \TELnc-TBox \T, there exists a context-free grammar $\Gamma_\T = (N, \{c, d\}, R)$, of size polynomial in $|\T|$, such that for any $A, B \in \CN(\T)$, there is 
		$\Nmc_{AB} \in N$ such that $\T \mdl A \sqs \Next^n B$ iff there exists $w \in L_{\Gamma_\T}(\Nmc_{AB})$ with $\#c(w) - \#d(w) = n$.
	\end{restatable}
	
	To reuse the ideas of the proof of Theorem \ref{thm:conjunctive-grammars:tbox-to-grammar}, we would need to get rid of local role names that occur in \T. However, we cannot rely on the construction of \Trig as in Section~\ref{sec:conjunctive-grammars}, since it introduces concept inclusions of the form $A\sqcap C_r\sqsubseteq A'_r$, so that \Trig does not belong to \TELnc even if \T does. 
	We treat separately the case where \T does not feature any local role name and the case where it does.
	
\subsection{The case of rigid role names only}\label{sec:linear:rigid-only}
	
	Suppose \T is a \TELnc-TBox such that all role names in \T are rigid. We construct a grammar $\Gamma_\T = (N, \{c, d\}, R)$ in a similar way as $G_\T$ in Section~\ref{sec:conjunctive-grammars}.
	
	\begin{definition}\label{def:linear:GammaT}
		Given a \TELnc-TBox \T such that $\RN(\T) \sbs \RNrig$, $\Gamma_\T = (N_\T, \{c, d\}, R_\T)$, where $N_\T = \{\Nmc_{AB} \mid A, B \in \CN(\T)\}$ and $R_\T$ contains exactly the rules defined by \eqref{math:grammar:from-tbox:epsilon}, \eqref{math:grammar:from-tbox:shift}, \eqref{math:grammar:from-tbox:down-and-up}, \eqref{math:grammar:from-tbox:middle} with $\Trig = \T$, as well as the following rules. 
		\begin{align}
			&\Nmc_{AB} \grto d^{\,|n|}, && &&\text{for } A \sqs \Next^n B \in \T,\ n < 0
			\label{math:grammar:from-tbox:negative-shift}
			\tag{\ref{math:grammar:from-tbox:shift}$^*$}
		\end{align}
	\end{definition}
	
	In a word $w \in \{c, d\}^*$, a symbol $c$ corresponds to a step forwards in time, and a symbol $d$ to a step backwards. Otherwise, the intuition behind $\Gamma_\T$ is the same as that given for $G_\T$ in Section~\ref{sec:conjunctive-grammars}. 
	For every derivation witnessing $\Gamma_\T \vdash \Nmc_{AB}(w)$, we can construct a corresponding derivation of $B(a, \# c(w) - \# d(w))$ from $\T\cup\{A(a, 0)\}$. Conversely, from a derivation witnessing 
	$(\T, \{A(a, 0)\}) \vdash B(a, n)$, we obtain a word $w$ such that $\Gamma_\T \vdash \Nmc_{AB}(w)$ as follows: whenever a derivation rule uses a concept inclusion of the form $A' \sqs \Next^k B'$, 
	we write $c^k$ if $k \geqslant 0$, and $d^{\,|k|}$ when $k < 0$. 
			\begin{figure}[t]
		\centering
		
		% !TEX root =  ../ms.tex

\scalebox{0.90}{
\begin{tikzpicture}[node distance=1.5cm]
	
	\node (tm2) {-2};
	\node[right of=tm2] (tm1) {-1};
	\node[right of=tm1] (t0) {0};
	\node[right of=t0] (t1) {1};
	\node[right of=t1] (t2) {2};
	
	\node[small node, below of=t0, node distance=1cm, label=above right:$A$] (a0) {};
	\node[above left of=a0, node distance=0.5cm] (a) {$a$};
	\node[small node, below of=a0, node distance=1cm, label=above right:$B$] (b0) {};
	
	\node[left of=b0] (bm1) {};
	\node[small node, left of=bm1, label=left:$C$] (bm2) {};
	
	\node[right of=b0] (b1) {};
	\node[small node, right of=b1, label=right:$D$] (b2) {};
	
	\node[small node, above of=b2, node distance=1cm, label=above right:$E$] (a2) {}; 
	
	\draw[directed edge] (a0) to node[pos=0.3, right] {$r$} (b0);
	\draw[directed edge, dotted] (a2) to node[pos=0.3, left] {$r$} (b2);
	
	\draw[next edge] 
		($(b0) + (-0.15, 0.03)$) to node[midway, above, yshift=0.6em, black] {$d$}
		($(bm1) + (0, 0.03)$) to node[midway, above, yshift=0.6em, black] {$d$}
		($(bm2) + (0.15, 0.03)$); 
	
	\draw[next edge] 
		($(bm2) + (0.15, -0.06)$) to node[midway, below, yshift=-0.6em, black] {$c$}
		($(bm1) + (0, -0.06)$) to node[midway, below, yshift=-0.6em, black] {$c$}
		($(b0) + (0, -0.06)$) to node[midway, below, yshift=-0.6em, black] {$c$}
		($(b1) + (0, -0.06)$) to node[midway, below, yshift=-0.6em, black] {$c$}
		($(b2) + (-0.15, -0.06)$);
		
	\draw[dotted, thick, rounded corners, red, -{Stealth[length=3pt]}] 
		($(a0) + (-0.2, 0)$) to
		($(b0) + (-0.2, 0.15)$) to 
		($(bm2) + (-0.2, 0.15)$) to 
		($(bm2) + (-0.2, -0.15)$) to
		($(b2) + (0.2, -0.15)$) to
		($(a2) + (0.2, 0)$);
	
\end{tikzpicture}
}
		\caption{A derivation witnessing $(\T,\{A(a, 0)\})\vdash E(a, 2)$ and the corresponding word $ddcccc$ can be read along the \textcolor{red}{dotted line}.}
		\label{fig:linear:tbox-to-grammar}
	\end{figure}
	
		\begin{restatable}{lemma}{LemLinearTBoxesRigidToCFGrammars}\label{lem:linear:rigid-tbox-to-grammar}
		If $\RN(\T)\subseteq\RNrig$, for any $A, B \in \CN(\T)$, $\T \mdl A \sqs \Next^n B$ iff there exists $w \in L_{\Gamma_\T}(\Nmc_{AB})$ with $\#c(w) - \#d(w) = n$.
	\end{restatable}
	\begin{example}\label{ex:linear:tbox-to-grammar}
		Consider a \TELnc-TBox \T containing the following concept inclusions, with $r \in \RNrig$. Figure \ref{fig:linear:tbox-to-grammar} shows a derivation witnessing $(\T,\{A(a, 0)\})\vdash E(a, 2)$ 
		and the corresponding word.
		\begin{align*}
			&A \sqs \exists r . B
			&&B \sqs \Next^{-2} C
			&&C \sqs \Next^4 D
			&&\exists r . D \sqs E
		\end{align*}
	\end{example}
	
	Lemma~\ref{lem:linear:rigid-tbox-to-grammar} cannot be extended beyond \TELnl because different words, say $ddcccc$ and $cdcccd$, correspond to the same ``shift'' in time (here, $2$), but the semantics of rules of type \eqref{math:grammar:from-tbox:conjunction} treats them as different: if $\Nmc\to ddcccc$ and $\Mmc \to cdcccd$, then $\Nmc\&\Mmc$ generates an empty language.

\subsection{The case of both rigid and local role names}\label{sec:linear:rigid-and-local}

For this case, we provide a translation from \T to $\Gamma_\T$ which is not constructive but nevertheless guarantees the existence of $\Gamma_\T$ and the bound on its size stated by Theorem \ref{thm:linear:tbox-to-grammar}. 
To actually build $\Gamma_\T$, one needs to compute the set $\{A \sqs B \mid \T \mdl A \sqs B\}$, and we do not provide any recipe for that when \T features local role names. However, Theorem~\ref{thm:linear:tbox-to-grammar} is enough to show ultimate periodicity of  \TELnc-TBoxes, as we do in the next section.

	Consider a \TELnc-TBox \T and let $\T_0$ be obtained from \T by removing all concept inclusions that use local role names. Then, let 
	\begin{align}\label{math:linear:rigid-and-local:Triglin}
		\Triglin = \T_0 \cup \{A \sqs B \mid \T \mdl A \sqs B\}.
	\end{align}
	Intuitively, in a derivation using \T, if a rule of form  \eqref{math:inference-rule:existential} produces a fact $r(a, b, n)$ with some $r \in \RNloc$, and later a rule of form  \eqref{math:inference-rule:return} uses $r(a, b, n')$ to ``propagate back'' the consequences of facts derived about $b$ on $a$, the locality of $r$ implies that $n = n'$. A derivation that uses \Triglin ``skips'' the derivation between these two rule applications, immediately inferring a new fact for $a$ at time $n$. 

	\begin{example}\label{ex:linear:local-roles-removal}
		Let $r \in \RNloc$ and \T contain the concept inclusions:
		\begin{align*}
			&A \sqs \Next B &&F \sqs \Next G
			&C \sqs \Next^{-3} D &&D \sqs \Next^{3} E
			\\
			&&&B \sqs \exists r . C &\exists r . E \sqs F
		\end{align*}
		Then $\T \mdl B \sqs F$, and \Triglin contains the concept inclusions in the first line and $B \sqs F$. 
		Figure \ref{fig:linear:rigidisation} illustrates two derivations of $G(a, 2)$: one from $\{A(a, 0)\}\cup\T$, and one from $\{A(a, 0)\}\cup\Triglin$. 
	\end{example}
	In general, the following lemma holds.
	\begin{restatable}{lemma}{LmLinearLocalRolesRemoval}\label{lm:linear:local-roles-removal}
		Let $\T$ be a \TELnc-TBox. For any $A, B \in \CN(\T)$ and $n \in \Z$, $\T \mdl A \sqs \Next^n B$ if and only if $\Triglin \mdl A \sqs \Next^n B$.
	\end{restatable}

	\begin{figure}[t]
		\centering
		
		% !TEX root =  ../ms.tex

\scalebox{0.85}{
\begin{tikzpicture}[node distance=1.6cm, label distance=0.5em]
	
	\node (t0) {0};
	\node[left of=t0] (tm1) {-1};
	\node[left of=tm1] (tm2) {-2};
	\node[right of=t0] (t1) {1};
	\node[right of=t1] (t2) {2};
		
	\node[small node, label=above:$A$, below of=t0, node distance=0.9cm] (a0) {};
	\node[small node, label=above:{$B, F$}, right of=a0] (a1) {};
	\node[small node, label=above:$G$, right of=a1] (a2) {};
	\node[label=above:{}, left of=a0] (am1) {};
	\node[label=above:{}, left of=am1] (am2) {};
	
	\draw[next edge] (a0) to (a1);
	\draw[next edge] (a1) to (a2);
	
	\node[left of=a0, node distance=0.5cm] (a) {$a$};
	\node[below of=a, node distance=1.4cm] (b) {};
		
	\node[small node, label={below:{$C, E$}}, below of=a1,node distance=1.2cm] (b1) {};
	\node[label=below:{}, left of=b1] (b0) {};
	\node[label=below:{}, left of=b0] (bm1) {};
	\node[small node, label={below:{$D$}}, left of=bm1, label distance=0.5em] (bm2) {};
		
	\node[above of=bm2, node distance=0.05cm] (bm2bis) {};
	\node[above of=b1, node distance=0.05cm] (b1bis) {};
	\draw[next edge] (b1bis) to (bm2bis);
	
	\node[below of=bm2, node distance=0.07cm] (bm2ter) {};
	\node[below of=b1, node distance=0.07cm] (b1ter) {};
	\draw[next edge] (bm2ter) to (b1ter);
	
	\draw[directed edge] (a1) to node[midway, right, xshift=0.5em] {$r$} (b1);
		
	\draw[dotted, thick, rounded corners, red, -{Stealth[length=3pt]}] 
		($(a0) + (0, -0.1)$) to
		 ($(a1) + (-0.2, -0.1)$) to 
		 ($(b1) + (-0.2, 0.1)$) to 
		 ($(bm2) + (-0.2, 0.1)$) to
		 ($(bm2) + (-0.2, -0.1)$) to
		 ($(b1) + (0.2, -0.1)$) to
		 ($(a1) + (0.2, -0.1)$) to
		 ($(a2) + (0, -0.1)$);

	\draw[dash dot dot, blue, thick, -{Stealth[length=3pt]}] 
		($(a0) + (0, 0.1)$) to
		($(a2) + (0, 0.1)$);

	\begin{pgfonlayer}{background}
		\draw[dashed, rounded corners, fill=gray!20] 
		($(am2) + (-0.4, -0.5)$) to
		($(bm2) + (-0.4, -0.5)$) to
		($(b1) + (0.7, -0.5)$) to
		($(a1) + (0.7, -0.2)$) to
		($(am2) + (-0.4, -0.2)$) to
		($(bm2) + (-0.4, 0)$);
	\end{pgfonlayer}

\end{tikzpicture}
}
		
		\caption{A derivation of $G(a, 2)$ from $\{A(a, 0)\}$  and the TBox \T of Example~\ref{ex:linear:local-roles-removal} goes along the \textcolor{red}{dotted line}. A derivation with \Triglin uses $B \sqs F$ and skips the part in the grey box, passing along the \textcolor{blue}{dash-dotted line}.}
		\label{fig:linear:rigidisation}
	\end{figure}
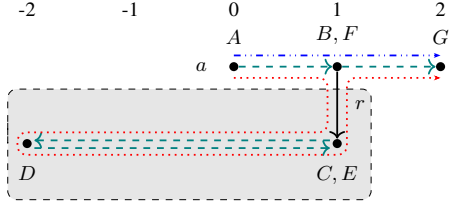

	The proof of Theorem \ref{thm:linear:tbox-to-grammar} uses Lemma~\ref{lm:linear:local-roles-removal} and the facts that $\RN(\Triglin) \sbs \RNrig$ and that $|\Triglin|$ is polynomial in $|\T|$ to assume that \T uses only rigid roles, and concludes via Lemma~\ref{lem:linear:rigid-tbox-to-grammar} and the polynomial time construction of $\Gamma_\T$ in Definition~\ref{def:linear:GammaT}.

% !TEX root =  ../ms.tex

\section{Answering temporal atomic queries}\label{sec:consequences}

We now draw the consequences of the correspondences between \TELnf or \TELnl-TBoxes and grammars for TAQA.

\subsection{Undecidability results}\label{sec:consequneces:undecidability}

 First, we close negatively the question of \TELn ultimate periodicity.

\begin{theorem}\label{thm:consequences:not-ultimately-periodic}
	The following statements hold.
	\begin{enumerate}[label=(\roman*)]
		\item There exists a \TELnf-TBox which is not ultimately periodic. \label{part:nup:exists}
		
		\item It is undecidable to check if the set $\{n \in \N \mid \T \mdl A \sqs \Next^n B\}$ is semilinear for a \TELnf-TBox \T and $A, B \in \CN(\T)$. \label{part:nup:undecidable}
	\end{enumerate}	
\end{theorem}
\begin{proof}
\emph{\ref{part:nup:exists}} Let $G$ be the grammar of Example~\ref{ex:preliminaries:jez}. By Theorem \ref{thm:conjunctive-grammars:grammar-to-tbox}, there exists a \TELnf-TBox $\T_G$ and $A, B \in \CN(\T_G)$ such that $\T_G \mdl A \sqs \Next^n B$ iff $c^n \in L_G(\Nmc_1)$. Hence, $\{n \in \N \mid \T_G \mdl A \sqs \Next^n B\} = \{4^k \mid k \in \N\}$, which is not semilinear.

	\emph{\ref{part:nup:undecidable}} We proceed by a reduction to the problem of deciding whether $L(G)$ is regular for a unary conjunctive grammar $G$, which is undecidable by point \emph{(iii)} of Theorem \ref{thm:jez-okhotin:undecidability}.  Given a unary conjunctive grammar $G$, by Theorem~\ref{thm:conjunctive-grammars:grammar-to-tbox}, one can construct a \TELnf-TBox $\T_G$ and $A, B \in \CN(\T_G)$ such that $\T_G \mdl A \sqs \Next^n B$ iff $c^n \in L_G(\Smc)$, where $\Smc$ is the start symbol of $G$, so that $L(G)=L_G(\Smc)$. Hence, $\{(n) \in \N \mid \T_G \mdl A \sqs \Next^n B\}$ is the Parikh image of $L(G)$. Thus, by Theorem \ref{thm:parikh}, $L(G)$ is regular iff this set is semilinear.
\end{proof}

Second, we obtain some undecidability results for TAQA.
\begin{theorem}\label{thm:consequences:undecidability}
	TAQA is undecidable, for combined complexity, with \TELn-TBoxes and with \TELnf-TBoxes extended with rigid concept names.
\end{theorem}
\begin{proof}	
	By Theorem \ref{thm:jez-okhotin:undecidability} \e{(i)}, checking if $L(G) = \emptyset$ for a given unary conjunctive grammar $G$ is undecidable. We reduce this problem to TAQA. As in the proof of Theorem~\ref{thm:consequences:not-ultimately-periodic} \e{(ii)}, given a unary conjunctive grammar $G$, one can construct a \TELnf-TBox $\T_G$ and $A, B \in \CN(\T_G)$ such that $\T_G \mdl A \sqs \Next^n B$ iff $c^n \in L(G)$. Let $C$ be a fresh concept name. We construct $\T'_G$ such that $L(G) \neq \emptyset$ iff $(\T'_G, \{A(a, 0)\}) \mdl C(a, 0)$ as follows. For the \TELn case, let $\T'_G = \T_G \cup \{B \sqs C, C \sqs \Next^{-1} C\}$, and for \TELnf with rigid concept names, let $C$ be rigid and $\T'_G = \T_G \cup \{B \sqs C\}$.
\end{proof}
It is known
that rigid concept names can be simulated with rigid role names if the language allows for the concept $\top$ (which is such that $\top^{\Imc_i} = \Delta^\Jmf$ for every $\Jmf=(\Delta^\Jmf, (\Imc_i)_{i\in\Z})$ and $i\in\Z$, and is not allowed in the original definition of \TELn by \citet{Basulto-et-al:TEL}): adding $C\equiv \exists r . \top$ for a fresh $r \in \RNrig$ to the TBox makes $C$ rigid. 
It thus follows from Theorem \ref{thm:consequences:undecidability} that TAQA is undecidable also if one extends \TELnf with $\top$. 
Note that Theorem~\ref{thm:consequences:undecidability} does not imply undecidability for data complexity. However, it implies that even having fixed a TBox \T, there is no ``computational'' way to obtain an algorithm for TAQA with that \T.

\subsection{Decidability results}\label{sec:consequneces:decidability}

We first show that TAQA is decidable with \TELnf-TBoxes. 
\begin{restatable}{theorem}{ThmFuturePTime}\label{thm:consequences:future-answering} 
	TAQA with \TELnf-TBoxes is \PTime-complete, both for combined and data complexity.
\end{restatable}
\begin{proof}[Proof sketch]
	The lower bounds hold already for the description logic \EL (without temporal operators) \citep{DBLP:journals/ai/CalvaneseGLLR13}. 
	For the upper bounds, we provide a polynomial reduction from the problem of deciding whether $(\T, \A) \mdl A(a, n)$ to that of checking whether a word belongs to the language of a conjunctive grammar, which can be tested in polynomial time (Theorem~\ref{thm:membership-testing}). Our reduction builds a \TELnf-TBox $\T'\cup\T_\A$, an assertion $C_a(a, l)$ and a concept name $A_n$ such that $(\T, \A) \mdl A(a, n)$ iff $(\T' \cup \T_\A, \{C_a(a, l)\}) \mdl A_n(a, n)$, then use Proposition \ref{prop:derivationTEL} and Theorem \ref{thm:conjunctive-grammars:tbox-to-grammar} to conclude. The idea is to encode all information about $a$ in $\A$ into $C_a(a, l)$ thanks to $\T_\A$. 
	
Let  $\ind(\A)$ be the set of individual names that occur in \A, and $l, m \in \Z$ be the least and the greatest timestamps appearing in \A. 
	We introduce fresh concept names $\{C_a \mid a \in \ind(\A)\}$ and $\{A_k \mid A \in \CN(\T), l \leqslant k \leqslant m + 1\}$, and role names $\{\role{ab} \in \RNrig \mid  a, b \in \ind(\A), {r\in\NR(\T)}\}$. For the convenience of notation, we write $A_k$ for all $k \geqslant l$, assuming that $A_k = A_{m + 1}$ when $k > m$. The TBox $\T'$ contains the following inclusions, for all $k\in\{l,\dots, m+1\}$.
	\begin{align}
			&A_k \sqs \Next^s B_{k + s} &&\text{for } A \sqs \Next^s B \in \T
			\label{math:ci:abox-to-tbox:shift}\\
			&A_k \sqcap A'_k \sqs B_k &&\text{for } A \sqcap A' \sqs B \in \T
			\label{math:ci:abox-to-tbox:conjunction}\\
			&A_k \sqs \exists r . B_{k} &&\text{for } A \sqs \exists r . B \in \T
			\label{math:ci:abox-to-tbox:existential:null}\\
			&\exists r . A_k \sqs B_{k} &&\text{for } \exists r . A \sqs B \in \T
			\label{math:ci:abox-to-tbox:return:null}
	\end{align}
	Additionally, the TBox $\T_\A$ contains the following inclusions.
	\begin{align}
			&C_a \sqs \Next^{k - l} A_k &&\text{for } A(a, k) \in \A
			\label{math:ci:abox-to-tbox:mark}\\
			&C_a \sqs \exists \role{ab}\, .\, C_{b} &&\text{for } r(a, b, \ell) \in \A
			\label{math:ci:abox-to-tbox:existential:abox}\\
			&\exists \role{ab}\, .\, A_k \sqs B_k &&\text{for } \exists r . A \sqs B \in \T, r(a, b, \ell) \in \A,
			\label{math:ci:abox-to-tbox:return:abox}\\
			&\ \notag 														&&\text{where } r \in \RNrig \text{ or } \ell = k.
 	\end{align}
	Both $\T'$ and $\T_\A$ can be constructed in polynomial time w.r.t.~$|\T|+|\A|$ and are expressed in \TELnf (since \T is a \TELnf-TBox and $k-l\geqslant 0$ for every $A(a, k) \in \A$ by definition of $l$) and we show that $(\T, \A) \mdl A(a, n)$ iff $(\T' \cup \T_\A, \{C_a(a, l)\}) \mdl A_n(a, n)$. 
\end{proof}
		
One could prove Theorem~\ref{thm:consequences:future-answering} without using grammars. Given $(\T, \A)$, 
and $A(a, n)$, construct $(\T_n, \A')$ with $\T_n$ defined as $\T'$ above except that we use {$k\in\{l,\dots, n\}$} in the concept inclusions, and {$\A'=\{A_k(b, k) \mid A(b, k) \in \A \} \cup \{ r(b, c, k) \in \A\}$}. Using the inability of \TELnf to reason backwards, one can show that $(\T, \A) \mdl A(a, n)$ iff $(\T_n, \A') \mdl A_n(a, n)$. 
Then, $\T_n$ is \emph{temporally acyclic}, in the terminology of \citet{Basulto-et-al:TEL}, who prove that TAQA with such TBoxes is in \PTime, even with rigid concept names (this is not a contradiction with Theorem~\ref{thm:consequences:undecidability}, since $\T_n$ is constructed for a fixed $n$, while in the proof of Theorem~\ref{thm:consequences:undecidability} rigid concept names are used to simulate an existential query of the form $\exists n . A(a, n)$). 
On the other hand, the technique we present here allows us to reuse existing algorithms and tools for conjunctive grammars: a parser generator \whalecalf \citep{Okhotin:Whale-Calf} and an efficient parsing method tailored to \e{unary} conjunctive grammars \citep{Okhotin-Reitwiessner:Parsing-Unary-Grammars}.

We also obtain positive results for the linear fragment. In the next theorem, we use the following measure for the ``size'' of semilinear sets. 
	Let $\|u\| = |u_1| + \dots + |u_n|$ for $\vec{u} \in \Z^n$. We define $\|\Lmc\|$ as the least number $\|\vec{b}\| + \|\vec p_1\| + \dots + \|\vec p_l\|$ among all representations of a linear set \Lmc by an offset $\vec{b}$ and periods $\vec p_1,\dots \vec p_l$, and $\|S\|$ as the least sum $\|\Lmc_1\| + \dots + \|\Lmc_m\|$ among all representations of a semilinear set $S$ as a union of linear sets. For an ultimately periodic TBox \T, we set $\|\T\|=\max_{A, B \in \CN(\T)}(\|\{n \in \Z \mid \T \mdl A \sqs \Next^n B\}\|)$.

\begin{restatable}{theorem}{ThmLinearUP}\label{thm:consequences:linear:complexity}
	The following statements hold.
	\begin{enumerate}[label=(\roman*)]
		\item Every \TELnc-TBox \T is ultimately periodic, $\|\T\| \leqslant 2^\poly{|\T|}$. \label{part:linear:up}
		
		\item TAQA with \TELnc-TBoxes is \NLogSpace-complete, for data complexity. \label{part:linear:data}
		
		\item TAQA with \TELnc-TBoxes without local role names is in \ExpSpace, for combined complexity.
		\label{part:linear:combined}
	\end{enumerate}	
\end{restatable}
\begin{proof}[Proof sketch] 
For \e{\ref{part:linear:up}}, let \T be a \TELnc-TBox. By Theorem \ref{thm:linear:tbox-to-grammar}, there exists a context-free grammar $\Gamma_\T=(N,\{c,d\},R)$ such that for any $A, B \in \CN(\T)$, there is $\Nmc_{AB}\in N$ such that $\Tmc\models A\sqsubseteq \Next^n B$ iff there exists $w\in L_{\Gamma_\T}(\Nmc_{AB})$ with $\#c(w) - \#d(w)=n$. 
Let $L = L_{\Gamma_\T}(\Nmc_{AB})$. By Theorem \ref{thm:parikh}, since $\Gamma_\T$ is context-free, the Parikh image $p(L)\sbs \N^2$ of $L$ (with alphabet ordered as $c, d$) is semilinear. 
Hence, for every $A, B \in \CN(\T)$, $\{n \in \Z \mid \T \mdl A \sqs \Next^n B\} = \{n \in \Z \mid n = \#c(w) - \#d(w), w \in L\} = \{u_1 - u_2 \mid \vec{u} \in p(L)\}$ is semilinear, since it is the image of the semilinear set $p(L)$ under the linear mapping $\vec{u} \mapsto u_1 - u_2$. For the bound on $\|\T\|$, we use the methods introduced by \citet{Esparza:Parikh-Automaton} to establish that $\|p(L)\| = 2^\poly{|\Gamma_\T|}$, and further observe that $\|\{u_1 - u_2 \mid \vec{u} \in p(L)\}\| \leqslant 2\|p(L)\|$.
	
	The lower bound in \e{\ref{part:linear:data}} follows from the atemporal case \cite{DBLP:journals/ai/CalvaneseGLLR13}. 
	For the upper bounds, both in \e{\ref{part:linear:data}} and \e{\ref{part:linear:combined}}, we provide a translation of \TELnc-TBoxes to programs of linear \dlnd that extends linear Datalog with operators \Next/\Prev and \Df/\Dp 
	(cf. proof of \citep[Th. 5]{Basulto-et-al:TEL}).
		
	Given a \TELnc TBox \T, let $\Phi_\T$ be a (linear) \dlnd program defined as follows. For each $\exists r.A \sqs B \in \T$, it contains 
	$ B(x) \impd r(x, y), A(y)$, and, if $r \in \RNrig$, also $B(x) \impd \D r(x, y), A(y)$ and $B(x) \impd \Dp r(x, y), A(y)$. 
		If $r(a, b, k) \in \A$ for $r \in \RNrig$, the facts $r(a, b, m)$, $m \in \Z$, are not derived explicitly, but are simulated by the latter two rules.
		Further, for every $A, B  \in \CN(\T)$, we take a representation of the semilinear set $\{n \in \Z \mid \T \mdl A \sqs \Next^n B\}$ as a union of linear sets $\Lmc_1 \cup \dots \cup \Lmc_m$, $\Lmc_i = \{b^i + k_1 p^{i}_{1} + \dots k_l p^{i}_{l} \mid k_1, \dots, k_l \in \N\}$.
		Then, $\Phi_\T$ contains the rules $F_i^{AB}(x) \impd \Next^{-b^i} A(x)$, $F_i^{AB}(x) \impd \Next^{-p^i_j} F^{AB}_i(x)$, $B(x) \impd F_i^{AB}(x)$, for all $i \in \{1, \dots, m\}$.
		Thus, if $A(a, k) \in \A$, the fact $F_i^{AB}(a, m)$, and therefore $B(a, m)$ is derived for all $m$ such that $m - k \in \Lmc_i$. It can be shown that $(\T, \A) \mdl A(a, n)$ iff $(\Phi_\T, \A) \mdl A(a, n)$ and that $|\Phi_\T| \leqslant 2^\poly{|\T|}$.
		
		Moreover, if $\RN(\T) \sbs \RNrig$, $\Phi_\T$ can be built from $\T$ using the automata-theoretic construction of \citet{Esparza:Parikh-Automaton}.  Points \e{(ii)} and \e{(iii)} then follow from the results on \dlnd \citep{Artale:Linear-Temporal-Datalog}.
\end{proof}

\begin{remark}
Following \citet{Basulto-et-al:TEL}, we assume a unary encoding of numbers. 
If sequences $\Next\cdots\Next$ were written as $\Next^n$ with $n$ encoded in binary, translating a TBox into a grammar would cause an exponential blow-up (since $c^n$ is actually a word of length $n$), making our algorithms exponentially slower w.r.t.~$|\Tmc|$.
\end{remark}

% !TEX root =  ../ms.tex

\section{Discussion}\label{sec:conslusions}

Connections between temporal logics, such as \LTL, and formal languages (in particular, regular languages) are well-known \citep{Vardi:Automata-for-LTL, Demri-Goranko-Lange:Temporal-Logics}. This paper makes new ones, between the fragments \TELnc and \TELnf of the temporal description logic \TELn, on the one side, and context-free and unary conjunctive languages, on the other. Using these connections, we obtain several important results on \TELn, both negative and positive, from the formal language theory.

In particular, \TELnf-TBoxes are in a one-to-one correspondence with unary conjunctive grammars (Theorems \ref{thm:conjunctive-grammars:tbox-to-grammar} and \ref{thm:conjunctive-grammars:grammar-to-tbox}). 
Therefore, \TELnf is not ultimately periodic (Theorem \ref{thm:consequences:not-ultimately-periodic}), which is arguably unexpected, as its temporal component, \LTL, is ultimately periodic \citep{Sistla-and-Clarke:LTL-Ultimately-Periodic}, and its DL component, \EL, is such that every KB possesses a canonical model which has, informally speaking, a regular structure \cite{DBLP:conf/rweb/KontchakovZ14}. Moreover, TAQA with \TELn-TBoxes is undecidable (Theorem \ref{thm:consequences:undecidability}). On the other hand, the same correspondence allows us to use parsing algorithms for conjunctive grammars as tools for TAQA with \TELnf-TBoxes, leading to a drastic decrease of complexity, from undecidability to polynomial time, owed to a mere removal of the temporal operator $\Next^{-}$ (previous). 
Despite the partial undecidability result of Theorem~\ref{thm:consequences:not-ultimately-periodic}, it remains open if ultimate periodicity of \TELn-TBoxes is decidable, since for the corresponding problem---given a unary conjunctive grammar tell if all its nonterminals generate regular languages---no result is known.

The linear fragment, \TELnc, is connected to context-free grammars (Theorem \ref{thm:linear:tbox-to-grammar}). As a result, it is ultimately periodic, and enjoys considerably low data complexity of query answering (Theorem \ref{thm:consequences:linear:complexity}).

Language theorists may find interesting that every unary conjunctive grammar is transformable, in polynomial time, to the ``normal form'' of Definition \ref{def:conjuntive-grammars:GT}, by applying first Theorem \ref{thm:conjunctive-grammars:grammar-to-tbox}, then Theorem \ref{thm:conjunctive-grammars:tbox-to-grammar}. 
Moreover, the grammars that correspond to 
{``pure \LTL'' TBoxes (i.e., do not have any rule of form \eqref{math:grammar:from-tbox:down-and-up} in this normal form)} are guaranteed to generate regular languages. This is, to the best of our knowledge, the first nontrivial sufficient condition for this property.

On the more theoretical side, it is possible that the TBox-grammar correspondence can be lifted to more expressive temporal description logics and more general classes of formal grammars (e.g. Boolean grammars \citep{Okhotin:Conjunctive-and-Boolean-Grammars-Survey}). 
On the applications side, we hope to employ this correspondence to develop a practical reasoner for \TELnf.

\begin{ack}
Many thanks to Roman Kontchakov for introducing us to the problem of ultimate periodicity from the perspective of semilinear sets. We would also like to thank Alexander Okhotin for providing us with the code of \whalecalf and valuable suggestions on the implementation.
\end{ack}

\bibliography{bibliography}

\appendix
% !TEX root =  ../ms.tex

\section{Two definitions of ultimate periodicity}\label{app:prelim:quasimodels}

\citet{Basulto-et-al:TEL} define ultimate periodicity using  quasimodels. We show here that every \TELn-TBox $\T$ ultimately periodic under this definition is also ultimately periodic as defined in Section~\ref{sec:preliminaries}, and vice versa. First, we recall the notion of quasimodels. The following few paragraphs quote almost verbatim \citet{Basulto-et-al:TEL} (the difference is that they use concept inclusions of the form $\Next^k A \sqs B$, $k \in \{-1, 0, 1\}$, while we use $A \sqs \Next^k B$, $k \in \Z$; it is straightforward that these forms are equivalent in terms of expressive power).

Fix a KB $(\T, \A)$ with a \TELn-TBox \T, and let $\CN(\T, \A)$ be the set of concept names used in $(\T, \A)$. A map
$\pi \colon \Z \to 2^{\CN(\T, \A)}$ is a trace for \T if it satisfies the following:
\begin{itemize}
	\item[\textbf{(t1)}] if $A \sqcap A' \sqs B \in \T$ and $A, A' \in \pi(n)$, then $B \in \pi(n)$;
	\item[\textbf{(t2)}] if $A \sqcap \Next^k B\in\T$ and $A \in \pi(n)$, then $B \in \pi(n + k)$. 
\end{itemize}
Let $\pi$ be a trace for \T. For a rigid role name $r \in \RNrig$ the $r$-projection of $\pi$ is a map $\proj_r(\pi) \colon \Z \to 2^{\CN(\T, \A)}$ that sends each $i \in \Z$ to $\{A \mid \exists r . B \sqs A \in \T, B \in \pi(i)\}$. For a local role name  $r \in \RNloc$, $\proj_r(\pi)$ is defined in the same way
on 0 but is $\emptyset$ for all other $i \in \Z$. Given a map $\rho \colon \Z \to 2^{\CN(\T, \A)}$ and
$n \in \Z$, we say that $\pi$ contains the $n$-shift of $\rho$ and write
$\rho \sbs^n \pi$ if $\rho(i - n) \subseteq \pi(i)$, for all $i \in \Z$. 

We now define quasimodels.
Let $D = \ind(\A) \cup \CN(\T, \A)$. A \e{quasimodel} $\Qmf$ for
$(\T, \A)$ is a set $\{\pi_d \mid d \in D\}$ of traces for \T such that
\begin{itemize}
	\item[\textbf{(q1)}] $A \in \pi_a(n)$, for all $A(a, n) \in \A$; 
	\item[\textbf{(q2)}] $B \in \pi_B(0)$, for all $B \in \CN(\T, \A)$; 
	\item[\textbf{(q3)}] $A \in \pi_a(n)$, for all $B \in \pi_b(n)$, $r(a, b, n) \in \A$, and $ \exists r . B \sqsubseteq A \in \T$;
	\item[\textbf{(q3')}] $\proj_r(\pi_b) \sbs^0 \pi_a$, for all $r(a, b, n) \in \A$, $r \in \RNrig$;
	\item[\textbf{(q4)}] if $A \in \pi_d(n)$, then $\proj_r(\pi_B) \sbs^n \pi_d$, for all $d \in D$, $n \in \Z$ and $A \sqs \exists r . B \in \T$.
\end{itemize}
{Compared to the original \citep{Basulto-et-al:TEL}, we split the point \textbf{(q3)} into two versions, treating local and rigid roles, fixing a small glitch in the original definition after a discussion with the authors. This does not affect further results.}

Intuitively, $\pi_a$
represents $a \in \ind(\A)$; and $\pi_B$ represents all elements that witness $B$ for $A \sqs \exists r . B \in \T$. 
For the purposes of TAQA, canonical quasimodels are defined. The \e{canonical
quasimodel} is the limit of the following saturation procedure. Start with initially empty maps $\pi_d$, for $d \in D$, and
apply \textbf{(t1)}-\textbf{(t2)}, \textbf{(q1)}-\textbf{(q4)} as rules: e.g., \textbf{(q3')} says ``if
$r(a, b, n) \in \A$, for $r \in \RNrig$, and $A \in \proj_r(\pi_b)(i)$, then add $A$ to $\pi_a(i)$''.

\begin{apptheorem}[\citet{Basulto-et-al:TEL}] \label{thm:app:qm}
	Let $\Qmf = \{\pi_d \mid d \in D\}$ be the canonical quasimodel of $(\T, \A)$ where $\T$ is a \TELn-TBox. Then, for any $A \in \CN$, $(\T, \A) \mdl A(a, n)$ iff $A \in \pi_a(n)$, for $a \in \ind(\A), n \in \Z$.
\end{apptheorem}

Finally, let \T be a \TELn-TBox and \Qmf the canonical quasimodel for $(\T, \emptyset)$. We
say that \T is \textit{ultimately periodic w.r.t. quasimodels}, if for 
each $A \in \CN(\T)$ there are positive integers $m_P, p_P, m_F, p_F \in \N$, such that the following conditions hold for $\pi_d$, with $d = A$.
\begin{align}
	&\pi_d(n - p_P) = \pi_d(n)  &&\text{for all } n \leqslant -m_P
	\label{math:app:prelims:up:past}\\
	&\pi_d(n + p_F) = \pi_d(n) &&\text{for all } n \geqslant m_F
	\label{math:app:prelims:up:future}
\end{align}

Note that by definition, the traces $\pi_A$, $A \in \CN(\T)$, are the same in all canonical quasimodels of $(\T, \A)$, for any \A. We observe the following property.
\begin{applemma}\label{lm:app:qm}
	For any \TELn-TBox \T, $A \in \CN(\T)$, and $n \in \Z$, $\pi_A(n) = \{B \in \CN(\T) \mid \T \mdl A \sqs \Next^n B\}$.
\end{applemma}
\begin{proof}
	Suppose $\Qmf = \{\pi_d \mid d \in  \CN(\T) \cup \{a\}\}$ is the canonical quasimodel of $(\T, \{A(a, 0)\})$, and observe further that in this case $\pi_A = \pi_a$. The lemma follows by Theorem \ref{thm:app:qm} and Proposition \ref{prop:derivationTEL}.
\end{proof}

Now, recall from Section \ref{sec:preliminaries}, that we call a \TELn-TBox \T \textit{ultimately periodic (w.r.t. concept inclusions)}, if for every pair $A, B \in \CN(\T)$ the set $\{n \in \Z \mid \T \mdl A \sqs \Next^n B\}$ is semilinear.

Our goal in this section is to prove that the two definitions are equivalent (Proposition \ref{prop:app:def-equivalence} below). First, we recall a useful fact from arithmetic (see Niven, I., Zuckerman, H. S., Montgomery, H. L., ``An Introduction to the Theory of Numbers'' (1991), for details).

\begin{applemma}[Linear Diophantine equations] \label{lm:app:diophante}
	An equation $px + qy = c$, where $p, q, c \in \Z$, has a solution (where x and y are integers) if and only if $c$ is a multiple of the greatest common divisor of $p$ and $q$. Moreover, if $(x, y)$ is a solution, then the other solutions have the form $(x + tp', y - tq')$, $t \in \Z$, and $p'$ and $q'$ are the quotients of $p$ and $q$ (respectively) by the greatest common divisor of $p$ and $q$.
\end{applemma}

In this section, a \textit{simple} set is a set of the form $\{b + kp \mid k \in \N\}$, where $b, p \in \Z$. One can show that every semilinear set $S \sbs \Z$ is representable as a union of simple sets. For completeness of the presentation, we give a direct proof here.

\begin{applemma}\label{lm:app:simple}
	The following statements hold.
	\begin{enumerate}[label=(\roman*)]
		\item Every set of the form $\{kp + mq \mid k, m \in \N\}$, where $p, q \in \Z$, is a finite union of simple sets.
		
		\item Every semilinear $S \sbs \Z$ is a finite union of simple sets.
	\end{enumerate}
\end{applemma}
\begin{proof}
	\e{(i)} Let $\Lmc = \{kp + mq \mid k, m \in \N\}$. If $p=0$ or $q = 0$, \Lmc is itself a simple set. So suppose $p \neq 0$ and $q\neq 0$. We consider three cases:
	\begin{itemize}
		\item If $p, q > 0$, then $\Lmc \sbs \N$, and by Theorem~\ref{thm:parikh} $\Lmc = p(L)$ for some regular language $L \sbs \{c\}^*$. Let $M = (Q, \{c\}, q_0, \delta, F)$ be the minimal deterministic automaton recognising $L$, with the set of states $Q$ and final states $F \sbs Q$, initial state $q_0 \in Q$ and $\delta \colon Q \times \{c\} \to Q$. It is easy to see that for the sequence $q_0, q_1, q_2, \dots$, where $q_{i + 1} = \delta(q_i, c)$, there exist unique $m, p \in \N$ such that $q_{j} = q_{j + p}$ for all $j \geqslant m$. Then 
		$$
		\Lmc = 
		\{j < m \mid q_j \in F \} \quad \cup \hspace{-10pt} \bigcup_{\begin{array}{c}
				m \leqslant j < m + p\\
				q_j \in F
		\end{array}} \hspace{-20pt}\{j + k p \mid k \in \N\}
		$$
		
		\item If $p, q < 0$, we take $\Lmc' = \{-1 \cdot n \mid n \in \Lmc\}$, and obtain the representation as in the previous case, then multiply everything by $-1$.
		
		\item If $p$ and $q$ are of different signs. Let $d \geqslant 1$ be their greatest common divisor, and $p' = p / d$, and $q' = q / d$. We have $\Lmc = f(\Lmc')$, where $f(x) = dx$, $x \in \Z$, and $\Lmc' = \{kp' - mq' \mid k, m \in \N\}$. Observe that if a set is simple, its image under $f$ is also simple --- hence it is enough to represent $\Lmc'$ as a union of simple sets. By Lemma \ref{lm:app:diophante}, there are $k, m \in \Z$ such that $kp' + mq' = 1$. Since $p$ and $q$ are of different signs, we can safely assume that $k, m \in \N$ (otherwise, we find $t \in \Z$ such that $k + tp', m - tq' > 0$, and take that solution). It follows that $\N \sbs \Lmc'$. Similarly, we find a positive integer solution for $kp' + mq' = -1$, and conclude that $\{-n \mid n \in \N\} \sbs \Lmc'$. Thus $\Lmc' = \Z = \{k \cdot 1 \mid k \in \N\} \cup \{k \cdot (-1) \mid k \in \N\}$.
	\end{itemize}
	
	\e{(ii)} Since a semilinear set is a finite union of linear sets, it is enough to prove that every linear set \Lmc is a finite union of simple sets. Suppose $\Lmc = \{b + k_1p_1 + \dots + k_lp_l \mid k_1, \dots, k_l \in \N\}$. We do induction on $l$. The cases $l = 0$ or $l = 1$ are trivial. Now suppose that any linear set representable using $l - 1$ periods is a union of simple sets.
	
	Then we have:
	\begin{align*}
		&S = \{b + k_1p_1 + \dots + k_lp_l \mid k_i \in \N\} =\\
		&\{b + k_1 p_1 + \dots + k_{l-1} p_{l - 1} \mid k_i \in \N\} + \{k_l p_l \mid k_l \in \N\} =\\
		& \left(\bigcup_{i = 1}^m \Lmc_i \right) +  \{k_l p_l \mid k_l \in \N\} =
		  \bigcup_{i = 1}^m (\Lmc_i + \{k_l p_l \mid k_l \in \N\})
	\end{align*}
	where $S_1 + S_2 = \{x + y \mid x \in S_1, y \in S_2\}$, and $\Lmc_i$ are simple.
	
	If $\Lmc_i = \{b\}$, then $\Lmc_i + \{k_l p_l \mid k_l \in \N\} = \{b + k_l p_l \mid k \in \N\}$, and we are done. If $\Lmc_i = \{b + kp \mid k \in N\}$, then $\Lmc_i + \{k_l p_l \mid k_l \in \N\} = \{b\} + \{kp + k_lp_l \mid k, k_l \in \N\}$. By point \e{(i)}, the latter is a finite union of simple sets $ \bigcup_{j = 1}^r\{b_r+k_rp_r\mid k_r\in\N\}$ so $\Lmc_i + \{k_l p_l \mid k_l \in \N\} = \{b\} +\bigcup_{j = 1}^r\{b_r+k_rp_r\mid k_r\in\N\}=\bigcup_{j = 1}^r\{b+b_r+k_rp_r\mid k_r\in\N\}$. Hence, in both cases, $S$ is a finite union of simple sets.
\end{proof}

\begin{appproposition}\label{prop:app:def-equivalence}
	A \TELn-TBox is ultimately periodic w.r.t. quasimodels iff it is ultimately periodic w.r.t. concept inclusions.
\end{appproposition}
\begin{proof}
	$(\Rightarrow)$ Let \T be ultimately periodic w.r.t. quasimodels and fix $A, B \in \CN(\T)$. Let $m_P, p_P, m_F, p_F$ be such that \eqref{math:app:prelims:up:past}-\eqref{math:app:prelims:up:future} hold for $\pi_A$. By Lemma \ref{lm:app:qm}, $\{n \in \Z \mid \T \mdl A \sqs \Next^n B\} = \{n \in \Z \mid B \in \pi_A(n)\}$. We represent the latter set as a union of the following linear sets:
	\begin{align*}
		&\Lmc_n = \{n\} &&\text{ for } \begin{array}{l}
			n \in (-m_P, m_F)\\ 
			B \in \pi_a(n)
		\end{array}\\
		&\Lmc'_n = \{n - k p_P \mid k \in \N\} &&\text{ for } \begin{array}{l}
			n\in (-m_P - p_P, -m_P]\\ 
			B \in \pi_a(n)
		\end{array}\\
		&\Lmc''_n = \{n + k p_F \mid k \in \N\} &&\text{ for } \begin{array}{l}
			n \in [m_F, m_F + p_F)\\ 
			B \in \pi_a(n)
		\end{array}
	\end{align*}
	Thus, $\{n \in \Z \mid \T \mdl A \sqs \Next^n B\}$ is semilinear.
	
	\noindent
	$(\Leftarrow)$
	Assume that  \T is ultimately periodic w.r.t. concept inclusions. We fix an $A \in \CN(\T)$. By the assumption and Lemma~\ref{lm:app:qm}, for each $B \in \CN(\T)$ the set $\{n \in \Z \mid B \in \pi_A(n)\} = \{n \in \Z \mid \T \mdl A \sqs \Next^n B\}$ is semilinear. By Lemma \ref{lm:app:simple} it is equal to $\Lmc_1 \cup \dots \cup \Lmc_m$, where $\Lmc_i$ are some simple sets. Fix such a representation for each $B \in \CN(\T)$, and let $\Lmc_1, \dots, \Lmc_s$ be all simple sets that appear in these representations, with $\Lmc_i = \{b_i + kp_i \mid k \in \N\}$.
	\begin{align*}
		&m_P, m_F = \max_{1 \leqslant i \leqslant s} b_i &&p_P, p_F = \prod_{i = 1}^{s} p_i
	\end{align*}
		
	It is easy to see that the conditions \eqref{math:app:prelims:up:past}-\eqref{math:app:prelims:up:future} hold for $\pi_A$ with these $m_P, p_P, m_F, p_F$.
\end{proof}

\section{Proof of Proposition~\ref{prop:derivationTEL}}\label{app:prelim}

We show Proposition~\ref{prop:derivationTEL} using the canonical model $\Jmf_{(\T,\A)}$ of the \TELn KB $(\T,\A)$, defined in a similar way as in the temporal DL-Lite case \cite{DBLP:conf/ijcai/ArtaleKWZ13,DBLP:conf/ijcai/ArtaleKKRWZ15}. 
Let $\Delta^{\Jmf_{(\T,\A)}}=\NI\cup\NN$ (recall that $\NN$ is a set of named nulls disjoint from $\IN$). Following \citet{DBLP:conf/ijcai/ArtaleKWZ13}, we represent $\Jmf_{(\T,\A)}$ as a (potentially infinite) set $\Jmc$ of atoms built from $\RN$, $\CN$ and $\NI\cup\NN$ such that for every $n\in\Z$, $d\in A^{\Jmf_{(\T,\A)},n}$ iff $A(d,n)\in\Jmc$ and $(d,e)\in r^{\Jmf_{(\T,\A)},n}$ iff $r(d,e,n)\in\Jmc$. We define $\Jmc=\bigcup_{i\geq 0}\Jmc_i$ where $\Jmc_0=\A$ and $\Jmc_{i+1}$ is built from $\Jmc_i$ by applying a rule of one the following forms, assuming that the rule application is fair (i.e., if a rule can be applied, it is eventually applied):
\begin{enumerate}[label=(\roman*)]
\item if $r(a, b, n)\in\Jmc_i$, $r \in \RNrig$ and there is $k\in\Z$ such that $r(a, b, k)\notin\Jmc_i$, then $\Jmc_{i+1}=\Jmc_i\cup\{r(a, b, k) \mid k \in \Z \}$;

\item if $A(a, n)\in\Jmc_i$, $A \sqs \Next^k B\in\Tmc$ and $B(a, n + k)\notin\Jmc_i$, then $\Jmc_{i+1}=\Jmc_i\cup\{B(a, n + k)\}$;

\item if $A(a, n), A'(a, n)\in\Jmc_i$ , $ A \sqcap A' \sqs B\in\T$ and $B(a, n)\notin\Jmc_i$, then $\Jmc_{i+1}=\Jmc_i\cup\{B(a, n)\}$;

\item if $r(a, b, n),A(b, n)\in\Jmc_i$, $ \exists r.A \sqs B\in\T$ and $B(a, n)\notin\Jmc_i$, then $\Jmc_{i+1}=\Jmc_i\cup\{B(a, n)\}$;

\item if $A(a, n)\in\Jmc_i$, $A \sqs \exists r. B\in\T$ and there is no $b\in\NN$ such that $r(a, b, n), B(b, n)\in\Jmc_i$, then $\Jmc_{i+1}=\Jmc_i\cup\{ r(a, b, n), B(b, n)\}$ for some $b\in\NN$ which does not occur in $\Jmc_i$.
\end{enumerate}	
\begin{applemma}\label{canonicalmodel}
$\Jmf_{(\T,\A)}\models(\T,\A)$ and for every $\Jmf\models(\T,\A)$, there is a homomorphism $h:\NI\cup\NN\mapsto\Delta^\Jmf$ from $\Jmf_{(\T,\A)}$ to $\Jmf$. 
\end{applemma}
\begin{proof}
It is easy to check that $\Jmf_{(\T,\A)}\models(\T,\A)$: $\Jmf_{(\T,\A)}$ is a model of all facts in \A by construction of $\Jmc_0$, and if $\Jmf_{(\T,\A)}$ was not a model of some concept inclusion of \T, or if there was a rigid role name whose interpretation changed over time, this would imply that a rule of one the forms (i)--(v) is applicable in $\Jmc$, contradicting the definition of $\Jmc$. 

Let $\Jmf$ be a model of $(\T,\A)$. We show how to inductively construct a homomorphism $h$ from $\Jmf_{(\T,\A)}$ to $\Jmf$, i.e., a homomorphism $h$ from $\Jmc$ to the set of atoms $\Imc$ such that for every $n\in\Z$, $d\in A^{\Jmf,n}$ iff $A(d,n)\in\Imc$ and $(d,e)\in r^{\Jmf,n}$ iff $r(d,e,n)\in\Imc$.  
Let $h_0:\NI\mapsto\Delta^\Jmf$ be the identity (recall that $\IN\subseteq\Delta^\Jmf$ by the standard name assumption). Since $\Jmf\models \A$, $\A\subseteq\Imc$ so $h_0$ is a homomorphism from $\Jmc_0=\A$ to $\Imc$. Assume that we have built a homomorphism $h_i: \NI\cup\{e\mid e\in\NN, e\text{ occurs in }\Jmc_i\}\mapsto \Delta^\Jmf$ from $\Jmc_i$ to $\Imc$ and consider $\Jmc_{i+1}$. We distinguish two cases:
\begin{itemize}
\item If $\Jmc_{i+1}$ has been obtained from $\Jmc_i$ by applying a rule of one of the forms (i)--(iv), let $h_{i+1}=h_i$. It is easy to verify that in any case, $h_{i+1}$ is a homomorphism from $\Jmc_{i+1}$ to $\Imc$. This follows from the facts that $h_i$ is a homomorphism from $\Jmc_i$ to $\Imc$ and that $\Jmf\models \T$ and respects rigid roles. 
\item Otherwise, $\Jmc_{i+1}$ has been obtained from $\Jmc_i$ by applying a rule of form (v): there are $A(a, n)\in\Jmc_i$ and $A \sqs \exists r. B\in\T$, and $\Jmc_{i+1}=\Jmc_i\cup\{ r(a, b, n), B(b, n)\}$ for some $b\in\NN$ which does not occur in $\Jmc_i$. Since $h_i$ is a homomorphism from $\Jmc_{i}$ to $\Imc$, $A(h_i(a), n)\in\Imc$. Hence, since $\Jmf\models A \sqs \exists r. B$, there is $d\in\Delta^\Jmf$ such that $r(h_i(a), d, n)$ and $B(d, n)$ are in $\Imc$. We define $h_{i+1}: \NI\cup\{e\mid e\in\NN, e\text{ occurs in }\Jmc_{i+1}\}\mapsto \Delta^\Jmf$ by $h_{i+1}(x)=h_i(x)$ for every $x\in \NI\cup\{e\mid e\in\NN, e\text{ occurs in }\Jmc_{i}\}$ and $h_{i+1}(b)=d$. It is easy to check that $h_{i+1}$ is a homomorphism from $\Jmc_{i+1}$ to $\Imc$.
\end{itemize}
We obtain a homomorphism  $h:\NI\cup\NN\mapsto\Delta^\Jmf$ from $\Jmf_{(\T,\A)}$ to $\Jmf$ by setting $h=\bigcup_{i\geq 0} h_i$ and extending $h$ to the nulls that do not occur in $\Jmc$ by mapping them to any element of $\Delta^\Jmf$.
\end{proof}

Proposition~\ref{prop:derivationTEL} then follows from the next two lemmas.

\begin{applemma}\label{lem:derivationTEL1}
For every $A\in\NC$, $a\in\NI$, and $n\in\Z$, 
$(\T,\A)\models A(a,n)$ iff $(\T,\A)\vdash A(a,n)$.
\end{applemma}
\begin{proof}
Since $\Jmc_0=\A$ and rules of form (i)--(v) correspond exactly to derivation rules of form \eqref{math:inference-rule:rigid}--\eqref{math:inference-rule:existential}, it is easy to see that $(\T,\A)\vdash A(a,n)$ iff $A(a,n)\in\Jmc$, i.e., $(\T,\A)\vdash A(a,n)$ iff $a\in A^{\Jmf_{(\T,\A)},n}$. 
Since by Lemma~\ref{canonicalmodel}, $a\in A^{\Jmf_{(\T,\A)},n}$ implies that $a\in A^{\Jmf,n}$ for every $\Jmf\models (\T,\A)$, the result follows.
\end{proof}

\begin{applemma}\label{lem:derivationTEL2}
For every $A,B\in\NC$, $a\in\NI$, and $n,k\in\Z$, 
$\T\models A\sqsubseteq\Next^n B$ iff $(\T,\{A(a,k)\})\vdash B(a,k+n)$.
\end{applemma}
\begin{proof}
By Lemma~\ref{lem:derivationTEL1}, $(\T,\{A(a,k)\})\vdash B(a,k+n)$ iff $(\T,\{A(a,k)\})\models B(a,k+n)$. We show that $(\T,\{A(a,k)\})\models B(a,k+n)$ iff $\T\models A\sqsubseteq\Next^n B$. 

\noindent($\Leftarrow$) If $\T\models A\sqsubseteq\Next^n B$, every model $\Jmf$ of $(\T,\{A(a,k)\})$ is such that $\Jmf\models A(a,k)$ and $\Jmf\models A\sqsubseteq\Next^n B$, hence $\Jmf\models B(a,k+n)$.

\noindent($\Rightarrow$) Assume that $\T\not\models A\sqsubseteq\Next^n B$: there exists a model $\Jmf=(\Delta^{\Jmf}, (\Imc_i)_{i\in\Z})$ of \T with $e\in\Delta^\Jmf$ and $j\in\Z$ such that $e\in A^{\Jmf,j}$ and $e\notin B^{\Jmf,j+n}$. 
Let $\Jmf'=(\Delta^{\Jmf}, (\Imc'_i)_{i\in\Z})$ be the interpretation obtained from $\Jmf$ by switching $e$ and $a$ in all concept and role interpretations (note that since $a\in\NI$, $a\in\Delta^{\Jmf}$), and let $\Jmf''=(\Delta^{\Jmf}, (\Imc''_i)_{i\in\Z})$ where $\Imc''_i=\Imc'_{i+k-j}$ for every $i\in\Z$. 
It is easy to see that $\Jmf''\models (\T,\{A(a,k)\})$ while $\Jmf''\not\models B(a,k+n)$. 
Hence, $(\T,\{A(a,k)\})\not\models B(a,k+n)$. 
\end{proof}

% !TEX root =  ../ms.tex

\section{Additional notation and conventions}

\paragraph{Derivations} Recall that a derivation witnessing $(\T, \A) \vdash A(a, n)$ is a sequence $(\F_0, \dots, \F_m)$ such that $\Fmc_0 = \A \cup \T$, $A(a, n) \in \Fmc_m$ and $\Fmc_i$ is obtained from $\Fmc_{i-1}$ by applying a rule of the form \eqref{math:inference-rule:rigid}--\eqref{math:inference-rule:existential}. It will be convenient to represent such $(\F_0, \dots, \F_m)$ as 
\begin{align*}
	\Fmc_0 \xto{f_1} \dots \xto{f_m} \Fmc_m
\end{align*}
where each \emph{rule application} is represented by a pair of sets of formulas $f_i=(\dom(f_i),\rn(f_i))$ where $\dom(f_i) \sbs \Fmc_{i - 1}$ is the \emph{premise} of $f_i$ 
 and $\rn(f_i) = \Fmc_{i} \setminus \Fmc_{i - 1}$ is its \emph{conclusion}, which match, respectively, the left and the right sides of the rule applied to get $\Fmc_i$ from $\Fmc_{i - 1}$. We say that $f_i$ \emph{uses} the formulas in $\dom(f_i)$ and \emph{produces} those in $\rn(f_i)$. 
Similarly, given $G = (N, \Sigma, R)$, a derivation witnessing $G \vdash X(w)$ can be represented as
\begin{align*}
	\Gmc_0 \xto{g_1} \dots \xto{g_m} \Gmc_m
\end{align*}
where $\Gmc_0 = \{c(c) \mid c \in \Sigma\}$, $X(w) \in \Gmc_m$, and each rule application $g_i$ refers to one of the rules defined by    \eqref{math:grammar:semantics:deduction-rules}.

Moreover, \emph{we extend the notion of derivation to derivations of formulas of the form $A(b,n)$ where $b\in\NN$} (instead of $\NI$) from a set of formulas $\Fmc_0$ in which $b$ occurs. Hence, one can write, e.g., $(\Tmc,\{(A(b,k)\})\vdash B(b,n+k)$ for $b\in\NN$.

\paragraph{Extension and use of Proposition \ref{prop:derivationTEL}} In the proofs given in the next sections, Proposition \ref{prop:derivationTEL} allows us to equivalently write $\T\models A\sqsubseteq\Next^n B$, $(\T,\{A(a,0)\})\models B(a,n)$, or $(\T,\{A(a,0)\})\vdash B(a,n)$. 
Moreover, one can use $a\in\NI\cup\NN$ in the last formula since it is easy to see by considering derivations that $(\T,\{A(a,0)\})\vdash B(a,n)$ for $a\in\NI$ iff $(\T,\{A(b,0)\})\vdash B(b,n)$ for $b\in\NN$.

\section{Proofs for Section \ref{sec:conjunctive-grammars}}\label{app:conjunctive-grammars}

\subsection{Proof of Theorem~\ref{thm:conjunctive-grammars:tbox-to-grammar}}
We start with some lemmas. Recall that given a \TELn-TBox \T, \Trig is the TBox defined right after Theorem~\ref{thm:conjunctive-grammars:tbox-to-grammar} in Section~\ref{sec:conjunctive-grammars}.

\begin{applemma}\ \label{lm:app:local-stays-local}
	The following statements hold.
	\begin{enumerate}[label=(\roman*)]
		\item Given a derivation $\Fmc_0 \xto{f_1} \dots \xto{f_m} \Fmc_m$ witnessing $(\T, \{A(a, 0)\}) \vdash B(a, n)$, if $r(b_1, b_2, k) \in \Fmc_i$ and $r(b_1, b_2, k') \in \Fmc_j$, for some $i, j \in \{1, \dots, m\}$ and $r \in \RNloc$, then $k' = k$. \label{part:role}
		
		\item Given a derivation $\Fmc_0 \xto{f_1} \dots \xto{f_m} \Fmc_m$ witnessing $(\Trig, \{A(a, 0)\}) \vdash B(a, n)$, if $C_r(b, k) \in \Fmc_i$ and $C_r(b, k') \in \Fmc_j$, for some $i, j \in \{1, \dots, m\}$, then $k' = k$. \label{part:concept}
	\end{enumerate}
\end{applemma}
\begin{proof}
	\textit{(i)} Since $\Fmc_0 = \{A(a, 0)\}\cup\T$, by the form of the derivation rules (cf.~\eqref{math:inference-rule:rigid}--\eqref{math:inference-rule:existential}), $b_2 \in \NN$. Suppose $k \neq k'$ and let $i, j$ be the least indexes such that $r(b_1, b_2, k) \in \Fmc_i$ and $r(b_1, b_2, k') \in \Fmc_j$. Since a rule application produces several role facts only if the rule is of form~\eqref{math:inference-rule:rigid}, and $r \in \RNloc$, it must be the case that $i \neq j$. Suppose $i < j$ and observe that $r(b_1, b_2, k')$ is produced by the application $f_j$ of a rule of form \eqref{math:inference-rule:existential}. Then $b_2$ should be fresh, but it already appears in $\Fmc_i$. Hence $k=k'$.
	
\noindent\textit{(ii)}	The proof is analogous to that of \textit{(i)}, using the fact that $C_r$ only occurs in the right-hand side of concept inclusions of $\Trig$ in concept inclusions of the form $A\sqsubseteq\exists r.C_r$, so that a fact of the form $C_r(b, \ell)$ can only be produced by an application of a rule of form \eqref{math:inference-rule:existential} which introduces $b$ as a fresh null.
\end{proof}

\conjunctivegrammarslocalroleremoval*
\begin{proof}
$(\Rightarrow)$ Let $\Fmc_0 \xto{f_1} \dots \xto{f_m} \Fmc_m$ be a derivation witnessing $(\T, \{A(a, 0)\}) \vdash B(a, n)$. We build a derivation witnessing $(\Trig, \{A(a, 0)\}) \vdash B(a, n)$. Let $(h_1,\dots, h_p)$ be the sequence of rule applications obtained from $(f_1,\dots,f_m)$ by applying the following steps. 
\begin{enumerate}
\item Substitute every application of a rule of form \eqref{math:inference-rule:existential} with $r \in \RNloc$, $f_i=(\{C\sqsubseteq\exists r.D,C(e,\ell)\},\{r(e,d,\ell), D(d,\ell)\})$, by consecutive applications of rules of form \eqref{math:inference-rule:existential} and \eqref{math:inference-rule:shift}, $f'_i,f''_i$, where $f'_i=(\{C\sqsubseteq\exists r.D_r,C(e,\ell)\},\{r(e,d,\ell), D_r(d,\ell)\})$ and $f''_i=(\{D_r\sqsubseteq D, D_r(e,\ell)\},\{D(d,\ell)\})$. 
\label{point-from-T-to-Trig}	
		
\item For every application of rule of form \eqref{math:inference-rule:return}  with $r \in \RNloc$, $f_i=(\{r(e,d,\ell), C(d,\ell),\exists r.C\sqsubseteq D\}, \{D(e,\ell)\})$, by Lemma~\ref{lm:app:local-stays-local}, there is no $\ell'\neq\ell$ such that $r(e,d,\ell')$ belongs to any $\Fmc_j$, so since $r(e,d,\ell)$ has been produced by the application of a rule of form \eqref{math:inference-rule:existential}, $f_j$, with $j<i$, there exists $E_r$ such that $E_r(d,\ell)$ is in the conclusion of $f'_j$ defined in point \ref{point-from-T-to-Trig}. 
Substitute $f_i$ by consecutive applications of rules of form \eqref{math:inference-rule:conjunction} and \eqref{math:inference-rule:return}, $f'_i,f''_i$, where $f'_i=(\{C(d,\ell),E_r(d,\ell), C\sqcap E_r\sqsubseteq C'_r\}, \{C'_r(d,\ell)\})$ and $f''_i=(\{r(e,d,\ell), C'_r(d,\ell),\exists r.C'_r\sqsubseteq D\},\{D(e,\ell)\})$.
			
\item Substitute every occurrence of $r$ in the resulting sequence of rule applications by the $r' \in \RNrig$ used in \Trig. 
\end{enumerate}

Since all concept inclusions used in the premises of $h_1,\dots, h_p$ belongs to \Trig by construction, we indeed obtain a derivation $\Fmc'_0 \xto{h_1} \dots \xto{h_p} \Fmc'_p$ witnessing $(\Trig, \{A(a, 0)\}) \vdash B(a, n)$ by setting $\Fmc'_0=\{A(a, 0)\}\cup\Trig$ and $\Fmc'_i=\Fmc'_{i-1}\cup\rn(h_i)$.  
\smallskip
		
\noindent$(\Leftarrow)$ Suppose $\Fmc_0 \xto{f_1} \dots \xto{f_m} \Fmc_m$ is a derivation witnessing $(\Trig, \{A(a, 0)\}) \vdash B(a, n)$. 
We build a derivation witnessing $(\T, \{A(a, 0)\}) \vdash B(a, n)$. 
Let $(h_1,\dots, h_p)$ be the sequence of rule applications obtained from $(f_1,\dots,f_m)$ by applying the following steps. 
\begin{enumerate}
\item Restore each $r'$ of $\Trig$ to the original $r \in \RNloc$ and omit every application of a rule of form \eqref{math:inference-rule:rigid} with $r'$.
		
\item Omit every application of a rule of form \eqref{math:inference-rule:shift} that uses $D_r\sqsubseteq D$, and substitute every application of a rule of form \eqref{math:inference-rule:existential} that uses $C\sqsubseteq \exists r.D_r$, $f_i=(\{C(e,\ell), C\sqsubseteq \exists r.D_r\},\{r(e,d,\ell), D_r(d,\ell)\})$, with one using $C \sqs \exists r . D$ instead: $f'_i=(\{C(e,\ell), C\sqsubseteq \exists r.D\},\{r(e,d,\ell), D(d,\ell)\})$. \label{point-from-Trig-to-T}
			
\item Omit every application of a rule of form \eqref{math:inference-rule:conjunction} that uses a concept inclusion of the form $C\sqcap D_r\sqsubseteq C'_r$.

\item For every application of a rule of form \eqref{math:inference-rule:return} that uses $\exists r. C'_r\sqsubseteq D$, $f_i=(\{r(e,d,\ell), C'_r(d,\ell), \exists r. C'_r\sqsubseteq D\},\{D(e,\ell)\})$, since $C'_r(d,\ell)\in\Fmc_{i-1}$, then 
\begin{enumerate}
\item $C(d,\ell)\in\Fmc_{i-1}$ and there exists $E_r$ such that $E_r(d,\ell)\in\Fmc_{i-1}$ (since $C'_r$ occurs in the right-hand side of concept inclusions in \Trig only in concept inclusions of the form $C\sqcap E_r\sqsubseteq C'_r$), and 

\item by Lemma \ref{lm:app:local-stays-local}, there is no $\ell'\neq\ell$ such that $E_r(d,\ell')$ belongs to any $\Fmc_j$, so since $E_r(d,\ell)$ has been produced by the application of a rule of form \eqref{math:inference-rule:existential}, $f_j$, with $j<i-1$, $r(e,d,\ell)\in\Fmc_{i-1}$ has been produced by $f'_j$ defined in point \ref{point-from-Trig-to-T}\footnote{This ensures that $r(e,d,\ell)$ is not ``lost'' in the derivation when we omit every application of a rule of form \eqref{math:inference-rule:rigid} with $r'$.} (note that since $d$ has been introduced as a fresh element by $f'_j$, there cannot be any $r(e',d,\ell)\in\Fmc_{i-1}$ with $e\neq e'$). 
\end{enumerate}
Substitute $f_i$ with $f'_i=(\{r(e,d,\ell), C(d,\ell), \exists r. C\sqsubseteq D\},\{D(e,\ell)\})$. 
\end{enumerate}

Once again, we can check that all concept inclusions used in the premises of $h_1,\dots, h_p$ belongs to \T so that we indeed obtain a derivation $\Fmc'_0 \xto{h_1} \dots \xto{h_p} \Fmc'_p$ witnessing $(\T, \{A(a, 0)\}) \vdash B(a, n)$ by setting $\Fmc'_0=\{A(a, 0)\}\cup\T$ and $\Fmc'_i=\Fmc'_{i-1}\cup\rn(h_i)$. 
\end{proof}

	Now we are ready to prove Theorem~\ref{thm:conjunctive-grammars:tbox-to-grammar}. 
	
	\TBoxesToGrammars*
\begin{proof}
Let $G_\T$ be the grammar defined from \T in Definition~\ref{def:conjuntive-grammars:GT}. 
	We show that for every $A,B\in\CN(\Trig)$, for every $a\in\NI\cup\NN$, and $n\in\N$, $G_\T \vdash \Nmc_{AB}(c^n)$ iff $(\Trig, \{A(a, 0)\}) \vdash B(a, n)$. The result will follow by Lemma~\ref{lm:conjunctive-grammars:local-role-removal} and the fact that $\NC(\T)\subseteq\NC(\Trig)$. 
	\smallskip
	
	\noindent
	$(\Leftarrow)$ 
	We show by induction on $m$ that for every $A,B\in\CN(\Trig)$, $a\in\NI\cup\NN$, and $n\in\N$, if there exists a derivation witnessing $(\Trig, \{A(a, 0)\}) \vdash B(a, n)$ of length at most $m$, then $G_\T \vdash \Nmc_{AB}(c^n)$. 
	
	In the base case, $m=0$, the derivation consists only of $\F_0 = \Trig \cup \{A(a,0)\}$, so $B(a, n) \in \F_0$ implies that $A=B$ and $n = 0$, and by \eqref{math:grammar:from-tbox:epsilon}, $\Nmc_{AA}\rightarrow \varepsilon$ is a rule of $G_\T$  so $G_\T \vdash \Nmc_{AB}(\varepsilon)$. 
	
	Induction step: assume that the property holds for $m-1$ and let $\Fmc_0 \xto{f_1} \dots \xto{f_m} \Fmc_m$ be a derivation witnessing $(\Trig, \{A(a,0)\}) \vdash B(a, n)$. If $B(a, n)\in\Fmc_{m-1}$, the result follows by induction hypothesis. Otherwise, there are three possible cases for the last rule application $f_m$ that produces $B(a, n)$:
	\begin{itemize}
		\item $\dom(f_m)=\{A'(a,n - k), A' \sqs \Next^k B \}$ for some  $k\in\{0,\dots, n\}$, and $f_m$ is the application of a rule of form \eqref{math:inference-rule:shift}. 
		Since $(\F_0, \dots, \F_{m - 1})$ is a derivation witnessing $(\Trig,\{A(a,0)\})\vdash A'(a, n - k)$, by induction hypothesis, $G_\T \vdash \Nmc_{AA'}(c^{n - k})$. 
		Moreover, since $A' \sqs \Next^k B\in\Trig$, by \eqref{math:grammar:from-tbox:epsilon} or \eqref{math:grammar:from-tbox:shift} depending on $k$, $\Nmc_{A'B}\rightarrow c^k$ is a rule of $G_\T$. Hence $G_\T \vdash \Nmc_{A'B}(c^{k})$. Then, since by \eqref{math:grammar:from-tbox:middle} $\Nmc_{AB}\rightarrow \Nmc_{AA'}\Nmc_{A'B}$ is a rule of $G_\T$, we obtain $G_\T \vdash \Nmc_{AB}(c^{n})$.
		
		\item $\dom(f_m)=\{A'(a,n),A''(a,n), A' \sqcap A''\sqs B \}$ and $f_m$ is the application of a rule of form  \eqref{math:inference-rule:conjunction}. Since $A'(a,n)$ and $A''(a,n)$ are in $\Fmc_{m-1 }$, $(\F_0, \dots, \F_{m - 1})$ is a derivation witnessing $(\Trig, \{A(a,0)\})\vdash A'(a, n)$ and $(\Trig, \{A(a,0)\})\vdash A''(a, n)$. 
		By the induction hypothesis, $G_\T \vdash \Nmc_{AA'}(c^n)$ and $G_\T \vdash \Nmc_{AA''}(c^n)$, and since $A' \sqcap A'' \sqs B \in \Trig$, by \eqref{math:grammar:from-tbox:conjunction}, $\Nmc_{AB}\rightarrow \Nmc_{AA'} \& \Nmc_{AA''}$ is a rule of $G_\T$ so $G_\T \vdash \Nmc_{AB}(c^n)$. 
		
		\item $\dom(f_m)=\{r(a,b, n),A'(b,n), \exists r.A' \sqs B \}$ and $f_m$ is the application of a rule of form \eqref{math:inference-rule:return}. By the form of the derivation rules, $r(a,b, n)\in\Fmc_{m-1}$ has been produced by the application of a rule of form \eqref{math:inference-rule:rigid} or \eqref{math:inference-rule:existential}. Hence, there must be an index $i < m - 1$ such that for some $A'', B'$ and $k \in \N$,  
		$\dom(f_i) = \{A'' \sqs \exists r . B', A''(a, k)\}$, $\rn(f_i) = \{r(a, b, k), B'(b, k)\}$. Thus $A''(a, k) \in \F_{i-1}$, and $(\Fmc_0,\dots,\Fmc_{i-1})$ is a derivation witnessing $(\Trig,\{A(a,0)\})\vdash A''(a, k)$.  By the induction hypothesis, we obtain $G_\T \vdash \Nmc_{AA''}(c^{k})$. Moreover, one can extract from $(\F_i, \dots, \F_{m - 1})$ a derivation of length at most $m - 1$ witnessing $(\Trig, \{B'(b, k)\}) \vdash A'(b, n)$. Since \T is a \TELnf-TBox, $k \leqslant n$. 
		We get a derivation length at most $m - 1$ witnessing $(\Trig, \{B'(b ,0)\}) \vdash A'(b, n-k)$ by shifting all timestamps in this derivation, and thus, by the induction hypothesis, $G_\T \vdash \Nmc_{B'A'}(c^{n-k})$. Since $A'' \sqs \exists r . B'$ and $\exists r.A' \sqs B$ are in \Trig, by \eqref{math:grammar:from-tbox:down-and-up}, $\Nmc_{A''B}\rightarrow \Nmc_{B'A'}$ is a rule of $G_\T$, so from $G_\T \vdash \Nmc_{B'A'}(c^{n-k})$, we get $G_\T \vdash \Nmc_{A''B}(c^{n-k})$. We combine this with $G_\T \vdash \Nmc_{AA''}(c^{k})$ and the fact that by \eqref{math:grammar:from-tbox:middle}, $\Nmc_{AB}\rightarrow\Nmc_{AA''}\Nmc_{A''B}$ is a rule of $G_\T$ to establish $G_\T \vdash \Nmc_{AB}(c^n)$.
	\end{itemize}
	
	\noindent
	$(\Rightarrow)$ We show by induction on $m$ that for every $A,B\in\CN(\Trig)$, $a\in\NI\cup\NN$, and $n\in\N$, if there exists a derivation witnessing $G_\T \vdash \Nmc_{AB}(c^n)$ of length at most $m$, then $(\Trig, \{A(a, 0)\}) \vdash B(a, n)$. 
	
	Since $\Gmc_0=\{c(c)\}$, the base case is $m=1$, when the derivation consists of $\Gmc_0 \xto{g_1}\Gmc_1$. 
	There are two possible cases  for the rule application $g_1$ that produces $\Nmc_{AB}(c^n)$ from $c(c)$:
	\begin{itemize}
		\item $g_1$ is the application of the rule from  \eqref{math:grammar:semantics:deduction-rules} that corresponds to $\Nmc_{AB}\rightarrow \varepsilon$. In this case, $n=0$ and by \eqref{math:grammar:from-tbox:epsilon}, $A\sqsubseteq B\in\Trig$ or $A=B$. In both cases, for every $a\in\NI\cup\NN$, $(\Trig, \{A(a, 0)\}) \vdash B(a, 0)$, i.e. $(\Trig, \{A(a, 0)\}) \vdash B(a, n)$. 
		\item $g_1$ is the application of the rule from  \eqref{math:grammar:semantics:deduction-rules} that corresponds to $\Nmc_{AB}\rightarrow c^n$. In this case, by \eqref{math:grammar:from-tbox:shift}, $A\sqsubseteq \Next^n B\in\Trig$. Hence, for every $a\in\NI\cup\NN$, $(\Trig, \{A(a, 0)\}) \vdash B(a, n)$. 
	\end{itemize}
	Induction step: assume that the property holds for $m-1$ and let $\Gmc_0 \xto{g_1} \dots \xto{g_m} \Gmc_m$ be a derivation witnessing $G_\T \vdash \Nmc_{AB}(c^n)$. If $\Nmc_{AB}(c^n)\in\Gmc_{m-1}$, the result follows by induction hypothesis. Otherwise, there are three possible cases for the last rule application $g_m$ that produces $\Nmc_{AB}(c^n)$:
	\begin{itemize}
		\item $\dom(g_m)=\{ \Nmc_{AC}(c^n), \Nmc_{AD}(c^n)\}$, and $g_m$ is the application of the rule from  \eqref{math:grammar:semantics:deduction-rules} that corresponds to $\Nmc_{AB}\rightarrow \Nmc_{AC}\&\Nmc_{AD}$. By \eqref{math:grammar:from-tbox:conjunction}, $C\sqcap D\sqsubseteq B\in\Trig$. Then $(\Gmc_0,\dots,\Gmc_{m-1})$ is a derivation witnessing $G \vdash \Nmc_{AC}(c^n)$ and $G \vdash \Nmc_{AD}(c^n)$, so by the induction hypothesis, for every $a\in\NI\cup\NN$, $(\Trig, \{A(a, 0)\}) \vdash C(a, n)$ and $(\Trig, \{A(a, 0)\}) \vdash D(a, n)$. Hence $(\Trig, \{A(a, 0)\}) \vdash B(a, n)$.
		
		\item $\dom(g_m)=\{\Nmc_{CD}(c^n)\}$, and $g_m$ is the application of the rule from  \eqref{math:grammar:semantics:deduction-rules} that corresponds to  $\Nmc_{AB}\rightarrow \Nmc_{CD}$. By \eqref{math:grammar:from-tbox:down-and-up}, there exists $r$ such that $A\sqsubseteq \exists r.C$ and $\exists r.D\sqsubseteq B$ are in \Trig. Then $(\Gmc_0,\dots,\Gmc_{m-1})$ is a derivation witnessing $G \vdash \Nmc_{CD}(c^n)$ so by the induction hypothesis, for every $a\in\NI\cup\NN$, $(\Trig, \{C(a, 0)\}) \vdash D(a, n)$. Let $\Fmc_0 \xto{f_1} \dots \xto{f_p} \Fmc_p$ be obtained from that derivation by substituting $a$ everywhere with a fresh $b \in \NN$. Then the following is a derivation witnessing $(\Trig, \{A(a, 0)\}) \vdash B(a, n)$:
		\begin{align*}
			\F \xto{f_0} \F'_0 \xto{f_1} \dots \xto{f_p} \F'_p \xto{f_{p + 1}} \F_{p + 1} \xto{f_{p + 2}} \F_{p + 2}
		\end{align*}
		where
		\begin{itemize}
			\item $\F = \{A(a, 0)\} \cup \Trig$;
			
			\item $f_0$ is an application of the rule of form \eqref{math:inference-rule:existential} with $\dom(f_0) = \{A(a, 0), A \sqs \exists r.C\}$ and $\rn(f_0) = \{r(a, b, 0), C(b, 0)\}$;
			\item $\F'_i=\F_i\cup\{A(a, 0),r(a, b, 0)\}$ for $0\leqslant i\leqslant p$;
			\item $f_{p+ 1}$ is an application of the rule of form \eqref{math:inference-rule:rigid} with $\dom(f_{p+1}) = \{r(a, b, 0)\}$ and $\rn(f_{p+ 1}) = \{r(a, b, n)\}$; %$\rn(f_{p+ 1}) = \{r(a, b, k) \mid k\in\Z\}$; 
			
			\item $f_{p + 2}$ is an application of the rule of form \eqref{math:inference-rule:return} with $\dom(f_{p+2}) = \{r(a, b, n), D(b, n), \exists r . D \sqs B\}$ and $\rn(f_{p + 2}) = \{B(a, n)\}$;
			
			\item $\F_{p + 1}=\F'_p\cup\rn(f_{p+1})$ and $\F_{p + 2}=\F_{p+1}\cup\rn(f_{p+2})$.
		\end{itemize}
		
		\item $\dom(g_m)=\{\Nmc_{AC}(c^{n-k}),\Nmc_{CB}(c^k)\}$ for some $k\in\{0,\dots,n\}$, and $g_m$ is the application of the rule from  \eqref{math:grammar:semantics:deduction-rules} that corresponds to $\Nmc_{AB}\rightarrow \Nmc_{AC}\Nmc_{CB}$. By \eqref{math:grammar:from-tbox:middle} $A,B,C\in\NC(\Trig)$. 
Then $(\Gmc_0,\dots,\Gmc_{m-1})$ is a derivation witnessing $G \vdash \Nmc_{AC}(c^{n - k})$ and $G \vdash \Nmc_{CB}(c^k)$. By the induction hypothesis, for every $a\in\NI\cup\NN$, $(\Trig, \{A(a, 0)\}) \vdash C(a, n - k)$ and $(\Trig, \{C(a, 0)\}) \vdash B(a, k)$, which is equivalent (by Proposition~\ref{prop:derivationTEL}) to $(\Trig, \{C(a, n-k)\}) \vdash B(a, n)$. Hence, $(\Trig, \{A(a, 0)\}) \vdash B(a, n)$.  \qedhere
	\end{itemize}
\end{proof}

	\subsection{Proof of Theorem~\ref{thm:conjunctive-grammars:grammar-to-tbox}}
	Again, we start with somme lemmas. Recall that given a unary conjunctive grammar $G$, the \TELnf-TBox $\T_G$ is defined by Definition~\ref{def:grammar-to-tbox}.
	
	\begin{applemma}\label{lm:app:LmCgMeetPreciseFirstDirection}
	For $\iota = i_1 \dots i_k \in \Jmc$, if $\T_G \mdl A \sqs \Next^n C_\iota$ then $n = n_1 + \dots + n_k$, such that $\T_G \mdl A \sqs \Next^{n_j} B_{i_j}$ for $1 \leqslant j \leqslant k$. Moreover, if there exists a derivation for $(\T_G,\{A(a,0)\}) \vdash C_\iota(a,n)$ of length $p$, then for every for $1 \leqslant j \leqslant k$, there exists a derivation for $(\T_G,\{A(a,0)\}) \vdash B_{i_j}(a,n_j)$ of length at most $p-1$.
	\end{applemma}
	\begin{proof}
	
		We show by induction on $k$ that for every $\iota\in\Jmc$ of length $k$, if $\iota=i_1\dots i_k$, for every $n\in\N$, if $\T_G\models A\sqsubseteq\Next^n C_\iota$, then there exist $n_1,\dots, n_k$ such that $n=n_1+\dots+n_k$ and $\T_G\models A\sqsubseteq\Next^{n_j} B_{i_j}$ for $1\leqslant j\leqslant k$, and that if there exists a derivation for $(\T_G,\{A(a,0)\}) \vdash C_\iota(a,n)$ of length $p$, then for every for $1 \leqslant j \leqslant k$, there exists a derivation for $(\T_G,\{A(a,0)\}) \vdash B_{i_j}(a,n_j)$ of length at most $p-1$. 
		
In the base case, $k=1$ and $\iota=i_1$ (hence $n_1=n$). 
If $\T_G\models A\sqsubseteq\Next^n C_{i_1}$, i.e., $(\T_G, \{A(a, 0)\}) \vdash C_{i_1}(a, n)$, since the only concept inclusion in $\T_G$ with $C_{i_1}$ in the right-hand side is $B_{i_1}\sqsubseteq C_{i_1}$ (cf.~Definition~\ref{def:grammar-to-tbox}), any derivation $\Fmc_0 \xto{f_1} \dots \xto{f_p} \Fmc_p$ witnessing $(\T_G,\{A(a,0)\}) \vdash C_\iota(a,n)$ must contain a derivation (of length at most $p-1$) witnessing $(\T_G, \{A(a, 0)\}) \vdash B_{i_1}(a, n)$. Hence, $\T_G\models A\sqsubseteq\Next^n B_{i_1}$.

		Induction step: assume that the property holds for $k-1$ and let $\iota = i_1 \jmath\in\Jmc$ be of length $k$. Note that $\jmath=i_2\dots i_{k}\in\Jmc$ and is of length $k-1$.  Suppose $(\T_G, \{A(a, 0)\}) \vdash C_{i_1\jmath}(a, n)$ and let $\Fmc_0 \xto{f_1} \dots \xto{f_p} \Fmc_p$ be a  derivation witnessing it. Necessarily, the rule application that produces $C_{i_1\jmath}(a, n)$ uses the concept inclusion $\exists r_{i_1\jmath}. C_\jmath\sqsubseteq C_{i_1\jmath}$ since it is the only one with $C_{i_1\jmath}$ in the right-hand side by construction of $\T_G$.
		Thus there exists $b$ such that $r_{i_1\jmath}(a, b, n)$ and $ C_\jmath(b, n)$ are in $\F_{p - 1}$. 
		\begin{itemize}
		\item Let $f_l$ be the application of the rule of form \eqref{math:inference-rule:existential} that produced the first fact of the form $r_{i_1\jmath}(a, b, n_1)$ ($n_1 \in \N$ and $l\leqslant p-1$): by construction of $\T_G$, it must be the case that $\dom(f_l) = \{B_{i_1} \sqs \exists r_{i_1\jmath} . A, B_{i_1}(a, n_1)\}$ and $\rn(f_l) = \{r_{i_1\jmath}(a, b, n_1), A(b, n_1)\}$. 
		From $(\F_{l}, \dots, \F_{p - 1})$, one obtains a derivation  of length at most $p-1$ witnessing $(\T_G, \{A(b, 0)\}) \vdash C_\jmath(b, n - n_1)$ 
		by shifting all timestamps by $-n_1$. By construction of $\T_G$ (in particular, because it is a \TELnf-TBox), it must be the case that $n - n_1 \geq 0$. 
Hence, by the induction hypothesis, there exist $n_2,\dots, n_k$ such that $n - n_1 = n_2 + \dots + n_k$, and for $2 \leqslant j \leqslant k$, $\T_G \mdl A \sqs \Next^{n_j} B_{i_j}$ and there exists a derivation of length at most $p-1$ witnessing $(\T_G, \{A(a, 0)\}) \vdash B_{i_j}(a, n_j)$. 
		\item Furthermore, as $B_{i_1}(a, n_1) \in \dom(f_l) \sbs \F_{l - 1}$, $(\F_0,\dots,\F_{l - 1})$ is a derivation of length at most $p-1$ witnessing $(\T_G, \{A(a, 0)\}) \vdash B_{i_1}(a, n_1)$, and $\T_G \mdl A \sqs \Next^{n_1} B_{i_1}$.
		\end{itemize}
		Hence $n=n_1+\dots+n_k$ and for $1 \leqslant j \leqslant k$, $\T_G \mdl A \sqs \Next^{n_j} B_{i_j}$ and there exists a derivation of length at most $p-1$ witnessing $(\T_G, \{A(a, 0)\}) \vdash B_{i_j}(a, n_j)$. 	
\end{proof}

	\LmCgMeet*
	\begin{proof}
	$\Rightarrow$ The `only if' direction is given by Lemma~\ref{lm:app:LmCgMeetPreciseFirstDirection}. 
	
	\noindent$\Leftarrow$ We show by induction on $k$ that for every $\iota\in\Jmc$ of length $k$, if $\iota=i_1\dots i_k$, for every $n\in\N$, if there exist $n_1,\dots, n_k$ such that $n=n_1+\dots+n_k$ and $\T_G\models A\sqsubseteq\Next^{n_j} B_{i_j}$ for $1\leqslant j\leqslant k$, then $\T_G\models A\sqsubseteq\Next^n C_\iota$. 
		
		In the base case, $k=1$ and $\iota=i_1$. If  $(\T_G, \{A(a, 0)\}) \models B_{i_1}(a, n)$, then $(\T_G, \{A(a, 0)\}) \models C_{i_1}(a, n)$ since $B_{i_1}\sqsubseteq C_{i_1}\in\T_G$ by \eqref{math:ci:from-grammar:conclude}.

		Induction step: assume that the property holds for $k-1$ and let $\iota = i_1 \jmath\in\Jmc$ be of length $k$. Note that $\jmath=i_2\dots i_{k}\in\Jmc$ and is of length $k-1$. 		
 		Suppose that $n=n_1+\dots+n_k$ and $\T_G\models A\sqsubseteq\Next^{n_j} B_{i_j}$ for $1\leqslant j\leqslant k$.
		\begin{itemize}
		\item Since $\T_G\models A\sqsubseteq\Next^{n_1} B_{i_1}$, for every $a\in\NI$, $(\T_G, \{A(a,0)\})\vdash B_{i_1}(a,n_1)$. Let $\Fmc_0 \xto{f_1} \dots \xto{f_m} \Fmc_m$ be a derivation witnessing it.
		
		\item Since $i_1\jmath\in\Jmc$, by \eqref{math:ci:from-grammar:down}, $B_{i_1}\sqsubseteq \exists r_{i_1\jmath}. A\in\T_G$. Let $\Fmc_{m+1}=\Fmc_m\cup\rn(f_{m+1})$ where $f_{m+1}=(\{B_{i_1}\sqsubseteq \exists r_{i_1\jmath}. A, B_{i_1}(a,n_1)\},\{r_{i_1\jmath}(a,b,n_1), A(b,n_1)\})$.

		\item  By the induction hypothesis: $\T_G\models A\sqsubseteq\Next^{n'} C_\jmath$ for $n'=n_2+\dots+n_k$.  
		Hence, $(\T_G, \{A(b,0)\})\vdash C_\jmath(b,n')$, and by Proposition~\ref{prop:derivationTEL}, $(\T_G, \{A(b,n_1)\})\vdash C_\jmath(b,n_1+n')$. Let $\Fmc'_0 \xto{h_1} \dots \xto{h_p} \Fmc'_p$ be a derivation witnessing it, $(f_{m+1},\dots,f_{m+p})=(h_1,\dots,h_p)$, and $\Fmc_{m+1}=\Fmc_m\cup\rn(f_{m+1})$,..., $\Fmc_{m+p}=\Fmc_{m+p-1}\cup\rn(f_{m+p})$. It is easy to check that for every $m+1\leqslant j\leqslant m+p$, $\dom(f_j)\subseteq\Fmc_{j-1}$ and that $r_{i_1\jmath}(a,b,n_1)$ and $C_\jmath(b,n_1+n')$ are in $\Fmc_{m+p}$. 
		
		\item Let %$f_{m+p+1}=(\{r_{i_1\jmath}(a,b,n_1)\},\{r_{i_1\jmath}(a,b,k)\mid k\in\Z\})$ 
		 $f_{m+p+1}=(\{r_{i_1\jmath}(a,b,n_1)\},\{r_{i_1\jmath}(a,b,n)\})$
		be the application of the corresponding rule of form \eqref{math:inference-rule:rigid} and $\Fmc_{m+p+1}=\Fmc_{m+p}\cup\rn(f_{m+p+1})$, so that $r_{i_1\jmath}(a,b,n)$ and $C_\jmath(b,n)$ are in $\Fmc_{m+p+1}$.
		
		\item Finally, since $i_1\jmath\in\Jmc$ and $\jmath\in\Jmc$, by \eqref{math:ci:from-grammar:up}, $\exists r_{i_1\jmath}. C_\jmath\sqsubseteq C_{i_1\jmath}\in\T_G$. Let $f_{m+p+2}=(\{r_{i_1\jmath}(a,b,n),C_\jmath(b,n),\exists r_{i_1\jmath}. C_\jmath\sqsubseteq C_{i_1\jmath}\},\{C_{i_1\jmath}(a,n)\})$  be the application of the corresponding rule of form \eqref{math:inference-rule:return} and $\Fmc_{m+p+2}=\Fmc_{m+p+1}\cup\rn(f_{m+p+2})$. 
		
		\end{itemize}
		We obtain a derivation $\Fmc_0 \xto{f_1} \dots \xto{f_{m+p+2}} \Fmc_{m+p+2}$ witnessing $(\T_G,\{A(a,0)\})\vdash C_{i_1\jmath}(a,n)$, i.e., $\T_G\models A\sqsubseteq \Next^n C_{\iota}$.
	\end{proof}

\LmCgMain*
	\begin{proof}
		$(\Rightarrow)$ We show by induction on $p$ that for every $i\in\{1,\dots,m\}$, if there exists a derivation of size $p$ witnessing $(\T_G, \{A(a, 0)\}) \vdash B_i(a, n)$, then $G\vdash \Bmc_i(c^n)$. 
		
		Since $A \neq B_i$ by construction of $\T_G$, the base case is $p=1$. The only possible rule application that produces $ B_i(a, n)$ using formulas from $\T_G\cup\{A(a,0)\}$ is $f_1$ with $\dom(f_1)=\{A(a,0), A\sqsubseteq\Next^n B_i\}$. By  \eqref{math:ci:from-grammar:now} or \eqref{math:ci:from-grammar:shift} (depending on whether $n=0$ or not), it follows that $\Bmc_i\rightarrow c^n$ is a rule of $G$. Hence, $G \vdash \B_i(c^n)$. 
		
		Induction step: assume that the property holds for $p-1$ and let $\Fmc_0 \xto{f_1} \dots \xto{f_p} \Fmc_p$ be a derivation witnessing $(\T_G, \{A(a, 0)\}) \vdash B_i(a, n)$. If $B_i(a, n)\in\Fmc_{p-1}$, the result follows by the induction hypothesis. Otherwise, there are three possible cases for the last rule application that produces $B_i(a, n)$.
\begin{itemize}
		\item $\dom(f_p)=\{A(a, 0), A\sqsubseteq\Next^n B_i\}$, and we conclude as in the base case that $G \vdash \B_i(c^n)$. 
						
		\item $\dom(f_p)=\{C_{\iota(\alpha_1)}(a, n), C_{\iota(\alpha_1)}\sqsubseteq B_i\}$, and by  \eqref{math:ci:from-grammar:imply}, $\Bmc_i\rightarrow \alpha_1$ is a rule of $G$. Since $C_{\iota(\alpha_1)}(a, n)\in\Fmc_{p-1}$, 
 $(\T_G, \{A(a, 0)\}) \vdash C_{\iota(\alpha_1)}(a, n)$. Let $\alpha_1 = \B_{i_1}\dots \B_{i_k}$. By Lemma~\ref{lm:app:LmCgMeetPreciseFirstDirection}, $n = n_1 + \dots + n_k$, and for $1 \leqslant j \leqslant k$, $(\T, \{A(a, 0)\} \vdash B_{i_j}(a, n_j)$ and there is a derivation of length at most $p-2$ witnessing it. Hence, by the induction hypothesis, we conclude that $G \vdash \B_{i_j}(c^{n_j})$. Then, since $\Bmc_i\rightarrow \B_{i_1}\dots \B_{i_k}$ is a rule of $G$, we obtain $G \vdash \B_i(c^n)$.
						
		\item $\dom(f_p)=\{C_{\iota(\alpha_1)}(a, n), C_{\iota(\alpha_2)}(a, n), C_{\iota(\alpha_1)}\sqcap C_{\iota(\alpha_2)}\sqsubseteq B_i\}$, and by  \eqref{math:ci:from-grammar:meet}, $\Bmc_i\rightarrow \alpha_1\&\alpha_2$ is a rule of $G$. The argument is analogous to the one above: we consider separately $\alpha_1= \B_{i_1}\dots \B_{i_k}$ and $\alpha_2= \B'_{i'_1}\dots \B'_{i'_{k'}}$ to obtain  $n = n_1 + \dots + n_k$ and $n = n'_1 + \dots + n'_{k'}$ such that $G \vdash \B_{i_j}(c^{n_j})$ for every $1 \leqslant j \leqslant k$ and $G \vdash \B'_{i'_j}(c^{n'_j})$ for every $1 \leqslant j \leqslant k'$, and conclude using the fact that $\Bmc_i\rightarrow \B_{i_1}\dots \B_{i_k}\& \B'_{i'_1}\dots \B'_{i'_{k'}}$ is a rule of $G$.
		\end{itemize}
		
		\noindent
		$(\Leftarrow)$ We show by induction on $p$ that for every $i\in\{1,\dots,m\}$, if there exists a derivation of size $p$ witnessing $G\vdash \Bmc_i(c^n)$, then $(\T_G, \{A(a, 0)\}) \vdash B_i(a, n)$. 
		
		In the base case, $p = 1$, and there are two possible cases for the rule application $g_1$ that produces $\Bmc_i(c^n)$ from $c(c)$.
		\begin{itemize}
			\item $g_1$ is the application of the rule from  \eqref{math:grammar:semantics:deduction-rules} that corresponds to $\Bmc_i\rightarrow \varepsilon$. In this case, $n=0$ and by \eqref{math:ci:from-grammar:now}, $A\sqsubseteq B_i\in\T_G$.
			
			\item $g_1$ is the application of the rule from  \eqref{math:grammar:semantics:deduction-rules} that corresponds to $\Bmc_i\rightarrow c^n$. In this case, by \eqref{math:ci:from-grammar:shift}, $A\sqsubseteq\Next^n B_i\in\T_G$.
			
		\end{itemize}
		In both cases, $(\T_G, \{A(a, 0)\}) \vdash B_i(a, n)$.
		
		Induction step: assume that the property holds for $p-1$ and let $\Gmc_0 \xto{g_1} \dots \xto{g_p} \Gmc_p$ be a derivation witnessing $G \vdash \B_i(c^n)$. 
		If $\B_i(c^n)\in\Gmc_{p-1}$, the result follows by the induction hypothesis. Otherwise, 
there are two cases for $g_p$ that produces $\B_i(c^n)$. 
		\begin{itemize}
			\item $\dom(g_p)=\{\B_{i_1}(c^{n_1}), \dots, \B_{i_k}(c^{n_k})\}$ with $n=n_1+\dots+n_k$ and $g_p$ is the application of the rule from  \eqref{math:grammar:semantics:deduction-rules} that corresponds to $\Bmc_i\rightarrow \alpha_1$ with $\alpha_1 = \B_{i_1} \dots \B_{i_k}$. For $1 \leqslant j \leqslant k$, since $\B_{i_j}(c^{n_j}) \in \G_{p - 1}$, so that $G \vdash \B_{i_j}(c^{n_j})$ via a derivation of length at most $p - 1$, by the induction hypothesis $(\T_G, \{A(a, 0)\}) \vdash B_{i_j}(c^{n_j})$. Hence, by Lemma~\ref{lm:conjunctive-grammars:meet}, $(\T_G, \{A(a, 0)\}) \vdash C_{\iota(\alpha_1)}(a, n)$. Since by \eqref{math:ci:from-grammar:imply}, $C_{\iota(\alpha_1)}\sqsubseteq B_i$, we obtain that $(\T_G, \{A(a, 0)\}) \vdash B_i(a, n)$.
			
			\item $\dom(g_p)=\{\B_{i_1}(c^{n_1}), \dots, \B_{i_k}(c^{n_k}),\B'_{i'_1}(c^{n'_1}), \dots, \B'_{i'_{k'}}(c^{n'_{k'}})\}$ with $n=n_1+\dots+n_k$, $n=n'_1+\dots+n'_{k'}$ and $g_p$ is the application of the rule from  \eqref{math:grammar:semantics:deduction-rules} that corresponds to $\Bmc_i\rightarrow \alpha_1\&\alpha_2$ with $\alpha_1 = \B_{i_1} \dots \B_{i_k}$ and $\alpha_2 = \B'_{i'_1} \dots \B'_{i'_{k'}}$. As above, we obtain $(\T_G, \{A(a, 0)\}) \vdash C_{\iota(\alpha_1)}(a, n)$ and $(\T_G, \{A(a, 0)\}) \vdash C_{\iota(\alpha_2)}(a, n)$, and conclude using the fact that by \eqref{math:ci:from-grammar:meet}, $C_{\iota(\alpha_1)}\sqcap C_{\iota(\alpha_1)}\sqsubseteq B_i$.
			\qedhere
		\end{itemize}
	\end{proof}
	Theorem~\ref{thm:conjunctive-grammars:grammar-to-tbox} then follows from Definition~\ref{def:grammar-to-tbox} and Lemma \ref{lm:conjunctive-grammars:main}.
\GrammarsToTBoxes*

% !TEX root =  ../ms.tex

\section{Proof of Theorem \ref{thm:linear:tbox-to-grammar}}\label{app:linear}

Recall that in this section, we consider a \TELnl-TBox \T. 
We start with the lemmas.

\LemLinearTBoxesRigidToCFGrammars*
\begin{proof}
The proof is similar to that of Theorem \ref{thm:conjunctive-grammars:tbox-to-grammar}, but this time we have to care about two directions in time and two symbols.
	\smallskip
	
	\noindent
	$(\Rightarrow)$ 
	We show by induction on $m$ that for every $A,B \in \CN(\T)$, $a \in \NI\cup\NN$, and $n \in \Z$, if there exists a derivation witnessing $(\T, \{A(a, 0)\}) \vdash B(a, n)$ of length at most $m$, then there is a word $w \in \{c, d\}^*$, such that $\Gamma_\T \vdash \Nmc_{AB}(w)$ and $\# c(w) - \# d(w) = n$. 
	
	In the base case, $m = 0$, the derivation consists only of $\F_0 = \T \cup \{A(a, 0)\}$, so $B(a, n) \in \F_0$ implies that $A = B$ and $n = 0$, and by \eqref{math:grammar:from-tbox:epsilon}, $\Nmc_{AA}\rightarrow \varepsilon$ is a rule of $\Gamma_\T$  so $\Gamma_\T \vdash \Nmc_{AB}(\varepsilon)$. We get $n = 0 = \# c(\varepsilon) - \# d(\varepsilon)$.
	
	Induction step: assume that the property holds for $m - 1$ and let $\Fmc_0 \xto{f_1} \dots \xto{f_m} \Fmc_m$ be a derivation witnessing $(\T, \{A(a,0)\}) \vdash B(a, n)$. If $B(a, n) \in \Fmc_{m-1}$, the result follows by induction hypothesis. Otherwise, there are two possible cases for the last rule application $f_m$ that produces $B(a, n)$:
	\begin{itemize}
		\item $\dom(f_m )= \{A'(a,n - k), A' \sqs \Next^k B \}$ for some $k\in\Z$ 
		is the application of a rule of form \eqref{math:inference-rule:shift}. 
		Since $(\F_0, \dots, \F_{m - 1})$ is a derivation witnessing $(\T,\{A(a,0)\})\vdash A'(a, n - k)$, by induction hypothesis, there is $u \in \{c, d\}^*$ such that $\Gamma_\T \vdash \Nmc_{AA'}(u)$, and $\# c(u) - \# d(u) = n - k$. 
		Moreover, $A' \sqs \Next^k B\in\T$, and, depending on $k$, we have the following two cases:
		\begin{itemize}
			\item if $k \geqslant 0$, then by \eqref{math:grammar:from-tbox:epsilon} or \eqref{math:grammar:from-tbox:shift}, $\Nmc_{A'B}\rightarrow c^k$ is a rule of $\Gamma_\T$, and thus $\Gamma_\T \vdash \Nmc_{A'B}(v)$ for $v = c^k$;
			
			\item if $k < 0$, then by \eqref{math:grammar:from-tbox:negative-shift}, $\Nmc_{A'B}\rightarrow d^{|k|}$ is a rule of $\Gamma_\T$, and thus $\Gamma_\T \vdash \Nmc_{A'B}(v)$ for $v = d^{|k|}$;
		\end{itemize}
		Then, since by \eqref{math:grammar:from-tbox:middle} $\Nmc_{AB} \rightarrow \Nmc_{AA'}\Nmc_{A'B}$ is a rule of $\Gamma_\T$, we obtain $\Gamma_\T \vdash \Nmc_{AB}(uv)$, and $w = uv$ is such that $\# c(w) - \# d(w) = n$. Indeed, 
		\begin{itemize}
		\item if $v=c^k$, then $\# c(w) = \# c(u)+k$ and $\# d(w)=\# d(u)$ so $\# c(w) - \# d(w) =\# c(u)+k- \# d(u)= n-k+k=n$, and 
		\item if $v = d^{|k|}$, $\# c(w) = \# c(u)$ and $\# d(w)=\# d(u)+|k|$ so $\# c(w) - \# d(w) =\# c(u)- \# d(u)-|k|=n-k-|k|=n$ since in this case, $k<0$.
		\end{itemize}
		\item $\dom(f_m)=\{r(a,b, n),A'(b,n), \exists r.A' \sqs B \}$ is the application of a rule of form \eqref{math:inference-rule:return}. By the form of the derivation rules, $r(a,b, n)\in\Fmc_{m-1}$ has been produced by the application of a rule of form \eqref{math:inference-rule:rigid} or \eqref{math:inference-rule:existential}. Hence, there must be an index $i < m - 1$ such that for some $A'', B'$ and $k \in \Z$, $\dom(f_i) = \{A'' \sqs \exists r . B', A''(a, k)\}$, $\rn(f_i) = \{r(a, b, k), B'(b, k)\}$. Thus $A''(a, k) \in \F_{i}$, and $(\Fmc_0,\dots,\Fmc_i)$ is a derivation witnessing $(\T, \{A(a,0)\})\vdash A''(a, k)$.  By the induction hypothesis, we obtain a word $u$, $\# c(u) - \# d(u) = k$, such that $\Gamma_\T \vdash \Nmc_{AA''}(u)$. Moreover, one can extract from $(\F_i, \dots, \F_{m - 1})$ a derivation of length at most $m - 1$ witnessing $(\T, \{B'(b, k)\}) \vdash A'(b, n)$. We get a derivation witnessing $(\T, \{B'(b, 0)\}) \vdash A'(b, n - k)$ by shifting all timestamps in this derivation, and thus, by the induction hypothesis, a word $v$, $\# c(v) - \# d(v) = n - k$, $\Gamma_\T \vdash \Nmc_{B'A'}(v)$. Since $A'' \sqs \exists r . B'$ and $\exists r.A' \sqs B$ are in \T, by \eqref{math:grammar:from-tbox:down-and-up}, $\Nmc_{A''B}\rightarrow \Nmc_{B'A'}$ is a rule of $\Gamma_\T$, so from $\Gamma_\T \vdash \Nmc_{B'A'}(v)$, we get $\Gamma_\T \vdash \Nmc_{A''B}(v)$. We combine this with $\Gamma_\T \vdash \Nmc_{AA''}(u)$ and the fact that by \eqref{math:grammar:from-tbox:middle}, $\Nmc_{AB}\rightarrow\Nmc_{AA''}\Nmc_{A''B}$ is a rule of $\Gamma_\T$ to establish $\Gamma_\T \vdash \Nmc_{AB}(uv)$. Then $w = uv$ is such that $\# c(w) - \# d(w) = \# c(u) - \# d(u) + \# c(v) - \# d(v)=k+n-k=n$. 
	\end{itemize}
	
	\noindent
	$(\Leftarrow)$ We show by induction on $m$ that for every $A, B \in \CN(\T)$, $a \in \NI\cup\NN$, and $n \in \Z$, if there exists a word $w \in \{c, d\}^*$ such that $\# c(w) - \# d(w) = n$, and a derivation witnessing $\Gamma_\T \vdash \Nmc_{AB}(w)$ of length at most $m$, then $(\T, \{A(a, 0)\}) \vdash B(a, n)$. 
	
	Since $\Gmc_0=\{c(c)\}$, the base case is $m = 1$, when the derivation consists of $\Gmc_0 \xto{g_1} \Gmc_1$. 
	There are two possible cases  for the rule application $g_1$ that produces $\Nmc_{AB}(w)$ from $c(c)$:
	\begin{itemize}
		\item $g_1$ is the application of the rule from  \eqref{math:grammar:semantics:deduction-rules} that corresponds to $\Nmc_{AB}\rightarrow \varepsilon$. In this case, $n = 0$ and by \eqref{math:grammar:from-tbox:epsilon}, $A\sqsubseteq B \in \T$ or $A = B$. In both cases, for every $a \in \NI\cup\NN$, $(\T, \{A(a, 0)\}) \vdash B(a, 0)$, i.e. $(\T, \{A(a, 0)\}) \vdash B(a, n)$.
		 
		\item $g_1$ is the application of the rule from  \eqref{math:grammar:semantics:deduction-rules} that corresponds to $\Nmc_{AB} \rightarrow c^n$ ($n\geqslant 0$) or to $\Nmc_{AB} \rightarrow d^{|n|}$ ($n<0$). In this case, by \eqref{math:grammar:from-tbox:shift} or \eqref{math:grammar:from-tbox:negative-shift}, respectively, $A\sqsubseteq \Next^n B \in \T$. Hence, for every $a \in \NI\cup\NN$,  $(\T, \{A(a, 0)\}) \vdash B(a, n)$. 
	\end{itemize}
	Induction step: assume that the property holds for $m - 1$ and let $\Gmc_0 \xto{g_1} \dots \xto{g_m} \Gmc_m$ be a derivation witnessing $\Gamma_\T \vdash \Nmc_{AB}(w)$, $\#c(w) - \#d(w) = n$. If $\Nmc_{AB}(w) \in \Gmc_{m - 1}$, the result follows by induction hypothesis. Otherwise, there are two possible cases for the last rule application $g_m$ that produces $\Nmc_{AB}(w)$:
	\begin{itemize}
		\item $\dom(g_m)=\{\Nmc_{CD}(w)\}$, and $g_m$ is the application of the rule from  \eqref{math:grammar:semantics:deduction-rules} that corresponds to  $\Nmc_{AB}\rightarrow \Nmc_{CD}$. By \eqref{math:grammar:from-tbox:down-and-up}, there exists $r$ such that $A\sqsubseteq \exists r.C$ and $\exists r.D\sqsubseteq B$ are in \T. Then $(\Gmc_0,\dots,\Gmc_{m-1})$ is a derivation witnessing $G \vdash \Nmc_{CD}(w)$ so by the induction hypothesis, for every $a \in \NI\cup\NN$, $(\T, \{C(a, 0)\}) \vdash D(a, n)$. Let $\Fmc_0 \xto{f_1} \dots \xto{f_p} \Fmc_p$ be obtained from that derivation by substituting $a$ everywhere with a fresh $b \in \NN$. Then the following is a derivation witnessing $(\T, \{A(a, 0)\}) \vdash B(a, n)$:
		\begin{align*}
			\F \xto{f_0} \F'_0 \xto{f_1} \dots \xto{f_p} \F'_p \xto{f_{p + 1}} \F_{p + 1} \xto{f_{p + 2}} \F_{p + 2}
		\end{align*}
		where
		\begin{itemize}
			\item $\F = \{A(a, 0)\} \cup \T$;
			
			\item $f_0$ is an application of the rule of form \eqref{math:inference-rule:existential} with $\dom(f_0) = \{A(a, 0), A \sqs \exists r.C\}$ and $\rn(f_0) = \{r(a, b, 0), C(b, 0)\}$;
			\item $\F'_i=\F_i\cup\{A(a, 0),r(a, b, 0)\}$ for $0\leqslant i\leqslant p$;
			\item $f_{p+ 1}$ is an application of the rule of form \eqref{math:inference-rule:rigid} with $\dom(f_{p+1}) = \{r(a, b, 0)\}$ and $\rn(f_{p+ 1}) = \{r(a, b, n)\}$;%$\dom(f_{p+1}) = \{r(a, b, 0)\}$ and $\rn(f_{p+ 1}) = \{r(a, b, k) \mid k\in\Z\}$;
			
			\item $f_{p + 2}$ is an application of the rule of form \eqref{math:inference-rule:return} with $\dom(f_{p+2}) = \{r(a, b, n), D(b, n), \exists r . D \sqs B\}$ and $\rn(f_{p + 2}) = \{B(a, n)\}$;
			
			\item $\F_{p + 1}=\F'_p\cup\rn(f_{p+1})$ and $\F_{p + 2}=\F_{p+1}\cup\rn(f_{p+2})$.
		\end{itemize}
		
		\item $\dom(g_m)=\{\Nmc_{AC}(u),\Nmc_{CB}(v)\}$, where $\#c(u) - \#d(u) = n - k$ and $\#c(v) - \#d(v) = k$, for some $k \in \Z$, and $g_m$ is the application of the rule from  \eqref{math:grammar:semantics:deduction-rules} that corresponds to $\Nmc_{AB} \rightarrow \Nmc_{AC}\Nmc_{CB}$. By \eqref{math:grammar:from-tbox:middle} $A,B,C\in\NC(\T)$. 
		Then $(\Gmc_0,\dots,\Gmc_{m-1})$ is a derivation witnessing $G \vdash \Nmc_{AC}(u)$ and $G \vdash \Nmc_{CB}(v)$. By the induction hypothesis, for every $a\in\NI\cup\NN$, $(\T, \{A(a, 0)\}) \vdash C(a, n - k)$ and $(\T, \{C(a, 0)\}) \vdash B(a, k)$, so (by Proposition~\ref{prop:derivationTEL}), $(\T, \{C(a, n-k)\}) \vdash B(a, n)$. Hence, $(\T, \{A(a, 0)\}) \vdash B(a, n)$.  \qedhere
	\end{itemize}
\end{proof}

\LmLinearLocalRolesRemoval*
\begin{proof}
	$(\Rightarrow)$ Let $\Fmc_0 \xto{f_1} \dots \xto{f_m} \Fmc_m$ be a derivation witnessing $(\T, \{A(a, 0)\}) \vdash B(a, n)$. We build a derivation witnessing $(\Triglin, \{A(a, 0)\}) \vdash B(a, n)$. Let $(h_1,\dots, h_p)$ be the sequence of rule applications obtained from $(f_1,\dots,f_m)$ as follows. For every application $f_i$ of a rule of the form \eqref{math:inference-rule:return}, where $f_i = (\{r(a, b, \ell), A'(b, \ell), \exists r . A' \sqs B'\}, \{B'(a, \ell)\})$ with $r \in \RNloc$, find the application $f_j$, $j < i$, of a rule of the form \eqref{math:inference-rule:existential} that produced $r(a,b,\ell)$: $f_j = (\{A''(a, \ell), A'' \sqs \exists r . B''\}, \{r(a, b, \ell), B''(b, \ell)\})$. From $(\F_{j - 1}, \dots, \F_i)$, by Proposition \ref{prop:derivationTEL}, we get $\T \mdl A'' \sqs B'$. Hence, $A'' \sqs B'\in\Triglin$. Substitute the sequence $(f_j, \dots, f_i)$ with a single application $h$ of a rule of the form \eqref{math:inference-rule:shift}, where $h = (\{A''(a, \ell), A'' \sqs B'\}, \{B'(a, \ell)\})$.
	
	Since all concept inclusions used in the premises of $h_1,\dots, h_p$ belongs to \Triglin by construction, we indeed obtain a derivation $\Fmc'_0 \xto{h_1} \dots \xto{h_p} \Fmc'_p$ witnessing $(\Triglin, \{A(a, 0)\}) \vdash B(a, n)$ by setting $\Fmc'_0=\{A(a, 0)\}\cup\Triglin$ and $\Fmc'_i=\Fmc'_{i-1}\cup\rn(h_i)$.  
	\smallskip
	
	\noindent$(\Leftarrow)$ Suppose $\Fmc_0 \xto{f_1} \dots \xto{f_m} \Fmc_m$ is a derivation witnessing $(\Triglin, \{A(a, 0)\}) \vdash B(a, n)$. Let $\T_1=\Triglin\setminus\T$. Since $\Triglin = \T_0 \cup \{A' \sqs B' \mid \T \mdl A' \sqs B'\}$, where $\T_0$ is obtained from \T by removing concept inclusions that use local role names, $\T_1$ contains only concept inclusions of the form $A'\sqsubseteq B'$. 
	
	We build a derivation witnessing $(\T, \{A(a, 0)\}) \vdash B(a, n)$. 
	Let $(h_1,\dots, h_p)$ be the sequence of rule applications obtained from $(f_1,\dots,f_m)$ as follows. 
	For each $f_i = (\{A'(b, \ell), A' \sqs B'\}, \{B'(b, \ell)\})$ such that $A' \sqs B' \in \T_1$, apply Proposition \ref{prop:derivationTEL} to obtain a sequence $(f'_1, \dots, f'_k)$ of rule applications corresponding to a derivation witnessing $(\T, \{A'(b, \ell)\}) \vdash B'(b, \ell)$. Then, substitute $f_i$ with $(f'_1, \dots, f'_k)$.

	We can check that all concept inclusions used in the premises of $h_1,\dots, h_p$ belong to \T, so that we indeed obtain a derivation $\Fmc'_0 \xto{h_1} \dots \xto{h_p} \Fmc'_p$ witnessing $(\T, \{A(a, 0)\}) \vdash B(a, n)$ by setting $\Fmc'_0=\{A(a, 0)\}\cup\T$ and $\Fmc'_i=\Fmc'_{i-1}\cup\rn(h_i)$. 
\end{proof}

\ThmLinearTBoxesToCFGrammars*

\begin{proof}
	Using Lemma \ref{lm:linear:local-roles-removal} and the facts that $\RN(\Triglin) \sbs \RNrig$ and that $|\Triglin|$ is polynomial in $|\T|$, we can assume that $\RN(\T) \sbs \RNrig$.  Let $\Gamma_\T$ be the grammar given in Definition~\ref{def:linear:GammaT}. 
	It is easy to check that the size of $\Gamma_\T$ is polynomial in $|\T|$, and by Lemma~\ref{lem:linear:rigid-tbox-to-grammar}, for any $A, B \in \CN(\T)$, $\T \mdl A \sqs \Next^n B$ iff there exists $w \in L_{\Gamma_\T}(\Nmc_{AB})$ with $\#c(w) - \#d(w) = n$.
\end{proof}
	
% !TEX root =  ../ms.tex

\section{Missing proofs for Section \ref{sec:consequences}}\label{app:consequences}

\subsection{Proof of Theorem \ref{thm:consequences:future-answering}}

\ThmFuturePTime*

\begin{proof}
	The lower bounds hold already for the description logic \EL (without temporal operators) \citep{DBLP:journals/ai/CalvaneseGLLR13}. For the upper bounds, we provide a polynomial reduction from the problem of deciding whether $(\T, \A) \mdl A(a, n)$ to that of checking whether a word belongs to the language of a conjunctive grammar, which can be tested in polynomial time (Theorem~\ref{thm:membership-testing}). Our reduction builds a \TELnf-TBox $\T'\cup\T_\A$, an assertion $C_a(a, l)$ and a concept name $A_n$ such that $(\T, \A) \mdl A(a, n)$ iff $(\T' \cup \T_\A, \{C_a(a, l)\}) \mdl A_n(a, n)$, then use Proposition \ref{prop:derivationTEL} and Theorem \ref{thm:conjunctive-grammars:tbox-to-grammar} to conclude. The idea is to encode all information about $a$ in $\A$ into the single fact $C_a(a, l)$ thanks to $\T_\A$.
	
	Let  $\ind(\A)$ be the set of individual names that occur in \A, and $l, m \in \Z$ be the least and the greatest timestamps appearing in \A. 
	We introduce fresh concept names $\{C_a \mid a \in \ind(\A)\}$ and $\{A_k \mid A \in \CN(\T), l \leqslant k \leqslant m + 1\}$, and role names $\{\role{ab} \in \RNrig \mid  a, b \in \ind(\A), {r\in\NR(\T)}\}$. For the convenience of notation, we write $A_k$ for all $k \geqslant l$, assuming that $A_k = A_{m + 1}$ when $k > m$.
		
	Let $\T'$ be the TBox containing the following concept inclusions for all $l\leqslant k\leqslant m+1$.	
	\begin{align}
			&A_k \sqs \Next^s B_{k + s} &&\text{for } A \sqs \Next^s B \in \T
			\tag{\ref{math:ci:abox-to-tbox:shift}}\\
			&A_k \sqcap A'_k \sqs B_k &&\text{for } A \sqcap A' \sqs B \in \T
			\tag{\ref{math:ci:abox-to-tbox:conjunction}}\\
			&A_k \sqs \exists r . B_{k} &&\text{for } A \sqs \exists r . B \in \T
			\tag{\ref{math:ci:abox-to-tbox:existential:null}}\\
			&\exists r . A_k \sqs B_{k} &&\text{for } \exists r . A \sqs B \in \T
			\tag{\ref{math:ci:abox-to-tbox:return:null}}
	\end{align}
	Additionally, define a TBox $\T_\A$ to contain the following concept inclusions.
	\begin{align}
			&C_a \sqs \Next^{k - l} A_k &&\text{for } A(a, k) \in \A
			\tag{\ref{math:ci:abox-to-tbox:mark}}\\
			&C_a \sqs \exists \role{ab}\, .\, C_{b} &&\text{for } r(a, b, \ell) \in \A
			\tag{\ref{math:ci:abox-to-tbox:existential:abox}}\\
			&\exists \role{ab}\, .\, A_k \sqs B_k &&\text{for } \exists r . A \sqs B \in \T, r(a, b, \ell) \in \A,
			\tag{\ref{math:ci:abox-to-tbox:return:abox}}\\
			&\ \notag 														&&\text{where } r \in \RNrig \text{ or } \ell = k
 	\end{align}
	Clearly, both $\T'$ and $\T_\A$ can be constructed in polynomial time w.r.t.~$|\T|+|\A|$ and are expressed in \TELnf (since \T is a \TELnf-TBox and $k-l\geqslant 0$ for every $A(a, k) \in \A$ by definition of $l$). 	
	
	\begin{applemma}\label{lm:app:consequences:indexed-concepts}
		{For all $A, B \in \CN(\T)$ and $a\in\NI$: 		
		\begin{enumerate}[label=(\roman*)]
			\item for all $n\in \N$ and $l\leqslant k,\ell\leqslant m+1$, $$(\T', \{A_k(a, 0)\}) \mdl B_{\ell}(a, n)\text{ implies that }\ell=k+n$$ (or $\ell=m+1$ and $k+n>m$);\label{lm:app:consequences:indexed-concepts:n}
			
			\item for all $s,n,k\in \N$, $$(\T, \{A(a, s)\}) \mdl B(a, n)\text{ iff }(\T', \{A_{k+s}(a, s)\}) \mdl B_{k + n}(a, n)$$ (note that if $n<s$ the equivalence holds trivially since $\T'$ is a \TELnf TBox, and recall that if $k+n>m$, $B_{k+n}=B_{m+1}$, and that if $k+s>m$, $A_{k+s}=A_{m+1}$).\label{lm:app:consequences:indexed-concepts:iff}	
		\end{enumerate}}
	\end{applemma}
	\begin{proof}
	For point \e{\ref{lm:app:consequences:indexed-concepts:n}}, we show by induction on $p$ that for all $A, B \in \CN(\T)$, $a\in\NI\cup\NN$, $n\in \N$, and $l\leqslant k,\ell\leqslant m+1$, if there is a derivation witnessing $(\T', \{A_k(a, 0)\}) \mdl B_{\ell}(a, n')$ of size $p$, then either $\ell=k+n$, or $B_\ell=B_{m+1}$ and $k+n>m$. 	
	
	The base case, $p=0$, is immediate, since in this case $n=0$ and $B_{\ell}=A_k$, so $\ell=k+n$. 
	
	Induction step: assume that the property holds for $p - 1$ and let $\F_0 \xto{f_1} \dots \xto{f_p} \F_p$ be a derivation witnessing $(\T', \{A_k(a, 0)\}) \mdl B_{\ell}(a, n)$. If $B_{\ell}(a, n) \in \F_{p - 1}$, the result follows by induction hypothesis. Otherwise, there are three possible cases for the last rule application $f_p$ that produces $B_{\ell}(a, n)$: 
	
	\begin{itemize}
			\item $\dom(f_p)=\{C_j\sqsubseteq \Next^{\ell-j} B_{\ell}, C_j(a,n-\ell+j)\}$, and $f_p$ is the application of a rule of the form \eqref{math:inference-rule:shift}. Since $C_j(a,n-\ell+j)\in\Fmc_{p-1}$, $(\T', \{A_k(a, 0)\}) \mdl C_{j}(a,n-\ell+j)$ with a derivation of size $p-1$, so by induction hypothesis, either $j=n-\ell+j+k$, or $C_j=C_{m+1}$ and $n-\ell+j+k>m$. In the first case, $\ell=k+n$, and in the second case, since $\ell-j>0$ as $\T'$ is a \TELnf-TBox, $\ell>j$ so $B_\ell=B_{m+1}$ and $k+n>m+\ell-j>m$. 
			
\item $\dom(f_p)=\{C_{\ell}\sqcap C'_{\ell}\sqsubseteq B_{\ell}, C_{\ell}(a,n),C'_{\ell}(a,n)\}$, and $f_p$ is the application of a rule of the form  \eqref{math:inference-rule:conjunction}. Since $C_{\ell}(a,n)\in\Fmc_{p-1}$ and $C'_{\ell}(a,n)\in\Fmc_{p-1}$, $(\Fmc_0,\dots,\Fmc_{p-1})$ is a derivation witnessing $(\T', \{A_k(a, 0)\}) \mdl C_{\ell}(a,n)$ and $(\T', \{A_k(a, 0)\}) \mdl C'_{\ell}(a,n)$. By the induction hypothesis, we obtain that either $\ell=k+n$ or $C_{\ell}=C_{m+1}$, $C'_{\ell}=C'_{m+1}$ and $k+n>m$. 

\item $\dom(f_p)=\{\exists r.C_{\ell}\sqsubseteq B_{\ell}, r(a,b,n), C_{\ell}(b,n)\}$, and $f_p$ is the application of a rule of the form  \eqref{math:inference-rule:return}.  
By the form of the derivation rules, $r(a, b, n) \in \Fmc_{p - 1}$ has been produced by the application of a rule of form \eqref{math:inference-rule:rigid} or \eqref{math:inference-rule:existential}. Hence, there must be an index $i < p - 1$ such that for some $A''_j, B'_j$ and $j'$, $\dom(f_i) = \{A''_j \sqs \exists r . B'_j , A''_\ell(a, j')\}$, $\rn(f_i) = \{r(a, b, j'), B'_j(b, j')\}$. Thus $A''_j (a, j') \in \F_{i - 1}$, and $(\Fmc_0, \dots, \Fmc_{i - 1})$ is a derivation witnessing $(\T', \{A_k(a, 0)\}) \mdl A''_j(a, j')$. By induction hypothesis, it follows that either (1) $j=k+j'$ or (2) $A''_j=A''_{m+1}$ and $k+j'>m$. 
Moreover, one can extract from $(\Fmc_i,\dots,\Fmc_{p-1})$ a derivation for $(\T',\{B'_j(b, j')\})\vdash C_{\ell}(b,n)$ of size at most $p-1$, from which we obtain a derivation for $(\T',\{B'_j(b, 0)\})\vdash C_{\ell}(b,n-j')$ by shifting all timestamps. Hence, by induction hypothesis, we have either (i) $\ell=j+n-j'$ or (ii) $C_\ell=C_{m+1}$ and $j+n-j'>m$. Moreover, since $\T'$ is a \TELnf-TBox, $n\geqslant j'$.
\begin{itemize}
\item In case (1-i), $\ell=j+n-j'=k+j'+n-j'=k+n$.
\item In case (1-ii), $C_\ell=C_{m+1}$ and $j+n-j'>m$ so $k+j'+n-j'>m$, i.e., $k+n>m$.
\item In case (2), $k+j'>m$ so since $n\geqslant j'$, $k+n>m$. Moreover, in this case, $A''_j=A''_{m+1}$ so $B'_j=B'_{m+1}$ and by the form of the concept inclusions in \T' (where concept names that occur in the right-hand side have always equal or higher indexes than those from the left-hand side), it must be the case that $C_\ell=C_{m+1}$.
\end{itemize}
\end{itemize}
	
	For point (ii), we show the two directions of the equivalence in a similar way.
	
	\noindent$(\Rightarrow)$ We show by induction on $p$ that for all $A,B \in \CN(\T)$, $a\in\NI\cup\NN$, and $s,n \in \N$, if there exists a derivation of length $p$ witnessing $(\T, \{A(a,s)\}) \vdash B(a, n)$, then for every $k$, $(\T', \{A_{k+s}(a,s)\}) \vdash B_{k+n}(a, n)$ (with $m+1$ instead of $k+s$ and/or $k+n$ if they are larger than $m$). 
	
	The base case, $p=0$, is immediate since in this case $B=A$ and $s=n=0$, and it holds that $(\T', \{A_k(a,0)\}) \vdash A_{k}(a, 0)$. 
	
	Induction step: assume that the property holds for $p - 1$ and let $\F_0 \xto{f_1} \dots \xto{f_p} \F_p$ be a derivation witnessing $(\T, \{A(a,s)\}) \vdash B(a, n)$. If $B(a, n) \in \F_{p - 1}$, the result follows by induction hypothesis. Otherwise, there are three possible cases for the last rule application $f_p$ that produces $B(a, n)$: 
	\begin{itemize}
		\item $\dom(f_p)=\{A'(a, n - \ell), A' \sqs \Next^\ell B \}$ for some  $\ell \in \{0,\dots, n - s\}$ (since \T is a \TELnf-TBox and $(\T,\{A(a,s)\})\vdash A'(a, n - \ell)$), and $f_p$ is the application of a rule of the form \eqref{math:inference-rule:shift}. 
			Since $(\F_0, \dots, \F_{p - 1})$ is a derivation witnessing $(\T, \{A(a,s)\}) \vdash A'(a, n - \ell)$, by induction hypothesis, for every $k$, $(\T',\{A_{k+s}(a,s)\}) \vdash A'_{k+n - \ell}(a, n - \ell)$. Since $A' \sqs \Next^\ell B\in \T$ it follows that $A'_{k+n - \ell} \sqs \Next^\ell B_{k+n} \in \T'$ (and $A'_{j}\sqs \Next^\ell B_{m+1}$ for every $j$ such that $j+\ell\geqslant m+1$), by \eqref{math:ci:abox-to-tbox:shift}. Hence, $(\T',\{A_{k+s}(a,s)\}) \vdash B_{k+n}(a, n)$.  
			
			\item $\dom(f_p)=\{A'(a, n), A''(a, n), A' \sqcap A'' \sqs B\}$ and $f_p$ is the application of a rule of form  \eqref{math:inference-rule:conjunction}. Since $A'(a, n)$ and $A''(a, n)$ are in $\Fmc_{p - 1}$, $(\F_0, \dots, \F_{p - 1})$ is a derivation witnessing $(\T, \{A(a,s)\}) \vdash A'(a, n)$ and $(\T, \{A(a,s)\}) \vdash A''(a, n)$. 
			By the induction hypothesis, for every $k$, $(\T', \{A_{k+s}(a,s)\}) \vdash A'_{k+n}(a, n)$ and $(\T',  \{A_{k+s}(a,s)\}) \vdash A''_{k+n}(a, n)$. Since $A' \sqcap A'' \sqs B \in \T$, $A'_{k+n} \sqcap A''_{k+n} \sqs B_{k+n} \in \T'$, by \eqref{math:ci:abox-to-tbox:conjunction}. Hence, $(\T' ,  \{A_{k+s}(a,s)\}) \vdash B_{k+n}(a, n)$. 
	\item $\dom(f_p)=\{r(a, b, n), A'(b, n), \exists r . A' \sqs B\}$ and $f_p$ is the application of a rule of form \eqref{math:inference-rule:return}.  	
	By the form of the derivation rules, $r(a, b, n) \in \Fmc_{p - 1}$ has been produced by the application of a rule of form \eqref{math:inference-rule:rigid} or \eqref{math:inference-rule:existential}. Hence, there must be an index $i < p - 1$ such that for some $A'', B'$ and $j \geqslant s$ (again, because \T is a \TELnf-TBox),
				$\dom(f_i) = \{A'' \sqs \exists r . B', A''(a, j)\}$, $\rn(f_i) = \{r(a, b, j), B'(b, j)\}$. Thus $A''(a, j) \in \F_{i - 1}$, and $(\Fmc_0, \dots, \Fmc_{i - 1})$ is a derivation witnessing $(\T,  \{A(a,s)\}) \vdash A''(a, j)$. By the induction hypothesis, for every $k$, there is a derivation $\der_1$ witnessing $(\T' ,  \{A_{k+s}(a,s)\}) \vdash A''_{k+j}(a, j)$. Moreover, we extract from $(\F_i, \dots, \F_{p - 1})$ a derivation witnessing $(\T, \{B'(b, j)\}) \vdash A'(b, n)$ (note that $n = j$ if $r \in \RNloc$). By the induction hypothesis, $(\T', \{B'_{k+j}(b, j)\}) \vdash A'_{k+n}(b, n)$. Let $\der_2$ be a derivation witnessing this. We obtain a derivation witnessing $(\T' ,  \{A_{k+s}(a,s)\}) \vdash B_{k+n}(a, n)$ as follows. Since $A'' \sqs \exists r . B' \in \T$, we have $A''_{k+j} \sqs \exists r . B'_{k+j}  \in \T'$, by \eqref{math:ci:abox-to-tbox:existential:null}. Further, since $\exists r . A' \sqs B \in \T$, we have $\exists r . A'_{k+n} \sqs B_{k+n}\in\T'$, by \eqref{math:ci:abox-to-tbox:return:null}. Thus, start from $\F'_0 = \T'  \cup  \{A_{k+s}(a,s)\}$, proceed as in $\der_1$ until you derive $A''_{k+j}(a, j)$. Apply a rule of the form \eqref{math:inference-rule:existential} using $A''_{k+j} \sqs \exists r . B'_{k+j}$ to obtain $r(a,b,j)$ and $B'_{k+j}(b, j)$, and proceed as in $\der_2$ until you have $A'_{k+n}(b, n)$. 
				If $r\in\RNrig$, apply a rule of the form \eqref{math:inference-rule:rigid} to obtain $r(a,b,n)$. Otherwise, it means that $j=n$, so we already have $r(a,b,n)$. 
				Finally, apply a rule of the form \eqref{math:inference-rule:return} using $\exists r . A'_{k+n} \sqs B_{k+n}$ to get $B_{k+n}(a, n)$.
	\end{itemize}
	\noindent$(\Leftarrow)$ We show by induction on $p$ that for all $A,B \in \CN(\T)$, $a\in\NI\cup\NN$, and $s,n,k \in \N$, if  there exists a derivation of length $p$ witnessing $(\T', \{A_{k+s}(a,s)\}) \vdash B_{k+n}(a, n)$ (with $m+1$ instead of $k+s$ and/or $k+n$ if they are larger than $m$), then $(\T, \{A(a,s)\}) \vdash B(a, n)$. 
	
	The base case, $p=0$, is immediate since in this case $B=A$ and $n=s$, and it holds that $(\T, \{A(a,s)\}) \vdash A(a, s)$. 
	
		Induction step: assume that the property holds for $p - 1$ and let $\F_0 \xto{f_1} \dots \xto{f_p} \F_p$ be a derivation witnessing $(\T', \{A_{k+s}(a,s)\}) \vdash B_{k+n}(a, n)$. If $B_{k+n}(a, n) \in \F_{p - 1}$, the result follows by induction hypothesis. Otherwise, there are three possible cases for the last rule application $f_p$ that produces $B_{k+n}(a, n)$:
		\begin{itemize}
			\item $\dom(f_p)=\{A'_{n+k-\ell}(a, n+k-\ell), A'_{n+k-\ell} \sqs \Next^\ell B_{k+n} \}$ for some  $\ell \in \{0,\dots, n+k-s\}$ (since $\T'$ is a \TELnf-TBox and $(\T',\{A_{k+s}(a,s)\}\mdl A'_{n+k-\ell}(a, n+k-\ell)$), and $f_p$ is the application of a rule of the form \eqref{math:inference-rule:shift}. 
			Since $(\F_0, \dots, \F_{p - 1})$ is a derivation witnessing $(\T',\{A_{k+s}(a,s)\}) \vdash A'_{n+k-\ell}(a, n+k-\ell)$, by induction hypothesis, $(\T,\{A(a,s)\}) \vdash A'(a, n - \ell)$. Since $A'_{n+k-\ell} \sqs \Next^\ell B_{k+n} \in \T'$ it follows that $A' \sqs \Next^\ell B \in \T$, by \eqref{math:ci:abox-to-tbox:shift}. We conclude that $(\T, \{A(a,s)\}) \vdash B(a, n)$ by applying a rule of the form \eqref{math:inference-rule:shift}.
			
	\item $\dom(f_p)=\{A'_{k+n}(a, n), A''_{k+n}(a,n), A'_{k+n} \sqcap A''_{k+n} \sqs B_{k+n}\}$ and $f_p$ is the application of a rule of form  \eqref{math:inference-rule:conjunction}. Since $A'_{k+n}(a, n)$ and $A''_{k+n}(a, n)$ are in $\Fmc_{p - 1}$, $(\F_0, \dots, \F_{p - 1})$ is a derivation witnessing $(\T', \{A_{k+s}(a,s)\}) \vdash A'_{k+n}(a, n)$ and $(\T', \{A_{k+s}(a,s)\}) \vdash A''_{k+n}(a, n)$. 
			By the induction hypothesis, $(\T, \{A(a,s)\}) \vdash A'(a, n)$ and $(\T, \{A(a,s)\}) \vdash A''(a, n)$. Since $A'_{k+n} \sqcap A''_{k+n} \sqs B_{k+n} \in \T'$, $A' \sqcap A'' \sqs B \in \T$, by \eqref{math:ci:abox-to-tbox:conjunction}. We conclude that $(\T, \{A(a,s)\}) \vdash B(a, n)$ by applying a rule of the form \eqref{math:inference-rule:conjunction}.
			
			\item $\dom(f_p)=\{r(a, b, n), A'_{k+n}(b, n), \exists r . A'_{k+n} \sqs B_{k+n}\}$ and $f_p$ is the application of a rule of form \eqref{math:inference-rule:return}. 
			By the form of the derivation rules, $r(a, b, n) \in \Fmc_{p - 1}$ has been produced by the application of a rule of form \eqref{math:inference-rule:rigid} or \eqref{math:inference-rule:existential}. Hence, there must be an index $i < p - 1$ such that for some $A''_\ell, B'_\ell$, $\dom(f_i) = \{A''_\ell \sqs \exists r . B'_\ell, A''_\ell(a, \ell')\}$, $\rn(f_i) = \{r(a, b, \ell'), B'_\ell(b, \ell')\}$ (note that $\ell'=n$ if $r \in \RNloc$). Thus $A''_\ell(a, \ell') \in \F_{i - 1}$, and $(\Fmc_0, \dots, \Fmc_{i - 1})$ is a derivation witnessing $(\T', \{A_{k+s}(a,s)\}) \vdash A''_\ell(a, \ell')$. 
			Since $(\T', \{A_{k+s}(a,0)\}) \vdash A''_\ell(a, \ell'-s)$, by point \e{\ref{lm:app:consequences:indexed-concepts:n}} of the lemma, $\ell=k+s+\ell'-s$, so $\ell=k+\ell'$. 
			Hence, since $(\T', \{A_{k+s}(a,s)\}) \vdash A''_{k+\ell'}(a, \ell')$, 
			 by the induction hypothesis, there is a derivation $\der_1$ witnessing $(\T, \{A(a,s)\}) \vdash A''(a, \ell')$. Moreover, we extract from $(\F_i, \dots, \F_{p - 1})$ a derivation witnessing $(\T', \{B'_{k+\ell'}(b, \ell')\}) \vdash A'_{k+n}(b, n)$. 
			By induction hypothesis, $(\T, \{B'(b, \ell')\}) \vdash A'(b, n)$. Let $\der_2$ be a derivation witnessing it. We obtain a derivation witnessing $(\T, \{A(a,s)\}) \vdash B(a, n)$ as follows. Since $A''_\ell \sqs \exists r . B'_\ell \in \T'$, we have $A'' \sqs \exists r . B' \in \T$, by \eqref{math:ci:abox-to-tbox:existential:null}. Further, since $\exists r . A'_{k+n} \sqs B_{k+n} \in \T'$, we have $\exists r . A' \sqs B\in\T$, by \eqref{math:ci:abox-to-tbox:return:null}. Thus, start from $\F'_0 = \T \cup \{A(a,s)\}$, proceed as in $\der_1$ until you derive $A''(a, \ell')$. Apply a rule of the form \eqref{math:inference-rule:existential} to obtain $r(a,b,\ell')$ and $B'(b, \ell')$, and proceed as in $\der_2$ until you have $A'(b, n)$. 
			If $r\in\RNrig$, apply a rule of the form \eqref{math:inference-rule:rigid} to obtain $r(a,b,n)$. Otherwise, it means that $\ell'=n$, so we already have $r(a,b,n)$. 
			Finally, apply a rule of the form \eqref{math:inference-rule:return} to get $B(a, n)$.	\qedhere
		\end{itemize}	
			\end{proof}
	
	\begin{applemma}\label{lm:app:consequences:compressed-abox}
		For all $A \in \CN(\T)$, $a\in\NI$ and $n \in \N$, $$(\T, \A) \mdl A(a, n)\text{ iff }(\T' \cup \T_\A, \{C_a(a, l)\}) \mdl A_n(a, n)$$ {(recall that if $n>m+1$, $A_n=A_{m+1}$, and that if $n<l$, $(\T, \A) \not\mdl A(a, n)$ and $(\T' \cup \T_\A, \{C_a(a, l)\}) \not\mdl A_n(a, n)$ since both \T and $\T' \cup \T_\A$ are \TELnf-TBoxes)}.
	\end{applemma}
	\begin{proof}	
		$(\Rightarrow)$ We show by induction on $m$ that for all $A \in \CN(\T)$, $a\in\NI$ and $n \in \N$, if there exists a derivation of length $m$ witnessing $(\T, \A) \vdash A(a, n)$, then $(\T' \cup \T_\A, \{C_a(a, l)\}) \vdash A_n(a, n)$. 
		
		In the base case, $m = 0$, the derivation consists only of $\F_0 = \T \cup \A$, so $A(a, n) \in \A$. Hence, by \eqref{math:ci:abox-to-tbox:mark}, $C_a\sqsubseteq \Next^{n-l}A_n\in\T_A$. It follows that $(\T' \cup \T_\A, \{C_a(a, l)\}) \vdash A_n(a, n)$.
		
		Induction step: assume that the property holds for $m - 1$ and let $\F_0 \xto{f_1} \dots \xto{f_m} \F_m$ be a derivation witnessing $(\T, \A) \vdash A(a, n)$. If $A(a, n) \in \F_{m - 1}$, the result follows by induction hypothesis. Otherwise, there are three possible cases for the last rule application $f_m$ that produces $A(a, n)$: 
		\begin{itemize}
			\item $\dom(f_m)=\{A'(a, n - k), A' \sqs \Next^k A \}$ for some  $k \in \{0,\dots, n - l\}$ (since $n - k$ cannot be less than $l$, given that \T is a \TELnf-TBox and $(\T,\A)\vdash A'(a, n - k)$), and $f_m$ is the application of a rule of the form \eqref{math:inference-rule:shift}. 
			Since $(\F_0, \dots, \F_{m - 1})$ is a derivation witnessing $(\T, \A) \vdash A'(a, n - k)$, by induction hypothesis, $(\T' \cup \T_\A, \{C_a(a, l)\}) \vdash A'_{n - k}(a, n - k)$. Since $A' \sqs \Next^k A \in \T$ it follows that $A'_{n - k} \sqs \Next^k A_n \in \T'$, by \eqref{math:ci:abox-to-tbox:shift}. Hence, $(\T' \cup \T_\A, \{C_a(a, l)\}) \vdash A_{n}(a, n)$. 
			
			\item $\dom(f_m)=\{A'(a, n), A''(a, n), A' \sqcap A'' \sqs A\}$ and $f_m$ is the application of a rule of form  \eqref{math:inference-rule:conjunction}. Since $A'(a, n)$ and $A''(a, n)$ are in $\Fmc_{m - 1}$, $(\F_0, \dots, \F_{m - 1})$ is a derivation witnessing $(\T, \A) \vdash A'(a, n)$ and $(\T, \A) \vdash A''(a, n)$. 
			By the induction hypothesis, $(\T' \cup \T_\A, \{C_a(a, l)\}) \vdash A'_n(a, n)$ and $(\T' \cup \T_\A, \{C_a(a, l)\}) \vdash A''_n(a, n)$. Since $A' \sqcap A'' \sqs A \in \T$, $A'_n \sqcap A''_n \sqs A_n \in \T'$, by \eqref{math:ci:abox-to-tbox:conjunction}. Hence, $(\T' \cup \T_\A, \{C_a(a, l)\}) \vdash A_{n}(a, n)$. 			
			\item $\dom(f_m)=\{r(a, b, n), A'(b, n), \exists r . A' \sqs A\}$ and $f_m$ is the application of a rule of form \eqref{math:inference-rule:return}. Then we have two subcases:
			\begin{itemize}
			\item $b \in \ind(\A)$. Since $A'(b, n) \in \F_{m - 1}$, there is a derivation of length $m - 1$ witnessing $(\T, \A) \vdash A'(b, n)$. By the induction hypothesis, there is a derivation \der witnessing 
			$(\T' \cup \T_\A, \{C_b(b, l)\}) \vdash A'_n(b, n)$. 
			Furthermore, since the derivation rules can only add a fact of the form $r(a, b, \ell)$ with $a,b\in\NI$ if $r$ is rigid (cf.~\eqref{math:inference-rule:rigid}--\eqref{math:inference-rule:existential}), there exists $r(a, b, \ell) \in \A$ such that either $r \in \RNrig$ or $\ell = n$. In both cases, by \eqref{math:ci:abox-to-tbox:existential:abox}, $C_a\sqsubseteq \exists\role{a,b}.C_b\in\T_\A$, and by \eqref{math:ci:abox-to-tbox:return:abox}, since $\exists r . A' \sqs A\in\T$, it follows that $\exists \role{ab}\, .\, A'_n \sqs A_n\in\T_\A$. 
			We obtain a derivation witnessing $(\T' \cup \T_\A, \{C_a(a, l)\})$ as follows. Start with $\F_0 = \T' \cup \T_\A \cup \{C_a(a, l)\}$. Apply a rule of the form \eqref{math:inference-rule:existential} using $C_a\sqsubseteq \exists\role{a,b}.C_b$ to obtain $\role{a, b}(a, b', l)$ and $C_b(b', l)$ for some $b' \in \NN$. Proceed as in \der to obtain $A'_n(b', n)$. Apply a rule of the form \eqref{math:inference-rule:rigid} to obtain $\role{a, b}(a, b', n)$, then a rule of the form \eqref{math:inference-rule:return} using $\exists \role{ab}\, .\, A'_n \sqs A_n$ to get $A_n(a, n)$.  
			
				\item $b \in \NN$. By the form of the derivation rules, $r(a, b, n) \in \Fmc_{m - 1}$ has been produced by the application of a rule of form \eqref{math:inference-rule:rigid} or \eqref{math:inference-rule:existential}. Hence, there must be an index $i < m - 1$ such that for some $A'', B'$ and $k \geqslant l$ (again, because \T is a \TELnf-TBox),
				$\dom(f_i) = \{A'' \sqs \exists r . B', A''(a, k)\}$, $\rn(f_i) = \{r(a, b, k), B'(b, k)\}$. Thus $A''(a, k) \in \F_{i - 1}$, and $(\Fmc_0, \dots, \Fmc_{i - 1})$ is a derivation witnessing $(\T, \A) \vdash A''(a, k)$. By the induction hypothesis, there is a derivation $\der_1$ witnessing $(\T' \cup \T_\A, \{C_a(a, l)\}) \vdash A''_k(a, k)$. Moreover, we extract from $(\F_i, \dots, \F_{m - 1})$ a derivation witnessing $(\T, \{B'(b, k)\}) \vdash A'(b, n)$ (note that $n = k$ if $r \in \RNloc$). Since $(\T, \{B'(b, k)\}) \vdash A'(b, n)$, by Lemma~\ref{lm:app:consequences:indexed-concepts}, $(\T', \{B'_k(b, k)\}) \vdash A'_{n}(b, n)$. Let $\der_2$ be a derivation witnessing this. We obtain a derivation witnessing $(\T' \cup \T_\A, C_a(a, l)) \vdash A_{n}(a, n)$ as follows. Since $A'' \sqs \exists r . B' \in \T$, we have $A''_k \sqs \exists r . B'_k \in \T'$, by \eqref{math:ci:abox-to-tbox:existential:null}. Further, since $\exists r . A' \sqs A \in \T$, we have $\exists r . A'_n \sqs A_n\in\T'$, by \eqref{math:ci:abox-to-tbox:return:null}. Thus, start from $\F'_0 = \T' \cup \T_\A \cup \{C_a(a, l)\}$, proceed as in $\der_1$ until you derive $A''_k(a, k)$. Apply a rule of the form \eqref{math:inference-rule:existential} using $A''_k \sqs \exists r . B'_k$ to obtain $r(a,b,k)$ and $B'_k(b, k)$, and proceed as in $\der_2$ until you have $A'_n(b, n)$. 
				If $r\in\RNrig$, apply a rule of the form \eqref{math:inference-rule:rigid} to obtain $r(a,b,n)$. Otherwise, it means that $k=n$, so we already have $r(a,b,n)$. 
				Finally, apply a rule of the form \eqref{math:inference-rule:return} using $\exists r . A'_n \sqs A_n$ to get $A_n(a, n)$.
							
			\end{itemize}		
		\end{itemize}
		
		$(\Leftarrow)$ Again, we show by induction on $m$ that for all $A \in \CN(\T)$, $a\in\NI$ and $n, n' \in \N$, if there exists a derivation of length $m$ witnessing $(\T' \cup \T_\A, \{C_a(a, l)\}) \vdash A_n(a, n')$, then $n'=n$ and $(\T, \A) \vdash A(a, n)$.
		
		The base case is $m = 1$. Then $\F_0 = \T' \cup \T_\A \cup \{C_a(a, l)\}$ and $A_n(a, n') \in \F_1$. Thus, $f_1$ is an application of a rule of the form \eqref{math:inference-rule:shift} using $C_a\sqsubseteq \Next^{n-l}A_n$, so it must be the case that $n'=n$ and $C_a\sqsubseteq \Next^{n-l}A_n\in\T_\A$. By \eqref{math:ci:abox-to-tbox:mark}, it follows that $A(a, n) \in \A$, so $(\T, \A) \vdash A(a, n)$.
		
		Induction step: assume that the property holds for $m - 1$ and let $\F_0 \xto{f_1} \dots \xto{f_m} \F_m$ be a derivation witnessing $(\T' \cup \T_\A, \{C_a(a, l)\}) \vdash A_n(a, n')$. If $A_n(a, n') \in \F_{m - 1}$, the result follows by induction hypothesis. Otherwise, there are five possible cases for the last rule application $f_m$ that produces $A_n(a, n')$:
		\begin{itemize}
		\item $\dom(f_m)=\{C_b(a, l), C_b\sqs \Next^{n-l} A_n \}$ and $f_m$ is the application of a rule of the form \eqref{math:inference-rule:shift}. However, by the form of concept inclusions of $\T'$ and $\T_\A$, it must be the case that $C_b=C_a$ since the derivation rules cannot produce a fact of the form $C_b(a,\ell)$ with $a\in\NI$ and $a\neq b$. Hence we obtain $n'=n$ and $(\T, \A) \vdash A(a, n)$ as in the base case.
			\item $\dom(f_m)=\{A'_{n - k}(a, n' - k), A'_{n - k} \sqs \Next^k A_n \}$ for some  $k \in \{0,\dots, \min(n - l, n'-l)\}$ (since $\T'\cup\T_\A$ is a \TELnf-TBox so that $n'-k\geqslant l$ and $n-k\geqslant l$ by definition of $\T'$), and $f_m$ is the application of a rule of the form \eqref{math:inference-rule:shift}. 
			Since $(\F_0, \dots, \F_{m - 1})$ is a derivation witnessing $(\T' \cup \T_\A, \{C_a(a, l)\}) \vdash A'_{n - k}(a, n' - k)$, by induction hypothesis, $n'-k=n-k$, i.e., $n'=n$, and $(\T, \A) \vdash A'(a, n - k)$. Since $A'_{n - k} \sqs \Next^k A_n \in \T'$ it follows that $A' \sqs \Next^k A \in \T$, by \eqref{math:ci:abox-to-tbox:shift}. We conclude that $(\T, \A) \vdash A(a, n)$ by applying a rule of the form \eqref{math:inference-rule:shift}.
		
			\item $\dom(f_m)=\{A'_n(a, n'), A''_n(a, n'), A'_n \sqcap A''_n \sqs A_n\}$ and $f_m$ is the application of a rule of form  \eqref{math:inference-rule:conjunction}. Since $A'_n(a, n')$ and $A''_n(a, n')$ are in $\Fmc_{m - 1}$, $(\F_0, \dots, \F_{m - 1})$ is a derivation witnessing $(\T' \cup \T_\A, \{C_a(a, l)\}) \vdash A'_n(a, n')$ and $(\T' \cup \T_\A, \{C_a(a, l)\}) \vdash A''_n(a, n')$. 
			By the induction hypothesis, $n'=n$ and $(\T, \A) \vdash A'(a, n)$ and $(\T, \A) \vdash A''(a, n)$. Since $A'_n \sqcap A''_n \sqs A_n \in \T'$, $A' \sqcap A'' \sqs A \in \T$, by \eqref{math:ci:abox-to-tbox:conjunction}. We conclude that $(\T, \A) \vdash A(a, n)$ by applying a rule of the form \eqref{math:inference-rule:conjunction}.
			
			\item $\dom(f_m)=\{r(a, b, n'), A'_n(b, n'), \exists r . A'_n \sqs A_n\}$ and $f_m$ is the application of a rule of form \eqref{math:inference-rule:return}. 
			By the form of the derivation rules, $r(a, b, n') \in \Fmc_{m - 1}$ has been produced by the application of a rule of form \eqref{math:inference-rule:rigid} or \eqref{math:inference-rule:existential}. Hence, there must be an index $i < m - 1$ such that for some $A''_k, B'_k$ and $k' \geqslant l$, $\dom(f_i) = \{A''_k \sqs \exists r . B'_k, A''_k(a, k')\}$, $\rn(f_i) = \{r(a, b, k'), B'_k(b, k')\}$. Thus $A''_k(a, k') \in \F_{i - 1}$, and $(\Fmc_0, \dots, \Fmc_{i - 1})$ is a derivation witnessing $(\T' \cup \T_\A, \{C_a(a, l)\}) \vdash A''_k(a, k')$. By the induction hypothesis, $k'=k$ and there is a derivation $\der_1$ witnessing $(\T, \A) \vdash A''(a, k)$. Moreover, we extract from $(\F_i, \dots, \F_{m - 1})$ a derivation 
			witnessing $(\T' \cup \T_\A, \{B'_k(b, k)\}) \vdash A'_n(b, n')$. Since $B'_k$ does not appear in the left-hand sides of concept inclusions in $\T_\A$, it is not hard to see that, in fact, $(\T', \{B'_k(b, k)\}) \vdash A'_n(b, n')$. Moreover, since $\T'$ is a \TELnf-TBox, $k \leqslant n'$, thus $n' - k \in \N$ (note that $n' = k$ if $r \in \RNloc$). 
			Since $(\T', \{B'_k(b, k)\}) \vdash A'_n(b, n')$, $(\T', \{B'_k(b, 0)\}) \vdash A'_n(b, n'-k)$, so 
			 by point \e{\ref{lm:app:consequences:indexed-concepts:n}} of Lemma~\ref{lm:app:consequences:indexed-concepts}, $n=n'-k+k=n'$, so $(\T', \{B'_k(b, k)\}) \vdash A'_n(b, n)$ and thus, by point \e{\ref{lm:app:consequences:indexed-concepts:iff}} of Lemma~\ref{lm:app:consequences:indexed-concepts}, $(\T, \{B'(b, k)\}) \vdash A'(b, n)$. Let $\der_2$ be a derivation witnessing it. We obtain a derivation witnessing $(\T, \A) \vdash A(a, n)$ as follows. Since $A''_k \sqs \exists r . B'_k \in \T'$, we have $A'' \sqs \exists r . B' \in \T$, by \eqref{math:ci:abox-to-tbox:existential:null}. Further, since $\exists r . A'_n \sqs A_n \in \T'$, we have $\exists r . A' \sqs A\in\T$, by \eqref{math:ci:abox-to-tbox:return:null}. Thus, start from $\F'_0 = \T \cup \A$, proceed as in $\der_1$ until you derive $A''(a, k)$. Apply a rule of the form \eqref{math:inference-rule:existential} to obtain $r(a,b,k)$ and $B'(b, k)$, and proceed as in $\der_2$ until you have $A'(b, n)$. 
			 If $r\in\RNrig$, apply a rule of the form \eqref{math:inference-rule:rigid} to obtain $r(a,b,n)$. Otherwise, it means that $k=n$, so we already have $r(a,b,n)$.
			 Finally, apply a rule of the form \eqref{math:inference-rule:return} to get $A(a, n)$.
			
			\item $\dom(f_m)=\{\role{a, b}(a, b', n'), A'_n(b', n'), \exists \role{a, b} . A'_n \sqs A_n\}$ and $f_m$ is the application of a rule of form \eqref{math:inference-rule:return}. Then $A'_n(b', n') \in \F_{m - 1}$.
			By the form of the derivation rules, $\role{a, b}(a, b', n') \in \Fmc_{m - 1}$ has been produced by the application of a rule of form \eqref{math:inference-rule:rigid} or \eqref{math:inference-rule:existential}. Moreover, there must be an index $i < m - 1$ such that $\dom(f_i) = \{C_a \sqs \exists \role{a, b} . C_b, C_a(a, l)\}$ (note that by the form of the concept inclusions in $\T'\cup\T_\A$, $C_a(a, l)$ is the only fact of the form $C_a(a, \ell)$ that can be derived since $a\in\NI$) and $\rn(f_i) = \{\role{a, b}(a, b', l), C_b(b', l)\}$. We extract a derivation witnessing $(\T' \cup \T_\A, \{C_b(b', l)\}) \vdash A'_n(b', n')$ of length at most $m - 1$. Renaming $b'$ to $b$, we get a derivation of length at most $m - 1$ witnessing $(\T' \cup \T_\A, \{C_b(b, l)\}) \vdash A'_n(b, n')$. By the induction hypothesis, $n'=n$ and there is a derivation \der witnessing $(\T, \A) \vdash A'(b, n)$. We obtain a derivation witnessing $(\T, \A) \vdash A(a, n)$ as follows. Since $C_a \sqs \exists \role{a, b} . C_b \in \T_\A$, we have $r(a, b, \ell) \in \A$, by \eqref{math:ci:abox-to-tbox:existential:abox}. Further, since $\exists \role{a, b} . A'_n \sqs A_n\in\T_\A$, either $r \in \RNrig$ or $\ell = n$, and $\exists r . A' \sqs A \in \T$, by \eqref{math:ci:abox-to-tbox:return:abox}. Thus, start with $\F_0 = \T \cup \A$. Use \der to derive $A'(b, n)$. If $r \in \RNrig$, apply a rule of the form \eqref{math:inference-rule:rigid} to obtain $r(a, b, n)$ (otherwise, it is already in \A). Conclude by applying a rule of the form \eqref{math:inference-rule:return} using $\exists r . A' \sqs A$. \qedhere
		\end{itemize}
	\end{proof}
	
	By Theorem~\ref{thm:conjunctive-grammars:tbox-to-grammar}, one can construct in polynomial time w.r.t.~$|\T'| + |\T_\A|$ (hence in polynomial time w.r.t. $|\T|+|\A|$) a unary conjunctive grammar $G_{\T' \cup \T_\A}=(N,\{c\},R)$ such that there exists $\Nmc_{C_aA}\in N$ such that $c^{n - l} \in L_{G_{\T' \cup \T_\A}} (\Nmc_{C_aA})$ iff $\T' \cup \T_\A \models (C_a\sqsubseteq \Next^{n - l} A_n)$. By Proposition \ref{prop:derivationTEL},  $\T' \cup \T_\A \models (C_a\sqsubseteq \Next^{n - l} A_n)$ iff  $(\T' \cup \T_\A,\{C_a(a,l)\}) \models A_n(a,n)$, and by Lemma \ref{lm:app:consequences:compressed-abox}, the latter is equivalent to $(\T, \A) \models A(a, n)$. Hence, $(\T, \A) \models A(a, n)$ iff $c^{n - l} \in L_{G_{\T' \cup \T_\A}} (\Nmc_{C_aA})$.
\end{proof}
	
% !TEX root =  ../ms.tex

\subsection{Proof of Theorem \ref{thm:consequences:linear:complexity}}

We first recall several definitions and results from the literature.

\paragraph{A bound for Parikh images}

Let $G = (N, \Sigma, R)$ be a context-free grammar. If $\Nmc \to \alpha \in R$, we write $\alpha/_N$ and $\alpha/_\Sigma$ for the words obtained from $\alpha$ by projecting on $N$ and on $\Sigma$, respectively. The \textit{degree} of $G$ is the number $m = -1 + \max_{\Nmc \to \alpha \in R} |\alpha/_N|$, and the \textit{productivity} of $G$ is the number $p = \sum_{\Nmc \to \alpha \in R} |\alpha/_\Sigma|$.

The following result is from A. W. To's PhD thesis ``Model-Checking Infinite-State Systems: Generic and Specific Approaches'' (2010), cited as given by \citet{Esparza:Parikh-Automaton}.

\begin{apptheorem}[To] \label{thm:app:consequences:set-size}
	Let $M$ be a nondeterministic state automaton with $s$ states over an alphabet of $l$ letters. Then the Parikh image of $L(M)$ is a union of $\bigO{s^{l^2 + 3l +3} l^{4l+6}}$ linear sets with at most $l$ period vectors; the maximum component of any oﬀset vector is $\bigO{s^{3l+3} l^{4l+6}}$, and the maximum component of any period vector is at most $s$.
\end{apptheorem}

\begin{apptheorem}[\citet{Esparza:Parikh-Automaton}]\label{thm:app:consequences:parikh-automaton}
	If $G$ is a context-free grammar with $n$ nonterminals, degree $m$ and productivity $p$,
	then one can construct, in polynomial space,\footnote{This is not stated explicitly in \citep{Esparza:Parikh-Automaton}, but follows from their construction.} a nondeterministic finite state automaton $M$ with $\binom{n + nm + 1}{n} \cdot p$ states such that $L(G)$ and $L(M)$ have the same Parikh image.
\end{apptheorem}

\begin{appcorollary}\label{cor:app:consequences:bounds}
	For every context-free grammar $G = (N, \Sigma, R, \Smc)$ of degree $n$ and productivity $p$, with $|N| = n$ and $|\Sigma| = t$, the bounds of Theorem \ref{thm:app:consequences:set-size} hold with $s =\binom{n + nm + 1}{n} \cdot p$ and $l = t$.
\end{appcorollary}

\paragraph{\e{Datalog}$_\textup{1S}$}

We informally describe \dlos programs. For a formal definition, the reader is referred to \citet{Chomicki:TDD}.

Fix a dedicated \e{temporal variable} $t$. A \dlos program $\Pi$ is a finite set of rules of the form
\begin{align}\label{math:dlos:rule}
	B(\ovl{x}, t) \impd A_1(\ovl{x}_1, t + i_1), \dots, A_k(\ovl{x}_k, t + i_k)
\end{align}
where $B, A_1, \dots, A_k$ are predicate symbols, $\ovl{x}_i$ are tuples of variables
and $i_1, \dots, i_k \in \Z$ are given in unary.\footnote{(\citet{Chomicki:TDD} also allow atoms of the form $A(\ovl{x}, k)$, for $k \in \Z$, but we will not need them).} 
Since we work with DL ABoxes, we assume that the arities of $B$ and $A_i$ are either 2 or 3.

Given a program $\Pi$ and an ABox \A, their \textit{canonical model} (the least Herbrand model, in the terminology of \citet{Chomicki:TDD}) is obtained by exhaustive application of the rules of $\Pi$ to \A. We write $(\Pi, \A) \mdl A(a, n)$ if $A(a, n)$ holds in the canonical model of $\Pi$ and \A. We write $\Pi \mdl A \sqs \Next^k B$ if $(\Pi, \A) \mdl A(a, n)$ implies $(\Pi, \A) \mdl B(a, n + k)$, for any ABox $\A$.

\paragraph{\dlnd}

We define \dlnd, introduced by \citet{Artale:Linear-Temporal-Datalog}, as a syntactic variant of \dlos. A rule of the form \eqref{math:dlos:rule} is written in \dlnd as:
\begin{align}\label{math:dlnd:rule}
	B(\ovl{x}) \impd \Next^{i_1}A_1(\ovl{x}_1), \dots, \Next^{i_k}A_k(\ovl{x}_k)
\end{align}
Intuitively, the explicit temporal variable $t$ of \dlos is hidden behind the temporal operator \Next of \dlnd. A \dlnd program $\Phi$ is a finite set of rules. The semantics are borrowed from \dlos.

Additionally to \Next, \dlnd allows for temporal operators \Df/\Dp (eventually in the future/past) in right-hand sides of the rules, but such rules can be equivalently written without this operator: a rule of the form $B(\ovl{x}) \impd \Df A(\ovl{x}), \p$ is equivalent to
\begin{align}\label{math:dlnd:diamond}
	& B(\ovl{x}) \impd  A'(\ovl{x}), \p
	&&A'(\ovl{x}) \impd A(\ovl{x})
	&&A'(\ovl{x}) \impd \Next A'(\ovl{x})
\end{align}
where $A'$ is a fresh symbol, and symmetrically for \Dp.

Both in \dlos and \dlnd, the predicate symbols that appear in the left-hand sides of the rules are called intensional predicates. A program is \emph{linear} if each of its rules contains at most one intensional predicate in the right-hand side. Note that a linear \dlnd program may not be expressible as a linear \dlos program, as the rewriting of the form \eqref{math:dlnd:diamond} introduces a new intensional predicate in the right-hand side of the first rule.

Let $l, m$ be the least and the greatest timestamps of \A. \citet{Artale:Linear-Temporal-Datalog} use the notion of derivations similar to ours to show that for a linear \dlnd program $\Phi$, one can limit attention to timestamps within the range $[l - \poly{|\Phi|, |\A|}, m + \poly{|\Phi|, |\A|}]$. In turn, when the range of timestamps is bounded, \dlnd programs can be rewritten to standard Datalog by a straightforward simulation of the order on the timestamps with binary relations. Thus, \citet{Artale:Linear-Temporal-Datalog} conclude that query answering with linear \dlnd programs is \NLogSpace-complete, for data complexity, matching the case of standard linear Datalog. However, the argument above works also for combined complexity. Thus, we obtain the following extended result that we attribute to \citet{Artale:Linear-Temporal-Datalog} (complexity of query answering with linear Datalog programs is from \citet{Gottlob-Papadimitriou:Datalog-Complexity}).

\begin{apptheorem}[\citet{Artale:Linear-Temporal-Datalog}] \label{thm:app:vlad}%	
	Query answering with linear \dlnd programs is \NLogSpace-complete and \PSpace-complete, respectively, for data and combined complexity.
\end{apptheorem}

We are now ready to prove Theorem \ref{thm:consequences:linear:complexity}.

\ThmLinearUP*

\begin{proof}
	\e{(i)} Let \T be a \TELnc-TBox. The proof that \T is ultimately periodic is in the main text. For the bound on $\|\T\|$, fix any $A, B \in \CN(\T)$ and let $\Gamma_\T = (N, \{c, d\}, R, \Nmc_{AB})$ be as in Theorem \ref{thm:linear:tbox-to-grammar}. Then by Definition \ref{def:linear:GammaT}, the parameters in Corollary \ref{cor:app:consequences:bounds} are as follows: $n, m, p \leqslant \poly{|\T|}$, and $t = 2$. Thus, $s \leqslant 2^\poly{|\T|}$, and $l = 2$, rendering $p(L(G))$ to be a union of at most $2^\poly{|\T|}$ linear sets with at most 2 periods, the maximum entry of any oﬀset and any period is $2^\poly{|\T|}$. It follows that $\|\T\| \leqslant 2^\poly{|\T|}$.
	
	\e{(ii), (iii)} The lower bound for data complexity is from linear~\EL~\citep{DBLP:conf/rr/DimartinoCPW16}. For the upper bound, we recall from \citet{Basulto-et-al:TEL} that every ultimately periodic (w.r.t. quasimodels) \TELn-TBox \T can be translated, in polynomial time, to a \dlos program $\Pi_\T$, such that for any $\A$ and $A(a, n)$ we have $(\T, \A) \mdl A(a, n)$ iff $(\Pi_\T, \A) \mdl A(a, n)$. To understand the translation, the reader will need the definition of quasimodels given in Section \ref{app:prelim:quasimodels}; \citet{Basulto-et-al:TEL} also provide intuitive explanations that we omit here. The program $\Pi_\T^1$ consists of the following rules.
	\begin{align}
		&r(x, y, t) \impd r(x, y, t \pm 1) &&\text{for } r \in \RNrig(\T)
		\label{math:dlos:rigid}\\
		&B(x, t) \impd A(x, t), A'(x, t) &&\text{for } A \sqcap A' \sqs B \in \T
		\label{math:dlos:conjunction}\\
		&B(x, t) \impd r(x, y, t), A(y, t) &&\text{for } \exists r . A \sqs B \in \T
		\label{math:dlos:up:abox}
	\end{align}
	Further, let $\Qmf = \{\pi_d \mid d \in \CN(\T)\}$ be the canonical quasimodel of $(\T, \emptyset)$. For each trace $\pi_A$, with integers $m_P$ , $p_P$,
	$m_F$ , $p_F$ given by ultimate periodicity of \T, the program $\Pi_\T^2$ contains the following rules with a fresh predicate~$F_A$: 
	\begin{align}
		&B(x, t) \impd A(x, t - i) \quad \text{for } 0 \leqslant i < m_F,\ B \in \pi_A(i)
		\label{math:dlos:fill:offset}\\
		&F_A(x, t) \impd A(x, t - m_F)
		\label{math:dlos:mark:offset}\\
		&F_A(x, t) \impd F_A(x, t - p_F)
		\label{math:dlos:mark:period}\\
		&B(x, t) \impd F_A(x, t - i) \ \ \text{for } 0 \leqslant i < p_F , B \in \pi_A(m_F + i)
		\label{math:dlos:fill:period}
	\end{align}
	and symmetric rules with $m_P$, $p_P$ and fresh predicate $P_A$. Intuitively, the
	rules of the form~\eqref{math:dlos:fill:offset} replicate the initial part of $\pi_A$, from 0
	to $m_F$. The rules of the forms \eqref{math:dlos:mark:offset} and \eqref{math:dlos:mark:period} mark the start of each period with $F_A$ while the rules of the form \eqref{math:dlos:fill:period}
	replicate the period of $\pi_B$, starting from each marker $F_A$.

	Finally, $\Pi_\T = \Pi_\T^1 \cup \Pi_\T^2$.
	
	\begin{appproposition}[\citet{Basulto-et-al:TEL}] \label{prop:app:dlos}
		For any \TELn-TBox \T, ABox \A and TAQ $A(a, n)$, $(\T, \A) \mdl A(a, n)$ iff $(\Pi_\T, \A) \mdl A(a, n)$.
	\end{appproposition}
	
	Note that rules of the forms \eqref{math:dlos:conjunction} and \eqref{math:dlos:up:abox} may be not linear.
	Now, suppose \T is a \TELnc-TBox, which excludes rules of the form \eqref{math:dlos:conjunction}.  
	Our goal is to rewrite $\Pi_\T$ into a linear \dlnd program $\Phi_\T$, using the additional power of operators \Df/\Dp. We will also use our definition of ultimate periodicity (w.r.t. concept inclusions) and Theorem~\ref{thm:app:consequences:set-size} to have a program of reasonable size.
	
	We define $\Phi_\T^1$ to contain, for each rule of the form \eqref{math:dlos:up:abox}, the following rules:
	\begin{align}
		&B(x) \impd r(x, y), A(y)
		\label{math:dlnd:up:abox:local}
		\tag{\ref{math:dlos:up:abox}$^\mathit{loc}$}\\
		&B(x) \impd \D r(x, y), A(y)
		&&\text{if } r \in \RNrig 
		\label{math:dlnd:up:abox:rigid:future}
		\tag{\ref{math:dlos:up:abox}$^\mathit{rig}_F$}\\
		&B(x) \impd \Dp r(x, y), A(y)
		&&\text{if } r \in \RNrig 
		\label{math:dlnd:up:abox:rigid:past}
		\tag{\ref{math:dlos:up:abox}$^\mathit{rig}_P$}
	\end{align}  
	Further, for every $A, B  \in \CN(\T)$, we take a representation
	\begin{equation}\label{math:app:shift-set-for-dlnd}
		\begin{aligned}
			\{n \in \Z \mid \T \mdl A \sqs \Next^n B\} = \Lmc_1 \cup \dots \cup \Lmc_m\\
			\Lmc_i = \{b^i + k_1 p^{i}_{1} + \dots k_l p^{i}_{l} \mid k_1, \dots, k_l \in \N\}.
		\end{aligned}
	\end{equation}
	Then, the program $\Phi_\T^2$ contains the following rules, $i \in \{1, \dots, m\}$.
	\begin{align}
		&F_i^{AB}(x) \impd \Next^{-b^i} A(x)
		\label{math:dlnd:offset}\\
		&F_i^{AB}(x) \impd \Next^{-p^i_j} F^{AB}_i(x) &&\text{for } 1 \leqslant j \leqslant l
		\label{math:dlnd:periods}\\
		&B(x) \impd F_i^{AB}(x)
		\label{math:dlnd:imply}
	\end{align}
	Finally, $\Phi_\T = \Phi_\T^1 \cup \Phi_\T^2$.
	
	It is readily seen that $\Phi_\T$ is a linear program, and that $|\Phi_\T| \leqslant \poly{\|\T\|}$. From point \e{(i)} we get the bound $|\Phi_\T| \leqslant 2^\poly{|\T|}$. Moreover, if $\RN(\T) \sbs \RNrig$, $\Phi_\T$ can be effectively constructed from $\T$. Indeed, in this case, one obtains the grammar $\Gamma_\T$ of Theorem \ref{thm:linear:tbox-to-grammar} in polynomial time. Then, using the automaton of Theorem \ref{thm:app:consequences:parikh-automaton}, one obtains representations \eqref{math:app:shift-set-for-dlnd} that respect the bounds of Corollary \ref{cor:app:consequences:bounds}. 
	
	Points \e{(ii)} and \e{(iii)} then follow from the fact that $(\Pi_\T, \A) \mdl A(a, n)$ iff $(\Phi_\T, \A) \mdl A(a, n)$. To see the latter, recall that when \T is a \TELnc-TBox, $\Pi^1_\T$ does not contain rules of the form \eqref{math:dlos:conjunction}. It is immediate that $\Pi^1_\T$ is equivalent to $\Phi^1_\T$. The fact that $\Pi_\T^2$ is equivalent to $\Phi_\T^2$ follows from the lemma below, given the facts that rules of the forms \eqref{math:dlos:fill:offset}-\eqref{math:dlos:fill:period} and \eqref{math:dlnd:offset}-\eqref{math:dlnd:imply} allow to derive $D(b, \ell')$ from $C(a, \ell)$ only if $a = b$.
	\begin{applemma}\label{lm:app:pi-t-s}
		For any \TELnc-TBox \T and $A, B \in \CN(\T)$, $\Pi^2_\T \mdl A \sqs \Next^n B$ iff $\Phi^2_\T \mdl A \sqs \Next^n B$.
	\end{applemma}
	\begin{proof}
		By the form of the rules \eqref{math:dlos:fill:offset}-\eqref{math:dlos:fill:period}, $\Pi_\T^2 \mdl A \sqs \Next^n B$ iff $B \in \pi_A(n)$. By Lemma \ref{lm:app:qm}, $B \in \pi_A(n)$ iff $\T \mdl A \sqs \Next^n B$. Finally, by the form of the rules \eqref{math:dlnd:offset}-\eqref{math:dlnd:imply}, $\T \mdl A \sqs \Next^n B$ iff $\Phi_\T^2 \mdl  A \sqs \Next^n B$.
	\end{proof}\phantom{\qedhere}
\end{proof}

\end{document}